%% file: 0_nrdc_monodic.tex
\documentclass[preprint,10pt]{elsarticle}

\usepackage{soul}
\usepackage{url}
\usepackage[hidelinks]{hyperref}
\usepackage[utf8]{inputenc}
\usepackage[small]{caption}
\usepackage{graphicx}
\usepackage{amsmath}
\usepackage{amsthm}
\usepackage{booktabs}
\usepackage{algorithm}
\usepackage{algorithmic}
\usepackage{mathtools}
\usepackage{paralist}
\usepackage[T1]{fontenc}
\usepackage{amsfonts}
\usepackage{amssymb}
\usepackage{stmaryrd}
\usepackage{wasysym}
\usepackage{bbm}
\usepackage{pifont}
\usepackage{array}
\usepackage{multirow}
\usepackage{tikz}
\usetikzlibrary{calc,patterns}
\usepackage{xcolor}
\usepackage[inline]{enumitem}
\usepackage{xspace}
\usepackage{phonetic}
\usepackage{todonotes}
\usepackage{thmtools} 

\usepackage{latexsym}

\newtheorem{theorem}{Theorem}
\newtheorem{example}[theorem]{Example}

\newtheorem{lemma}[theorem]{Lemma}
\newtheorem{corollary}[theorem]{Corollary}
\newtheorem{claim}{Claim}
\newtheorem{remark}[theorem]{Remark}

\allowdisplaybreaks

\newif\ifappendix
\appendixtrue

\newif\ifshort
\shortfalse

\usepackage{misc}
\input{macros}

\newcommand{\revised}[2]{\textcolor{purple}{\marginpar{\parbox{12mm}{\scriptsize\textcolor{purple}{\raggedright #1}}}#2}}
\renewcommand{\revised}[2]{#2}

\journal{\ldots}

\makeatletter
\def\ps@pprintTitle{%
  \let\@oddhead\@empty
  \let\@evenhead\@empty
  \def\@oddfoot{\reset@font\hfil\thepage\hfil}
  \let\@evenfoot\@oddfoot
}
\makeatother

\begin{document}

\begin{frontmatter}



\title{Decidability in First-Order Modal Logic\\
	with Non-Rigid Constants and Definite Descriptions} 


\author[label1]{Alessandro Artale} 
\affiliation[label1]{organization={
Free University of Bozen-Bolzano},
            country={Italy}
            }         
           
\author[label2]{Christopher Hampson} 
\affiliation[label2]{organization={
King's College London},
            country={UK}
            }

\author[label3]{Roman Kontchakov} 
\affiliation[label3]{organization={
Birkbeck, University of London},
            country={UK}
            }
            
\author[label1]{Andrea Mazzullo} 
            
\author[label4]{Frank Wolter} 
\affiliation[label4]{organization={
University of Liverpool},
            country={UK}
            }

\begin{abstract}
While modal extensions of decidable fragments of first-order logic
are usually
undecidable, their monodic counterparts, in which formulas in the scope of modal operators have at most one free variable, are typically decidable.
This only holds, however, under the provision that non-rigid constants, definite descriptions and non-trivial counting are not admitted.
Indeed, several
monodic fragments having at least one of these features are known to be undecidable. 
We investigate these features systematically and show that fundamental monodic fragments such as the two-variable fragment with counting and the guarded fragment of standard first-order modal logics $\mathbf{K}_{n}$ and $\mathbf{S5}_{n}$ are decidable. 
Tight complexity bounds are established as well. Under the expanding-domain semantics, we show decidability of the basic modal logic extended with the  transitive closure  operator on
frames without infinite ascending chains; this logic, however, is Ackermann-hard.
\end{abstract}




\end{frontmatter}




\input{1_introduction}

\input{2_preliminaries}

\input{4.1_quasimodel}

\input{4.2_constant}

\input{4.3_expanding}

\input{4.4_pspace}

\input{5_temporal}

\input{6_conclusion}


  \bibliographystyle{elsarticle-num} 
  \bibliography{bibliography}





\end{document}

%% file: macros.tex
\newcommand{\Dom}{\Delta}

\newcommand{\NRDC}{NRDC}

\newcommand{\true}{\textbf{true}\xspace}%
\newcommand{\QKsat}[1][\varphi,N]{\textsc{QKsat}_{#1}}%

\newcommand{\avec}[1]{\boldsymbol{#1}}
\newcommand{\extN}{\mathbb{N}_\infty}

\newcommand{\sub}[2][x]{\mathit{sub}_{#1}(#2)}%
\newcommand{\subs}[2][x]{\mathit{sub}^*_{#1}(#2)}%

\newcommand{\Pred}{\ensuremath{\mathsf{Pred}}\xspace}
\newcommand{\Ind}{\ensuremath{\mathsf{Con}}\xspace}
\newcommand{\Var}{\ensuremath{\mathsf{Var}}\xspace}

\newcommand{\defdes}{\ensuremath{\smash{\iota}}\xspace}

\newcommand{\Next}{{\ensuremath{\raisebox{0.25ex}{\text{\scriptsize{$\bigcirc$}}}}}}

\newcommand{\D}{\Diamond}

\DeclareRobustCommand*{\monodic}{\ooalign{\kern-0.15ex$\Box$ \cr \kern0.45ex \raisebox{0.4ex}{\scalebox{0.5}{$1$}}}\rule{0pt}{1.5ex} \kern-0.85ex}

\newcommand{\rpath}{\mathsf{rp}}

\newcommand{\md}{\ensuremath{\mathit{d}}}

\newcommand{\length}{\textit{length}}

\newcommand{\singl}{\ensuremath{\xi_{\psi}}}
\newcommand{\norml}{\ensuremath{\eta_{\psi}}}

\newcommand{\K}{\ensuremath{\mathbf{K}}}

\newcommand{\Kfn}{\ensuremath{\K\!\boldsymbol{f}\!_{\ast n}}}

\newcommand{\FO}{\ensuremath{\mathsf{FO}}}
\newcommand{\FOO}{\ensuremath{\mathsf{FO^1}}}

\newcommand{\CT}{\ensuremath{\mathsf{C^2}}}

\newcommand{\Sfive}{\ensuremath{\mathbf{S5}}}

\newcommand{\Sfiven}{\ensuremath{\mathbf{S5}_{n}}}

\newcommand{\QML}{\ensuremath{\smash{\mathsf{Q}^{=}\mathsf{ML}_{\defdes}}}\xspace} 
\newcommand{\QMLc}{\ensuremath{\smash{\mathsf{Q}^{=}\mathsf{ML}_c}}\xspace}
\newcommand{\QMLplusc}{\ensuremath{\smash{\mathsf{Q}^{=}\mathsf{ML}_c}}\xspace}

\newcommand{\QmonML}{\ensuremath{\smash{\mathsf{Q}^{=}_{\monodic}\mathsf{ML}_{\defdes}}}\xspace} 
\newcommand{\QmonMLc}{\ensuremath{\smash{\mathsf{Q}^{=}_{\monodic}\mathsf{ML}_c}}\xspace} 
\newcommand{\QmonMLplusc}{\ensuremath{\smash{\mathsf{Q}^{=}_{\monodic}\mathsf{ML}_c}}\xspace}

\newcommand{\QGuardmonMLb}{\ensuremath{\smash{\mathsf{GF}^{=}_{\monodic}\mathsf{ML}_{\defdes}}}\xspace} 
\newcommand{\QGuardmonMLc}{\ensuremath{\smash{\mathsf{GF}^{=}_{\monodic}\mathsf{ML}_{c}}}\xspace}
\newcommand{\QGuardmonMLplusc}{\ensuremath{\smash{\mathsf{GF}^{=}_{\monodic}\mathsf{ML}_{c}}}\xspace}

\newcommand{\QoneML}{\ensuremath{\smash{\mathsf{Q}^{1=}\mathsf{ML}_{\defdes}}}\xspace} 
\newcommand{\QoneMLc}{\ensuremath{\smash{\mathsf{Q}^{1=}\mathsf{ML}_{c}}}\xspace} 

\newcommand{\CtwomonML}{\ensuremath{\smash{\mathsf{C}^{2}_{\monodic}\mathsf{ML}_{\defdes}}}\xspace} 
\newcommand{\CtwomonMLc}{\ensuremath{\smash{\mathsf{C}^{2}_{\monodic}\mathsf{ML}_{c}}}\xspace} 
\newcommand{\CtwomonMLplusc}{\ensuremath{\smash{\mathsf{C}^{2}_{\monodic}\mathsf{ML}_{c}}}\xspace}

\newcommand{\MLDiff}{\ensuremath{\smash{\mathsf{Q}^{1\neq}\mathsf{ML}}}\xspace}	

\newcommand{\LTL}{\mathbf{LTL}}
\newcommand{\LTLf}{\mathbf{LTL}\!\boldsymbol{f}}
\newcommand{\LTLd}{\mathbf{LTL}^{\!\Diamond}}
\newcommand{\LTLfd}{\mathbf{LTL}\!\boldsymbol{f}^{\Diamond}}

\newcommand{\QGuardmonLTLb}{\ensuremath{\smash{\mathsf{GF}^{=}_{\monodic}\mathsf{LTL}_{\defdes}}}\xspace}
\newcommand{\QGuardmonLTLdb}{\ensuremath{\smash{\mathsf{GF}^{=}_{\monodic}\mathsf{LTL}^{\!\Diamond}_{\defdes}}}\xspace}

\newcommand{\QoneLTL}{\mathsf{Q}^{1=}\mathsf{LTL}_{\defdes}}
\newcommand{\QoneLTLd}{\ensuremath{\smash{\mathsf{Q}^{1=}\mathsf{LTL}^{\!\smash{\Diamond}}_{\defdes}}}\xspace}

\newcommand{\CtwomonLTL}{\ensuremath{\smash{\mathsf{C}^{2}_{\monodic}\mathsf{LTL}_{\defdes}}}\xspace}

\newcommand{\CtwomonLTLd}{\ensuremath{\smash{\mathsf{C}^{2}_{\monodic}\mathsf{LTL}^{\!\smash{\Diamond}}_{\defdes}}}\xspace}

\newcommand{\quasimod}{\ensuremath{\Qmf}\xspace}

\newcommand{\runs}{\ensuremath{\Rmf}\xspace}
\newcommand{\iruns}{\smash{\hat{\Rmf}}}
\newcommand{\runsp}{\Rmf^{\mathfrak{p}}}

\newcommand{\funcand}{\ensuremath{\boldsymbol{q}}\xspace}

\newcommand{\contp}{t}

\newcommand{\assign}{\ensuremath{\mathfrak{a}}\xspace}

\newcommand{\tvalue}[3]{#1^{\Mmf({#2}), {#3}}}

\newcommand{\Tree}[1]{\mathbf{Tree}_n^{#1}}
\newcommand{\Kn}{\K_n}
\newcommand{\SfTree}[1]{\mathbf{S5Tree}_n^{#1}}
\newcommand{\gen}[1]{\tilde{#1}}	
\newcommand{\rep}[1]{\tau_{#1}}

\newcommand{\prset}{\mathfrak{p}}
\newcommand{\prfun}[2][]{\prset#1_{#2}}	
\newcommand{\pranc}[1]{\smash{\hat{\mathfrak{h}}}_{#1}}
\newcommand{\iprset}{\smash{\hat{\mathfrak{p}}}}
\newcommand{\iprfun}[2][]{\iprset#1_{#2}}	
\newcommand{\inftypes}[1][]{\mathfrak{inf}#1}

\newcommand{\coNExp}{co\textsc{NExp}}
\newcommand{\twoExp}{\textsc{2Exp}}
\newcommand{\cdom}{$=$}
\newcommand{\expd}{$\subseteq$}
\newcommand{\whr}[1]{\tiny [#1]}
\newcommand{\compresult}[2]{\renewcommand{\tabcolsep}{0pt}\begin{tabular}{c}#1\\[-4pt]\whr{#2}\end{tabular}}
\newcommand{\thm}{Th.}
\newcommand{\cor}{Cor.}

%% file: 1_introduction.tex

\section{Introduction}

First-order modal logics are notorious for their bad computational behaviour, with modal extensions of decidable fragments of first-order logic ($\FO$) 
being typically undecidable.
A well-known example is the monadic fragment of $\FO$ which admits only unary predicates and no function symbols. While the monadic fragment of $\FO$ is known to be decidable~\cite{BoeEtAl97}, the monadic fragment of almost any standard first-order modal logic is undecidable. This applies, in particular, to the basic modal logic $\Kn$ of all Kripke frames  with $n$ accessibility relations and the basic epistemic modal logic $\Sfiven$ determined by the class of Kripke frames with $n$ equivalence relations, for all $n\geq 1$~\cite{kripke1962undecidability,DBLP:journals/sLogica/RybakovS25}. In fact, over the past fifty years, numerous powerful undecidability results have been obtained, reflecting the fact that the two-dimensional shape of combinations of Kripke frames with first-order domains often provides enough structure to encode Turing machines or other models of computation with undecidable halting problem~\cite{GabEtAl03,BraGhi07,DBLP:journals/sLogica/RybakovS19,DBLP:journals/igpl/Rybakov24}. While new families of decidable fragments of non-modal $\FO$ have been discovered on a regular basis (for instance, variants of the two-variable, the guarded, and the fluted fragments~\cite{pratt2023fragments}), progress has been slower in first-order modal logic.

In the early 2000s the \emph{monodic} fragment of first-order modal and temporal logic was introduced, in which modal operators are only applied to formulas with at most one free variable. This fragment turned out to behave much better in the sense that very often the monodic fragment of a modal extension of a decidable fragment of $\FO$ is decidable again. Prominent examples include the monadic monodic, two-variable monodic, and guarded monodic fragments of first-order modal logics $\Kn$ and $\Sfiven$, even if extended with modal operators for transitive closure~\cite{DBLP:journals/jsyml/WolterZ01}. Other examples are various monodic fragments of first-order
temporal logics, including linear temporal logic \textsf{LTL}~\cite{HodEtAl03,DBLP:journals/apal/Hodkinson06}.

These positive results hold independently of
whether expanding or constant domains are assumed, but they come with a crucial provision: the language should not admit \emph{non-rigid designators}, in the form of constants or definite descriptions denoting different
individuals
in different worlds, or non-trivial \emph{counting},
expressing, for instance, that a unary predicate has at most $n$ elements.
In the sequel, we say that a language has \emph{\NRDC{} features} if it contains non-rigid designators and/or non-trivial counting.
Extended with \NRDC{} features, even the one-variable fragment of first-order \textsf{LTL} with constant or expanding domains becomes undecidable~\cite{DBLP:journals/sLogica/DegtyarevFL02,HamKur15}, and so does 
the one-variable fragment of first-order $\Kn$ with the universal modality on constant-domain models~\cite{DBLP:conf/aiml/HampsonK12}.


\NRDC{} features are, however, often the main reason for using first-order modal logic
and thus are fundamental for many applications.
In philosophy, they
have been instrumental in shaping the analysis of
\emph{referential opacity} in modal contexts, related to the failure of the substitutivity principle for equality in the presence of modal operators, witnessed, for instance, by `the number of planets is necessarily greater than $7$' as opposed
to `$8$ is necessarily greater than~$7$'~\cite{Lin71, Lap12}.
Furthermore,
the development of an adequate syntactic and semantic treatment of non-rigid definite descriptions, such as
`the Evening Star',
in contrast to rigid names, such as
`Venus',
has
provided a formal background to
the debate between \emph{descriptivist} theories, on the one hand, and
\emph{direct reference} theories of proper names, on the other~\cite{Mar21}.
%
Recently, definite descriptions have received
renewed attention in free logic~\cite{Kur22, Ind23, Pet24}, description logic~\cite{ArtEtAl20b, NeuEtAl20, ArtEtAl21a}, as well as modal and hybrid logic literature~\cite{Ind20, Orl21, WalZaw23, Wal24}.
Under the \emph{epistemic} reading of modalities, non-rigid terms have been deployed
as a device for denoting distinct individuals in alternative scenarios considered conceivable by an agent~\cite{Hin62, LomCol96, BelLom09, DBLP:journals/sLogica/Wolter00}.
Similar applications have also been found  in the context of quantified \emph{temporal} logics, with formalisms involving
individual symbols
that are allowed to change their denotation over time~\cite{KroMer08, IndZaw23, GeaEtAl22}.

\smallskip

As discussed, despite these applications, mostly negative
results have been obtained so far regarding the computational behaviour of first-order modal logics with \NRDC{} features.
\emph{The aim of this paper is to initiate a systematic investigation
	of  decidability and complexity of monodic fragments of first-order modal logics extended with \NRDC{} features. In particular, we show that such extensions of many fundamental monodic fragments remain decidable.} 

\smallskip

\begin{table}[t]
\noindent\centering{\tabcolsep=3.1pt\begin{tabular}{cccccccc}\toprule 
\multirow{2}{*}{frames $\Cmc$} & \multirow{2}{*}{dom.} & \multicolumn{3}{c}{$\Cmc$-validity} & \multicolumn{3}{c}{global $\Cmc$-consequence}\\\cmidrule(lr){3-5}\cmidrule(lr){6-8}
& & $\QoneML$ & $\CtwomonML$ & $\QGuardmonMLb$ & $\QoneML$ & $\CtwomonML$ & $\QGuardmonMLb$   \\\midrule
$\Sfive$ & \cdom &  \compresult{\coNExp}{\thm~\ref{thm:complexity:onevar}}  & \compresult{\coNExp}{\thm~\ref{thm:c2:conexptime}} & \compresult{\twoExp}{\thm~\ref{thm:guarded:complexity}} &\compresult{\coNExp}{Rem.~\ref{rem:s5:global}}  & \compresult{\coNExp}{Rem.~\ref{rem:s5:global}} & \compresult{\twoExp}{Rem.~\ref{rem:s5:global}} \\
$\Sfiven$, \footnotesize $n \geq 2$ & \cdom &   \compresult{\coNExp}{\thm~\ref{thm:complexity:onevar}}  & \compresult{\coNExp}{\thm~\ref{thm:c2:conexptime}} & \compresult{\twoExp}{\thm~\ref{thm:guarded:complexity}} & \multicolumn{3}{c}{undecidable \whr{\thm~\ref{thm:undecglobal}}} \\
\multirow{2}{*}{$\Kn$} & \cdom &  \compresult{\coNExp}{\thm~\ref{thm:complexity:onevar}} & \compresult{\coNExp}{\thm~\ref{thm:c2:conexptime}}  & \compresult{\twoExp}{\thm~\ref{thm:guarded:complexity}} & \multicolumn{3}{c}{undecidable \whr{\thm~\ref{thm:undecglobal}}}\\
& \expd & \compresult{\PSpace}{\thm~\ref{thm:kn:exp}} & \compresult{\coNExp}{\cor~\ref{thm:expanding}} & \compresult{\twoExp}{\cor~\ref{thm:expanding}} & \multicolumn{3}{c}{?}\\\midrule
\multirow{2}{*}{\hspace*{-0.25em}$\K_{\ast n}$, $\LTL^{\smash{(\Diamond)}}$\hspace*{-1em}} & \cdom & \multicolumn{6}{c}{$\Pi_1^1$~\whr{L.~\ref{lem:restemp} + \thm~\ref{thm:temp}~(1)}} \\
& \expd & \multicolumn{6}{c}{undecidable~\whr{L.~\ref{lem:restemp} + \thm~\ref{thm:temp}~(2)}} \\[6pt]
\multirow{2}{*}{\hspace*{-0.25em}$\Kfn$, $\LTLf^{\smash{(\Diamond)}}$\hspace*{-1em}} & \cdom & \multicolumn{6}{c}{undecidable~\whr{L.~\ref{lem:restemp} + \thm~\ref{thm:temp}~(1)}} \\ 
& \expd & \multicolumn{6}{c}{decidable, Ackermann-hard~\whr{L.~\ref{lem:restemp} + \thm~\ref{thm:exp}, \ref{thm:temp}~(2)}}\\
\bottomrule
\end{tabular}}%
\caption{$\Cmc$-validity and global $\Cmc$-consequence in fragments of $\QmonML$ (complexity bounds are tight and, unless otherwise stated, $n \geq 1$).}\label{table:complexity}
\end{table}
We now discuss our results and techniques in detail. We consider standard Kripke semantics for first-order modal logic with expanding or constant domains. 
Our most general language includes equality, non-rigid and possibly non-denoting constants, and definite descriptions. We observe that allowing definite descriptions with free variables makes even the non-modal two-variable and guarded fragments of \FO{} undecidable, and therefore restrict our attention to closed definite descriptions. Since these, as well as non-denoting constants, do not increase the complexity of reasoning, we establish our main technical results for the language with non-rigid constants and equality.
We focus on the standard modal logics $\Kn$ and $\Sfiven$ and the extension 
$\K_{\ast n}$ of $\Kn$ with a modal operator interpreted by the transitive closure of the remaining $n$ accessibility relations. Of particular interest for mapping out the border between decidability and undecidability is the modal logic $\Kfn$, obtained from $\K_{\ast n}$ by 
considering only Kripke frames without infinite ascending chains. 

We investigate decidability and complexity of monodic fragments of these first-order modal logics restricted to decidable fragments of \textsf{FO}.
\revised{R2.1}{We start by developing a modular framework for studying monodic fragments with NRDC features that builds on and adapts the quasimodels introduced previously for monodic fragments without such features; see, e.g.,~\cite{GabEtAl03}.} We then apply our new framework and investigate the following three fragments, a minimal and two maximal ones.
\begin{itemize}
	\item The one-variable fragment of first-order modal logic, $\QoneML$, is underpinned by the one-variable fragment $\textsf{FO}^1$ of \textsf{FO} with equality and constants. Validity in $\textsf{FO}^1$ is coNP-complete, so it behaves exactly as propositional logic. 
    \item The monodic two-variable fragment with counting, $\CtwomonML$, is underpinned by the two-variable fragment $\textsf{C}^2$ of \textsf{FO} with counting quantifiers, which  
    has received considerable attention as a decidable fragment of \textsf{FO} supporting quantitative reasoning~\cite{DBLP:conf/lics/GradelOR97,DBLP:journals/jolli/Pratt-Hartmann05,DBLP:journals/iandc/Pratt-Hartmann09}. Validity in $\textsf{C}^2$ is co\NExpTime-complete, \revised{R2.1}{but it does not have the finite model property (fmp)}.
    \item The monodic guarded fragment, $\QGuardmonMLb$, is underpinned by the guarded fragment \textsf{GF} of \textsf{FO}, which generalises standard modal and description logics~\cite{ANvB98,DBLP:journals/jsyml/Gradel99}. Validity in \textsf{GF} is 2\ExpTime-complete, but it is incompatible with $\textsf{C}^{2}$ in the sense that \textsf{GF} with counting quantifiers is undecidable. 
\end{itemize} 
Our main results are shown in Table~\ref{table:complexity}.
In this table, the frame classes
are listed in the first column, with $\LTLf$ denoting the class of \emph{all finite strict linear orders} and $\LTL$ denoting $\{(\mathbb{N},<)\}$. The superscript $(\Diamond)$ indicates that the results are the same for the modal languages with both $\Diamond$ and $\Next$ and with $\Diamond$ only. 
The~``$=$'' and ``$\subseteq$'' in the second column indicate constant or expanding domain, respectively.
One can see that decidability does not depend on which fragment is considered, while, not so surprisingly, the computational complexity does depend on the fragment.\!\!%
\footnote{It has been shown recently~\cite{DBLP:journals/corr/abs-2603-04512} that decidability can also depend on the fragment considered: for the \emph{difference logic} \textbf{Diff} on finite Kripke frames, the one-variable fragment, $\QoneML$, is decidable, while the monodic two-variable fragment with counting, $\CtwomonML$, is not.}\
We note that, without the \NRDC{} features, all these logics are decidable in at most 2\ExpTime~\cite{DBLP:journals/jsyml/WolterZ01, GabEtAl03}, so  \NRDC{}  features have a significant impact.
There are, however, positive results to emphasise:
we obtain elementary decidability for the fundamental logics $\Kn$ and $\Sfive$; moreover, on
expanding-domain models
based on frames without infinite ascending chains,
even transitive closure can be added while still retaining decidability. 

The main steps of our proofs are as follows. We adapt the machinery developed for monodic
fragments without the \NRDC{} features. The basic abstraction needed is \emph{quasimodels}: Kripke models in which first-order domains are replaced by quasistates (sets of types), while sets of runs (functions from worlds to types in quasistates) are used to represent individual domain elements. To deal with the \NRDC{} features, quasistates now become \emph{multisets} of types and, similarly, sets of runs become multisets to reflect the need to count domain elements. Our first basic result shows that, when dealing with monodic fragments, Kripke models can be replaced by such `quantitative' quasimodels. In contrast to the case without the \NRDC{} features, quasimodels are typically not yet well-behaved. For instance, if the first-order fragment does not have the fmp, one can easily enforce infinite branching of Kripke frames. Hence, we replace quasimodels by \emph{weak} quasimodels in which the $\Diamond$-saturation conditions on runs are weakened. We show that one can reconstruct quasimodels from weak quasimodels, but only by expanding the underpinning Kripke frame. In particular, we prove that, on~$\Kn$ and $\Sfiven$ frames, monodic fragments are determined by weak quasimodels containing only exponentially many worlds. This approach underpins our decidability results for $\Kn$ and $\Sfiven$ and establishes the exponential fmp for the one-variable fragments of both first-order modal logic on  $\Kn$ and $\Sfiven$ frames and Fitting's $\mathbf{S5}\lambda$~\cite{Fit02}. To obtain tight complexity bounds, we employ techniques developed for $\textsf{GF}$~\cite{DBLP:journals/jsyml/Gradel99,DBLP:journals/corr/BaranyGO13} and $\textsf{C}^2$~\cite{pratt2023fragments}. The decidability proof for expanding-domain models with transitive closure also relies on 
weak quasimodels and shows, using Dickson's Lemma, a (non-primitive recursive) bound on their size. The lower bounds are all shown by reduction from known results for related fragments~\cite{GabEtAl03,HamKur15}.     

This article significantly extends ideas first developed by the authors in the context of modal and temporal description logics~\cite{AKMW:KR24, AKMW:DL24}. \revised{R2.2}{For instance, the \emph{multiset}-based  framework for studying monodic fragments (Section~\ref{sec:quasimodels})
is developed here for the first time, weak quasimodels (Sections~\ref{sec:kn:weak:quasimodels} and~\ref{sec:s5n:weak:quasimodels}) are also a new contribution. Moreover, covering  monodic fragments based on  the guarded fragment and the two-variable fragment with counting (Sections~\ref{sec:guarded} and~\ref{sec:2var}) requires specialised techniques not needed for standard description logics.}


%% file: 2_preliminaries.tex

\section{Preliminaries}\label{sec:prel}

In this section we first introduce the language $\QML$ of first-order modal logic with equality, constants and definite descriptions and then define its semantics, where constants and definite descriptions are interpreted non-rigidly. We then define a number of monodic fragments of $\QML$ and formulate the main decision problems for the fragments.

\subsection{Syntax and Semantics}\label{sec:syntax}

Let $A$ be a  finite set of \emph{modalities}. 
The alphabet of the \emph{first-order modal language with equality and definite descriptions},
$\QML$,
consists of countably infinite and pairwise disjoint sets of
predicate symbols $\Pred$
(each of a fixed non-negative arity),
constants $\Ind$
and variables $\Var$,
equality $=$,
Boolean connectives~$\lnot$ and~$\land$,
the existential quantifier~$\exists$,
the definite description operator~$\defdes$,
and modal operators~$\Diamond_{a}$ (diamond), for $a \in A$.
\emph{Terms} $\tau$ and \emph{formulas} $\varphi$ of $\QML$ are defined by mutual induction:
\begin{align*}
\tau \ \ ::=\ \ & \ x \ \mid \ c \ \mid \ \defdes x . \varphi,\\
\varphi \ \ ::= \ \ & \
P(\tau_{1}, \ldots, \tau_{m})
\ \mid \
\tau_{1} = \tau_{2} \ \mid \ 
\neg \varphi \ \mid \ (\varphi_1 \land \varphi_2) \ \mid \ 
 \exists x\, \varphi \ \mid \
\Diamond_{a} \varphi, 
\end{align*}
where $x \in \Var$, $c \in \Ind$,  $P \in \Pred$
($m$-ary, for $m \geq 0$) and $a\in A$.
Formulas of the form
$P(\tau_{1}, \ldots, \tau_{m})$ and $\tau_{1} = \tau_{2}$ are called \emph{atomic}.
Other standard syntactic abbreviations are assumed:
in particular,
$\varphi_1 \lor \varphi_2$ stands for~$\neg (\neg \varphi_1 \land \neg \varphi_2)$,
$\varphi_1 \to \varphi_2$ for $\lnot \varphi_1 \lor \varphi_2$,
$\varphi_1 \leftrightarrow \varphi_2$ for $(\varphi_1 \to \varphi_2) \land (\varphi_2 \to \varphi_1)$;
$\forall x\,\varphi$ abbreviates $\lnot \exists x\, \lnot \varphi$, and
$\Box_{a} \varphi$ abbreviates $\lnot \Diamond_{a} \lnot \varphi$.

The \emph{free variables}
 in terms and formulas
are defined in the standard way by mutual induction~(see, e.g.,~\cite{FitMen12}):
in particular, the free variables
of $P(\tau_{1}, \ldots, \tau_{m})$ are those of $\tau_{1}, \ldots, \tau_{m}$, while the free variables of
$\defdes x. \varphi$ and $\exists x\,\varphi$ are the free variables
of $\varphi$, with the exception of $x$.
A \emph{$\QML$ sentence} is a $\QML$ formula without free variables, and a \emph{$\QML$ theory} is a finite set of $\QML$ sentences.
The set of constants occurring in a formula $\varphi$ is denoted by $\Ind(\varphi)$, and for a theory $\Gamma$ we set $\Ind(\Gamma) = \bigcup_{\varphi \in \Gamma} \Ind(\varphi)$.

The \emph{size} of terms and formulas is introduced by mutual induction, setting
$|x| = |c| = 1$,
$|\defdes x. \varphi| = 1 + |\varphi|$,
$
|P(\tau_{1}, \ldots, \tau_{m}) | = 1 + | \tau_{1} | + \ldots + | \tau_{m} |
$,
and the other cases given in the natural way (we treat the Boolean connectives other than $\lnot$ and $\land$ as abbreviations).
The sets of \emph{subformulas} of terms and formulas are also defined by mutual induction, with
$\sub[]{x} = \sub[]{c} = \emptyset$,
$\sub[]{\defdes x. \varphi} = \sub[]{\varphi}$,
$\sub[]{P(\tau_{1}, \ldots, \tau_{m})} = \{ P(\tau_{1}, \ldots, \tau_{m}) \} \cup \sub[]{\tau_{1}} \cup \cdots \cup \sub[]{\tau_{m}}$,
and the remaining cases as usual.
Finally, the \emph{modal depth} of terms and formulas is the maximum number of nested modal operators,
defined again by mutual induction:
$\md(x) = \md(c) = 0$, $\md(\defdes x. \varphi) = \md(\varphi)$,
$\md(P(\tau_{1}, \ldots, \tau_{m})) = \max \{ \md(\tau_{1}), \ldots, \md(\tau_{m}) \}$,
$\md(\Diamond_{a} \varphi) = \md(\varphi) + 1$, and with the other cases given in the obvious way.

A \emph{partial interpretation with expanding domains} is a structure
$\Mmf = (\Fmf, \Dom, \cdot)$, where
\begin{itemize}
\item $\Fmf = (W, \{ R_{a} \}_{a \in A})$ is a \emph{frame}, with $W$
  being a non-empty set of \emph{worlds} and
  $R_{a} \subseteq W \times W$ being an \emph{accessibility relation}
  on $W$, for each {modality} $a \in A$ (we say that $\Mmf$ is
  \emph{based on $\Fmf$});
\item $\Dom$
is a function associating with every $w \in W$ a
  non-empty set, $\Dom_{w}$, called the \emph{domain of $w$
    in~$\Mmf$}, such that $\Dom_{w} \subseteq \Dom_{v}$, whenever
  $w R_{a} v$, for some $a \in A$;
\item $\cdot$ is a function associating with each $w \in W$ a
  \emph{partial} first-order interpretation $\Mmf({w})$ with domain
  $\Dom_w$ so that $P^{\Mmf({w})} \subseteq \Dom_{w}^{m}$, for each predicate
  $P \in \Pred$ of arity $m$, and
  $c^{\Mmf(w)} \in \Dom_{w}$, for each constant symbol $c$ in \emph{some subset} of $\Ind$.
\end{itemize}

Hence, every $\cdot^{\Mmf(w)}$ is a total function on~$\Pred$ but a \emph{partial} function on~$\Ind$.
If $\Mmf(w)$ is defined on $c \in \Ind$, then we say
that $c$ \emph{designates at~$w$}.
If every \mbox{$c \in \Ind$} designates at $w \in W$, then $\Mmf(w)$ is called \emph{total}.
We say that $\Mmf = ( \Fmf, \Dom, \cdot)$ is a \emph{total} interpretation if every $\Mmf(w)$, for $w \in W$, is a \emph{total} interpretation.
In the remainder of this work, we refer to partial interpretations as `interpretations' and add explicitly the adjective `total' when this is the case.

An \emph{interpretation with constant domains} is a special case of an interpretation with expanding domains where the function $\Dom$ is such that $\Dom_{w} = \Dom_{v}$, for every $w, v \in W$. With an abuse of notation, we denote the common domain by $\Delta$ and
call it the \emph{domain of $\Mmf$}.

Given an interpretation $\Mmf = (\Fmf, \Dom, \cdot)$,
an \emph{assignment at $w$} is a function~$\assign$ from $\Var$ to $\Dom_{w}$.
An \emph{$x$-variant} of an assignment $\assign$ at $w$ is an assignment $\assign'$ at $w$ that can differ from $\assign$ only on $x$.
Observe that, if $w R_{a} v$ for some $a \in A$, then an assignment at $w$ is also an assignment at $v$.
The definitions of the \emph{value} $\tvalue{\tau}{w}{\assign}$
of a term $\tau$ under assignment $\assign$ at world~$w$ of $\Mmf$, and of \emph{satisfaction} $\Mmf, w \models^\assign \varphi$ of a formula $\varphi$ at world $w$ of~$\Mmf$ under assignment $\assign$ are defined by mutual induction. We set
\begin{equation*}
\tvalue{\tau}{w}{\assign} =
\begin{cases}
\assign(x), & \text{ if } \tau \text{ is } x \in \Var; \\
c^{\Mmf(w)}, & \text{ if } \tau \text{ is } c \in \Ind \text{ and $c^{\Mmf(w)}$ is defined};\\
\assign'(x), & \text{ if } \tau  \text{ is } \defdes x . \varphi \text{ and } \Mmf, w \models^{\assign'} \varphi, \text{ for \emph{exactly one}}\\[-4pt]
&
\hspace*{10em}\text{ $x$-variant $\assign'$ of $\assign$ at $w$};\\
\text{undefined}, & \text{ otherwise};
\end{cases}
\end{equation*}
and define
\[%
{\renewcommand{\arraystretch}{1.25}%
\renewcommand{\arraycolsep}{5pt}%
\begin{array}{lcl}
		\Mmf, w \models^{\assign} P(\tau_{1}, \ldots, \tau_{m}) & \text{iff} & 
		\tvalue{\tau_{1}}{w}{\assign},\dots,\tvalue{\tau_{m}}{w}{\assign} \ \text{are defined and} \\
	  & &
	  (\tvalue{\tau_{1}}{w}{\assign}, \ldots, \tvalue{\tau_{m}}{w}{\assign}) \in P^{\Mmf(w)}; \\
		\Mmf, w \models^{\assign}  \tau_{1} =  \tau_{2} & \text{iff}   & \text{both } \tvalue{\tau_{i}}{w}{\assign}  \text{ are defined and} \
		\tvalue{\tau_{1}}{w}{\assign} = \tvalue{\tau_{2}}{w}{\assign};\\
		\Mmf, w  \models^{\assign} \neg \varphi & \text{iff}  &  \Mmf, w  \not\models^{\assign} \varphi; \\
		\Mmf, w  \models^{\assign} \varphi_1 \land \varphi_2 & \text{iff}   & \Mmf, w  \models^{\assign} \varphi_1 \text{ and } \Mmf, w  \models^{\assign} \varphi_2; \\
		\Mmf, w  \models^{\assign} \exists x\, \varphi & \text{iff} 
		& \parbox[t][][t]{\linewidth}{$\Mmf, w \models^{\assign'} \varphi$, for some $x$-variant $\assign'$ of $\assign$ at $w$;}\\
		\Mmf, w  \models^{\assign} \Diamond_{a} \varphi & \text{iff} 
		& \parbox[t][][t]{\linewidth}{$\Mmf, v \models^{\assign} \varphi$, for some $v \in W$ such that $wR_{a}v$.}\\
\end{array}
}%
\]
If $\tvalue{\tau}{w}{\assign}$ is defined, then we say that $\tau$ \emph{designates under $\assign$ at $w$} in $\Mmf$.
Observe that variables always designate under an assignment $\assign$ at $w$, and that assignments play no role in constants' designation (hence, we will simply say that a constant \emph{designates at $w$ in $\Mmf$}).
To simplify the notation, when there is no risk of confusion, we will often write $P^{w}$, $c^{w}$ and $\tau^{w, \assign}$ in place, respectively, of $P^{\Mmf(w)}$, $c^{\Mmf(w)}$ and $\tvalue{\tau}{w}{\assign}$.

\revised{R2.3}{Our semantic choices regarding the term-forming operator $\defdes$ for definite descriptions align with Russell's paraphrase~\cite{Rus05}, in that an atom $\alpha(\defdes x.\varphi(x, \avec{y}), \avec{\tau})$, with a list of variables $\avec{y}$ and terms $\avec{\tau}$, 
can be equivalently rewritten as 
\begin{equation}\label{eq:russel}
\exists z\, \bigl( \alpha(z, \avec{\tau}) \land \varphi(z, \avec{y}) \land \forall z'\, (\varphi(z', \avec{y}) \to z' = z) \bigr),
\end{equation}
where $z$ and $z'$ are fresh variables.
Such a replacement, however,
is not available in the fragments of our interest, and so Theorem~\ref{thm:reductions}~(a) presents
an alternative.}

\revised{R2.4}{Since the values of constants and definite descriptions in our setting can vary across worlds, we call such terms \emph{non-rigid designators}. An alternative option, often adopted in the literature, forces constants to take the same value at all (accessible) worlds: such constants are \emph{rigid designators} (see Example~\ref{ex:defdes} below).}

We say that a formula $\varphi$ is \emph{true in $\Mmf$}, written $\Mmf \models \varphi$, iff $\varphi$ is satisfied under every assignment $\assign$ at every world $w$ of $\Mmf$.
Dually, we say that $\varphi$ is \emph{satisfied in~$\Mmf$} iff $\varphi$ is satisfied under some assignment $\assign$ at some world $w$ of~$\Mmf$.
A theory~$\Gamma$ is \emph{true in $\Mmf$}, written $\Mmf \models \Gamma$, if every sentence in $\Gamma$ is true in $\Mmf$.

Given a class of frames $\Cmc$, we say that $\varphi$ is \emph{valid on $\Cmc$} (or \emph{$\Cmc$-valid}) if $\varphi$ is true in every interpretation $\Mmf$ based on a  frame $\Fmf \in \Cmc$.
Dually, we say that $\varphi$ is \emph{satisfiable in $\Cmc$} (or \emph{$\Cmc$-satisfiable}) if there exists an interpretation $\Mmf$ based on a frame $\Fmf\in \Cmc$ such that $\varphi$ is satisfied in~$\Mmf$.
A formula $\varphi$ is said to be
a \emph{global $\Cmc$-consequence of a theory~$\Gamma$}
if 
$\varphi$ is true in any interpretation $\Mmf$ based on a frame in $\Cmc$ such that $\Mmf \models \Gamma$; recall that theories have no free variables.

In the following,
the class of all frames with $n$ accessibility relations, for $n \geq 1$, is denoted by $\Kn$. The class of frames with $n$ equivalence relations is $\Sfiven$, and we write $\Sfive$ in place of $\Sfive_{1}$.

\begin{example}[Vulcan and Venus]\em
\label{ex:defdes}
Let us consider some examples of $\QML$ formulas:
%
%
first, ``it is conceivable \emph{that} Vulcan is the planet orbiting between the Sun and Mercury'' is represented as
\[
\Diamond \big( {\sf vulcan} = \defdes z. {\sf OrbitsBetween} (z, {\sf sun}, {\sf mercury}) \big),
\]
while
``even though
such a planet
does not exist''
can be written as
\[
\lnot \exists x \,\big(x =  {\sf vulcan} \big)
\land 
\lnot \exists x \,\big(x =  \defdes z. {\sf OrbitsBetween} (z, {\sf sun}, {\sf mercury}) \big).
\]
Here, neither the constant ${\sf vulcan}$ nor the definite description $\defdes z. {\sf OrbitsBetween} (z, \linebreak {\sf sun}, {\sf mercury})$ designate in the current world.

Second, ``it is known \emph{of} the planet orbiting between Mercury and Earth that it is Venus'' can be rendered as 
\[
\exists x \,\big(x = \defdes z. {\sf OrbitsBetween} (z, {\sf mercury}, {\sf earth}) \ \land \
\Box (x = {\sf venus} ) \big).
\]
Note that, in $\Sfive$ frames, the above formula implies that
${\sf venus}$
is a rigid designator; cf.~\cite[Proposition 19.3.4]{FitMen12}.
\end{example}

\subsection{Monodic Fragments and Decision Problems}\label{sec:monodic:fragments}

We next introduce various fragments of $\QML$. A $\QML$ formula $\varphi$ is called \emph{monodic} if every subformula of $\varphi$ of the form $\Diamond_{a} \psi$ has at most one free variable. 
We denote by $\QmonML$ the set of monodic $\QML$ formulas, and call it the \emph{monodic fragment} of $\QML$. The monodic fragment contains full first-order logic and so is undecidable. 
To obtain potentially decidable fragments of~$\QmonML$, we restrict its first-order component. 

A minimal language we consider in this paper extends the \emph{one-variable fragment} of first-order logic. Denote by $\QoneML$ the set of all formulas in $\QML$ that use a single variable and only predicates of arity at most one (note that equality, constants, and definite descriptions are allowed). Clearly, all $\QoneML$ formulas are monodic by definition. Variants of this fragment have been investigated extensively as product modal logics (see, e.g.,~\cite{DBLP:journals/logcom/Marx99,GabEtAl03,DBLP:journals/apal/GabelaiaKWZ06,DBLP:conf/aiml/HampsonK12}), and we point out specific relevant results below when discussing $\QoneML$.


At the other end of the spectrum, we consider two maximal languages. First, we extend
$\mathsf{C}^{2}$, the \emph{two-variable fragment of first-order logic with counting quantifiers $\exists^{\geq k}x$, $k\geq 0$}. Denote by $\CtwomonML$ the set of all formulas in $\QmonML$ extended by $\exists^{\geq k}$
and constructed using only \emph{two variables}
and predicates of arity at most two. We also have equality, constants and
\emph{closed} definite descriptions, which are terms of the form $\defdes x.\chi(x)$, for a $\CtwomonML$ formula $\chi(x)$ with a single free variable.
The semantics of the counting quantifiers is given in the obvious way by modifying the quantifier case as follows:
\begin{multline*}
		\Mmf, w  \models^{\assign} \exists^{\geq k} x\, \varphi \quad \text{iff}  \quad
		\Mmf, w \models^{\assign'} \varphi, \ \text{for at least $k$ distinct} \\[-4pt]\text{$x$-variants $\assign'$ of $\assign$ at $w$.} 
\end{multline*}
It should be clear that counting can be expressed in first-order logic with equality, but $\exists^{\geq k}$ would require $k$ variables. We use standard abbreviations such as \revised{R1.9}{$\exists^{\leq k} x\,\varphi = \neg \exists^{\geq k + 1}x\,\varphi$}
and $\exists^{=k} x\,\varphi = \exists^{\leq k}x\,\varphi\land  \exists^{\geq k}x\,\varphi$. 
\begin{example}[The number of planets]\em
We
illustrate the interaction between modal operators and counting quantifiers.
On arbitrary frames, $\Diamond \exists^{\geq 8} x \, {\sf Planet}(x)$ entails $\exists^{\geq 8} x \, \Diamond {\sf Planet}(x)$ under constant domains, but not expanding domains.
Instead, 
$\exists^{\geq 8} x\, \Box{\sf Planet}(x)$ entails $\Box\exists^{\geq 8}x\,{\sf Planet}(x)$ under both constant and expanding domains.
The converse implications fail in both cases.
 
\end{example}

Observe
that \emph{open} definite descriptions of the form $\defdes x.\chi(x, y)$, where $\chi(x, y)$ is a formula with free variables $x, y$, add expressive power already
to the non-modal two-variable fragment $\textsf{FO}^2$ without counting and equality.
First,
it can be seen that $\defdes y. F(x, y)$ effectively defines a partial unary function, and $\forall x\,F(x, \defdes y. F(x, y))$ ensures that binary predicate $F$ is interpreted as a \revised{R2.5}{total} unary function; that is, the formula is equivalent to $\forall x\exists^{= 1}y\, F(x,y)$. 
Moreover,
it is known that
the class $[\forall^2, (0,1), (1)]$ of formulas  in prenex normal form with quantifier prefix~$\forall\forall$ and featuring one binary predicate symbol, $P$, and one unary function symbol, $f$, but no equality, is undecidable; see, e.g.,~\cite[Theorem 4.2.1]{BoeEtAl97}.
Hence, extending $\textsf{FO}^2$ with open definite descriptions immediately leads to undecidability: any $[\forall^2, (0,1), (1)]$-sentence~$\varphi$ is
equisatisfiable with
$\varphi' \land \forall x\,F(x, \defdes y. F(x, y))$, where $\varphi'$ is the result of repeated replacement of each function term $f(t)$ in the formula with $\defdes y.F(t,y)$, where term $t$ has a free variable,~$x$, which is inherited by the
replacing
definite description, while $\defdes$ binds the other variable, $y$.
For instance,
$\forall x \forall y\, P\big( f(f(x)), f(y) \big)$ is equisatisfiable with the conjunction of 
$\forall x \forall y\, P \big( \iota y.F(\iota y.F(x,y), y),\, \iota x.F(y, x) \big)$ and $\forall x\,F(x, \defdes y. F(x, y))$.

The second maximal language extends $\mathsf{GF}$, the \emph{guarded fragment of} first-order logic.  
The fragment $\QGuardmonMLb$ of $\QmonML$ is defined by mutual induction on its terms and formulas as follows.
Terms in $\QGuardmonMLb$ are limited to variables, constants and \emph{closed} definite descriptions,
which are of the form $\defdes x.\chi(x)$, for a $\QGuardmonMLb$ formula $\chi(x)$ with a single free variable.
Formulas in $\QGuardmonMLb$ are constructed as in $\QmonML$ except that quantifiers are required to be guarded:
\begin{equation*}
	\exists x_{1}\cdots \exists x_{k}\,(\alpha\wedge \varphi), 
\end{equation*}
where $\alpha$ is a predicate or equality atom that contains all free variables of~$\varphi$.
Note that if $\varphi$ has only one free variable, $x$, then $\exists x\,\varphi(x)$ 
 can be considered guarded: it is equivalent to $\exists x\,((x = x)\land \varphi(x))$, 
 with $x = x$ as its guard.
Observe that in the definition of $\QGuardmonMLb$ the restriction to closed definite descriptions is necessary to guarantee decidability of the logic, even without modal operators.
We have seen that the formula \revised{R1.10}{$\forall x\,((x=x) \rightarrow F(x, \defdes y. F(x, y)))$}
ensures that $F$ is a graph of a
total
function, but the non-modal guarded fragment with such functionality statements is undecidable~\cite[Theorem~5.8]{DBLP:journals/jsyml/Gradel99}. 

We also consider fragments that do not contain definite descriptions or constants: if in a fragment name we replace the $\defdes$ subscript  by $c$, then we refer to its restriction that admits constant symbols but not definite descriptions, and if we drop the subscript altogether, then we refer to the restriction without definite descriptions and constant symbols (i.e., having only variables as terms of the language).

Let $\Lmc$ be a fragment of $\QML$ with $n$ modalities 
and $\Cmc$ a class of frames
with $n$ accessibility relations.
We consider the following decision problems.
\begin{description}
\item[$\Cmc$-Validity in $\Lmc$:] Given an $\Lmc$-formula $\varphi$,
is $\varphi$ valid on $\Cmc$?
\item[Global $\Cmc$-Consequence in $\Lmc$:]
Given an $\Lmc$-formula $\varphi$ and an $\Lmc$-theory $\Gamma$,
is $\varphi$ a global $\Cmc$-consequence of $\Gamma$?
\end{description}

In the sequel, we can add qualifiers to the problems above
when we intend to refer to the corresponding problem restricted to certain classes of interpretations:
for instance, ``total $\Cmc$-validity in $\Lmc$'' is the problem of validity of $\Lmc$-formulas in total interpretations based on frames from $\Cmc$.

Our main concern is decidability and complexity of the above reasoning problems for \revised{R2.6}{such fundamental classes of frames as $\Kn$ and $\Sfiven$.}

The \emph{signature} $\textit{sig}(\varphi)$ of a formula $\varphi$ is the set of predicate symbols and constants in~$\varphi$; the signature $\textit{sig}(\Gamma)$ of a theory~$\Gamma$ is defined similarly. Let $\Cmc$ be a class of frames. For sentences
$\varphi$ and $\varphi'$, we write $\varphi\leq_{\Cmc}\varphi'$ if  every interpretation based on a frame from $\Cmc$ and satisfying $\varphi$ at some $w$ can be made to satisfy~$\varphi'$ at the same $w$ by modifying the interpretation of symbols in $\textit{sig}(\varphi')\setminus\textit{sig}(\varphi)$, while preserving the interpretation of symbols in $\textit{sig}(\varphi)$.
A sentence $\varphi'$ is called a
\emph{$\Cmc$-model-conservative extension} of a sentence $\varphi$ if $\varphi\leq_{\Cmc}\varphi'$ and every interpretation based on a frame from $\Cmc$ and satisfying $\varphi'$ at some~$w$ also satisfies~$\varphi$ at the same~$w$. It should be clear that any polynomial-time transformation $\tau$ of formulas $\varphi$ such that $\tau(\varphi)$ is a $\Cmc$-model-conservative extension of $\varphi$ is a polytime reduction of $\Cmc$-validity of formulas of the form~$\neg\varphi$ to formulas of the form~$\neg\tau(\varphi)$.

We also define the notion of $\Cmc$-model-conservative extensions for pairs of the form $(\varphi, \Gamma)$, for a sentence $\varphi$ and a theory $\Gamma$:  we write \revised{R1.11}{$(\varphi,\Gamma)\leq_\Cmc (\varphi',\Gamma')$}
if, for any interpretation $\Mmf = ( \Fmf, \Dom, \cdot)$ such that $\Fmf \in \Cmc$, $\Mmf\models\Gamma$ and $\Mmf, w\models\varphi$, for some $w$ in $\Fmf$, there is an extension
\revised{R2.7}{$\Mmf' = (\Fmf, \Dom, \cdot')$ that coincides with $\Mmf$ on $\textit{sig}(\varphi)\cup\textit{sig}(\Gamma)$} and such that $\Mmf'\models\Gamma'$ and $\Mmf', w\models\varphi'$.
Then $ (\varphi',\Gamma')$ is a \emph{$\Cmc$-model-conservative extension} of~$(\varphi,\Gamma)$ if \revised{R1.11}{$(\varphi,\Gamma)\leq_\Cmc (\varphi',\Gamma')$}
and, 
for any interpretation~$\Mmf'$ based on a frame $\Fmf\in\Cmc$ such that $\Mmf'\models\Gamma'$ and $\Mmf', w'\models\varphi'$, for some~$w'$ in $\Fmf$, we have $\Mmf'\models\Gamma$ and $\Mmf', w'\models\varphi$.

\section{Reductions and Related Formalisms}
\label{sec:reductions}

In this section,
we first observe that reasoning in $\QoneMLc$ and a variant in which constants are replaced by an ``elsewhere'' quantifier are mutually polytime-reducible. This enables us to transfer existing undecidability results to $\QoneMLc$. We also establish polytime reductions that allow us to eliminate (closed) definite descriptions and partial designators. 

\input{3_counting}

We use Theorem~\ref{th:diff} to show how our first undecidability 
result follows in a straightforward way from earlier work on products of modal logics~\cite{DBLP:conf/aiml/HampsonK12}.

\begin{theorem}\label{thm:undecglobal}
	For constant-domain models, global $\Kn$-consequence  with $n\geq 1$ and global $\Sfiven$-consequence with $n\geq 2$ in $\QoneMLc$ are undecidable.
\end{theorem}
\begin{proof}
(Sketch) We give the proof for $\Kn$, $n\geq 1$. Then undecidability for
$\Sfiven$, $n\geq 2$, can be shown using standard reductions of modal logic on $\Kn$ frames for $n=1$ to modal logic on $\Sfiven$ frames 
for $n=2$~\cite[Theorem 2.37]{GabEtAl03}. To show undecidability of
the global $\Kn$-consequence in $\QoneMLc$ it suffices to show
undecidability of the global $\Kn$-consequence in $\MLDiff$. This is
shown in~\cite{DBLP:conf/aiml/HampsonK12}, where the result is
formulated in terms of a product modal logic that corresponds to the
extension of $\MLDiff$ with the universal modality on $\Kn$ frames with constant domains. It is straightforward to see that one can replace the universal modality by the global $\Kn$-consequence in the undecidability proof.
\end{proof}

\begin{remark}\label{rem:s5:global}\em
Note that Theorem~\ref{thm:undecglobal} does not hold for \Sfive. In fact, the decidability and complexity results we show below for \Sfive-validity also hold for global \Sfive-consequence since, for any relevant $\Gamma$ and $\varphi$, the following holds: $\varphi$ is a global \Sfive-consequence of $\Gamma$ iff $\Box\bigwedge\Gamma \rightarrow \varphi$ is \Sfive-valid. 
\end{remark}

\subsection{Simplifying the Problem Landscape}

We first observe that the distinction between partial and total interpretations and
the availability of (closed) definite descriptions
do not affect the complexity of our decision problems for the fragments we are concerned with.

\begin{theorem}\label{thm:reductions}
Let $\Lmc$ be $\QoneML$, $\CtwomonML$, $\QGuardmonMLb$ or $\QmonML$, and let $\Cmc$ be any class of frames.
\begin{enumerate}
\item[\textup{(a)}] $\Cmc$-validity in $\Lmc$ is polytime-reducible to $\Cmc$-validity in $\Lmc$ without definite descriptions, with both constant and expanding domains. 
\item[\textup{(b)}] In $\Lmc$, the problems of $\Cmc$-validity in partial and total interpretations are mutually polytime-reducible, with both constant and expanding domains. 
\end{enumerate}
Items \textup{(a)} and \textup{(b)} also apply to global $\Cmc$-consequence in $\Lmc$.
\end{theorem}

\begin{proof}
Without loss of generality, we consider the problems for sentences.

\paragraph{Item \textup{(a)}}
\revised{R2.10}{We begin by considering the full monodic fragment.
One could eliminate definite descriptions 
by replacing each atom of the form
$\alpha(\defdes x.\varphi(x, \avec{y}), \avec{\tau})$ with its Russell's paraphrase~\eqref{eq:russel}, thus obtaining an equivalent formula. 
In order to avoid an exponential blowup, we normalise the input by ensuring that definite descriptions contain only atomic
formulas. 
More precisely, in a formula~$\varphi$, we replace any occurrence of
a complex $\defdes x.\psi(x, \avec{y})$
by
$\defdes x.P_{\psi}(x, \avec{y})$, where $P_{\psi}$ is a fresh
predicate symbol. As in the proof of Theorem~\ref{th:diff}, for global $\Cmc$-consequence, we simply
add $\norml = \forall x \forall \avec{y}\, \bigl(P_{\psi}(x, \avec{y}) \leftrightarrow \psi(x, \avec{y})\bigr)$  to the
resulting theory; for $\Cmc$-validity, we add
the conjunct
$\Box^{\pi} \norml$ to the resulting formula,
for every $\pi \in \rpath(\varphi, \psi)$.}
By repeated application of this procedure to a formula (and theory), always starting from the innermost definite descriptions (containing no other definite
descriptions),
we obtain the normalised result in polynomial time.

\revised{R2.10}{However,
while Russell's paraphrase~\eqref{eq:russel} is suitable for the full monodic fragment, the replacement formula does \emph{not} belong to any of our three fragments of interest, $\QoneML$, $\CtwomonML$ or $\QGuardmonMLb$: it requires additional variables, and its last conjunct is not guarded. So, we present a modified reduction, which, however, applies only to closed definite descriptions.}

\revised{R2.10}{
For the three fragments $\Lmc$ with closed definite descriptions,
we perform normalisation and $\defdes$-elimination together. We first focus on the elimination of a single~$\defdes x. \psi(x)$ and then describe the complete  process. In a given $\varphi$, each atom of the form~$\alpha(\defdes x. \psi(x), \avec{\tau})$ is replaced
by
\begin{equation}\label{eq:russel:skolem}
\alpha(c_{\psi}, \avec{\tau})
\land 
P_{\psi}(c_{\psi})
\land
\forall x\, \bigl(P_{\psi}(x) \to x = c_{\psi} \bigr),
\end{equation}
which is obtained by introducing a fresh Skolem constant, $c_{\psi}$, for the existential quantifier in the Russell's paraphrase~\eqref{eq:russel}; this, however, relies on $x$ being the only free variable in $\psi$.}
Let $\varphi^{\ast}$ be the result and $\singl$ as in the proof of Theorem~\ref{th:diff}. Observe that these formulas belong to $\QoneML$, $\CtwomonML$ and $\QGuardmonMLb$, respectively, whenever $\varphi$ does.
\begin{claim}
\label{cla:skolemconst}
For any $\Mmf^{\ast}$ with $\Mmf^{\ast} \models \singl$, we have $\Mmf^{\ast}, w \models^{\assign} \varphi$ iff
$\Mmf^{\ast}, w \models^{\assign} \varphi^{\ast}$, for any world $w$ of $\Mmf^{\ast}$ and assignment $\assign$ at~$w$.
\end{claim}
\begin{proof}
The claim is shown by induction on the structure of $\varphi$. We focus on the base case, $\varphi = \alpha(\defdes x.\psi(x), \avec{\tau})$.
Take any $w\in W$ in $\Mmf^\ast$ and $\assign$ at $w$. We distinguish three cases.
If  $\psi(x)$ is true only on $d\in \Delta_w$,  
then \mbox{$(\defdes x.\psi(x))^{w, \assign} = d$} and, since $\Mmf^{\ast} \models \singl$, we have $P_{\psi}^{w} = \{ d \}$ and $\Mmf^{\ast}, w \models P_{\psi}(c_{\psi})$, so $\smash{c_{\psi}^{w}} = d$.
Thus,
the atom and its replacement,~\eqref{eq:russel:skolem},
have the same truth value
under~$\assign$ at~$w$.
If~$\psi(x)$ is true on two or more elements, then $\defdes x.\psi(x)$ does not designate at $w$, while the last conjunct of~\eqref{eq:russel:skolem} is false at $w$, so neither the atom nor its replacement is satisfied under $\assign$ at $w$.
If $\psi(x)$ is false on every domain element, then again $\defdes x.\psi(x)$ does not designate at $w$,
while $P_{\psi}(c_{\psi})$ is  false at $w$, so neither the atom nor its replacement is satisfied under $\assign$ at $w$.
\end{proof}

The rest of the argument is similar to the proof of Theorem~\ref{th:diff}. For global $\Cmc$-consequence, given an $\Lmc$-formula $\varphi$ and an $\Lmc$-theory $\Gamma$, let $\varphi^*$ and~$\Gamma^*$ be the results of eliminating one $\defdes x.\psi(x)$ from $\varphi$ and $\Gamma$, respectively. Also, let $\Xi_\Gamma$ and~$\Xi_\varphi$ be the sets of $\Lmc$-sentences $\singl$ required for $\Gamma$ and $\varphi$, respectively. It follows from Claim~\ref{cla:skolemconst} that $(\varphi^{\ast}, \Gamma^{\ast} \cup
\Xi_{\Gamma} \cup \Xi_{\varphi})$ is a $\Cmc$-model-conservative extension of
$(\varphi,\Gamma)$.
For $\Cmc$-validity, given an $\Lmc$-formula $\varphi$, we add conjuncts $\Box^{\pi}\singl$ to $\varphi^{\ast}$, for  every $\pi\in\rpath(\varphi,\alpha)$ with atom $\alpha$ of the form~$\alpha(\defdes x. \psi(x), \avec{\tau})$. Again, one can show that the resulting $\Lmc$-formula is a $\Cmc$-model-conservative extension of $\varphi$.

The complete $\defdes$-elimination procedure, similarly to normalisation for the full monodic fragment, proceeds in a bottom-up fashion starting from the innermost definite descriptions, thus resulting in the output of polynomial size.

\paragraph{Item~\textup{(b)}}
\label{point:redpartialtotal}
Due to Item~(a), we can assume that every term $\tau$ occurring in a formula $\varphi$ is either a variable or a constant.

\paragraph{From partial to total interpretations}
  Let $\varphi$ be an $\Lmc$-sentence
  and $\Gamma$ an $\Lmc$-theory.
For a constant $c$, take  a fresh 0-ary predicate $p_{c}$ (i.e., a propositional letter): it represents the statement that $c$ designates at the world. We define $\varphi'$ by replacing each $P(\tau_1, \ldots, \tau_m)$ and each $\tau_1 = \tau_2$ in $\varphi$ with
\begin{equation*}
\bigwedge_{c_{i} \in \{ \tau_{1}, \ldots, \tau_{m} \}\cap \Ind} \hspace*{-2em}p_{c_{i}} \ \ \land \ \ P(\tau_1, \ldots, \tau_m) \qquad \text{ and }\ \ 
\bigwedge_{c_{i} \in \{ \tau_{1}, \tau_{2} \}\cap \Ind} \hspace*{-1em}p_{c_{i}} \ \ \land \ \ (\tau_1= \tau_2), 
\end{equation*}
respectively.  Let $\Gamma' = \{\gamma'\mid \gamma\in\Gamma\}$. 
We show that $\varphi$ is satisfied in $\Mmf$ based on~$\Fmf$ (provided that
$\Mmf \models \Gamma$)
iff $\varphi'$ is satisfied in a total interpretation $\Mmf'$
based on~$\Fmf$ (respectively, provided that $\Mmf' \models \Gamma'$); moreover, $\Mmf$ and $\Mmf'$ share the same domains.
The claim for both problems follows.

$(\Rightarrow)$
Given an interpretation $\Mmf$, 
define a total interpretation $\Mmf'$ as $\Mmf$, with the addition of the following, for every $c_{i} \in \Ind(\varphi) \cup \Ind(\Gamma)$
and every world~$w$:
\begin{equation*}
p_{c_{i}}^{\Mmf'(w)} \ \text{is} \
\begin{cases}
\text{true},
& \ \text{if $c_{i}$ designates at $w$ in $\Mmf$}; \\
\text{false},
& \ \text{otherwise};
\end{cases}
\end{equation*}
and $c^{\Mmf'(w)}$ arbitrary when $c^{\Mmf(w)}$ is undefined.
We show, by induction on a formula $\psi$, that $\Mmf, w \models^{\assign} \psi$ iff $\Mmf', w  \models^{\assign} \psi'$, for every assignment $\assign$ at every world $w$.
For the base case $P(\tau_{1}, \ldots, \tau_{m})$, we have $\Mmf, w \models^{\assign} P(\tau_{1}, \ldots, \tau_{m})$ iff each $c_{i} \in \{ \tau_{1}, \ldots, \tau_{m} \} \cap \Ind$ designates
at $w$ in $\Mmf$ and $(\tvalue{\tau_{1}}{w}{\assign}, \ldots, \tvalue{\tau_{m}}{w}{\assign}) \in P^{\Mmf(w)}$.
By construction of $\Mmf'$, this is equivalent to
$p_{c_{i}}^{\Mmf'(w)}$ being true,
for each $c_{i} \in \{ \tau_{1}, \ldots, \tau_{m} \} \cap \Ind$, and
$(\tau_{1}^{\Mmf'(w), \assign}, \ldots, \tau_{m}^{\Mmf'(w), \assign}) \in P^{\Mmf'(w)}$,
meaning that
$\Mmf', w \models^{\assign} \bigwedge_{c_{i} \in \{ \tau_{1}, \ldots, \tau_{m} \}\cap \Ind} p_{c_{i}} \land P(\tau_{1}, \ldots, \tau_{m})$.
The base case $\tau_1=\tau_2$ is similar, and the inductive cases 
follow by a straightforward application of the inductive hypothesis.
It follows that $\Mmf \models \Gamma$ and
$\Mmf, v \models \varphi$
imply that $\Mmf' \models \Gamma'$ and
$\Mmf', v \models \varphi'$.

$(\Leftarrow)$
Given a  total interpretation $\Mmf'$, 
we define $\Mmf$ that coincides with $\Mmf'$, except for the following, for every $c_{i} \in \Ind(\varphi) \cup \Ind(\Gamma)$ and every world $w$:
\[
c_{i}^{\Mmf(w)} \ \text{is} \
\begin{cases}
c_{i}^{\Mmf'(w)}, & \ \text{if $\ p_{c_{i}}^{\Mmf'(w)}$ is true};
\\
\text{undefined}, & \ \text{otherwise};
\end{cases}
\]
We show, 
by induction on a formula $\psi$, that $\Mmf', w  \models^{\assign} \psi'$ iff $\Mmf, w \models^{\assign} \psi$, for every assignment $\assign$ at every world $w$.
For the base case $P(\tau_{1}, \ldots, \tau_{m})$, we have $\Mmf', w \models^{\assign} \bigwedge_{c_{i} \in \{ \tau_{1}, \ldots, \tau_{m} \}\cap \Ind} p_{c_{i}}$ $\land$ $P(\tau_{1}, \ldots, \tau_{m})$ iff $p_{c_{i}}^{\Mmf'(w)}$ is true,
 for every $c_{i} \in \{ \tau_{1}, \ldots, \tau_{m} \}\cap \Ind$, and $(\tvalue{\tau_{1}}{w}{\assign},$ $\ldots,$ $\tvalue{\tau_{m}}{w}{\assign})$ $\in P^{\Mmf'(w)}$.
By construction of $\Mmf$, this is equivalent to
$c_{i}^{\Mmf(w)}$ being defined, with $c_{i}^{\Mmf(w)} = c_{i}^{\Mmf'(w)}$, for every $c_{i} \in \{ \tau_{1}, \ldots, \tau_{m} \}\cap \Ind$,
and  $(\tvalue{\tau_{1}}{w}{\assign},$ $\ldots,$ $\tvalue{\tau_{m}}{w}{\assign}) \in P^{\Mmf(w)}$.
By definition, it means that
$\Mmf, w \models^{\assign}P(\tau_{1}, \ldots, \tau_{m})$.
The base case $\tau_1=\tau_2$ is similar, and 
the inductive cases 
follow again by a straightforward application of the inductive hypothesis.
In conclusion, 
$\Mmf' \models \Gamma'$ and
$\Mmf', v \models \varphi'$ imply
$\Mmf \models \Gamma$ and
$\Mmf, v \models \varphi$.

\paragraph{From total to partial interpretations}
First, consider global $\Cmc$-consequence.
Let
$\varphi$
be an $\Lmc$-sentence and $\Gamma$ an $\Lmc$-theory. Define $\Gamma'$ by adding to $\Gamma$ the sentences
\begin{equation*}\label{eq:designate}
\exists x \, (x = c),
\end{equation*}
for every constant symbol $c$ occurring in $\Gamma$ or $\varphi$.
We first show that $\varphi$ is
  satisfied in a total interpretation $\Mmf$
  based on $\Fmf$ with   $\Mmf \models \Gamma$
  iff
  $\varphi$ is satisfied in $\Mmf'$
  based on $\Fmf$ with $\Mmf' \models \Gamma'$.
Since  
the additional sentences in $\Gamma'$ are trivially satisfied in any total interpretation, if $\Mmf$ is total and such that $\Mmf \models \Gamma$ and
$\Mmf, v \models \varphi$,
then clearly $\Mmf \models \Gamma'$ and
$\Mmf, v \models \varphi$.
Conversely, every interpretation $\Mmf'$ for which $\Mmf' \models \Gamma'$ and
$\Mmf', v \models \varphi$
can be extended to a total interpretation $\Mmf$ such that $\Mmf \models \Gamma'$, hence $\Mmf \models \Gamma$, and
$\Mmf, v \models \varphi$,
as it only remains to arbitrarily fix the interpretation of constant symbols not occurring in $\varphi$ and $\Gamma$.

Now,
to show that 
$\Cmc$-satisfiability
in total interpretations is polytime-re\-ducible to
$\Cmc$-satisfiability,
we modify the reduction above by taking, for an $\Lmc$-sentence~$\varphi$,
the conjunction $\varphi'$ of $\varphi$ together with all the sentences
$\Box^{\pi} \exists x\,(x = c)$,
where
$\pi \in \rpath(\varphi, \psi)$, for an atomic $\psi$ such that $c \in \Ind(\psi)$.
It can be seen that
$\varphi$ is satisfied in a total interpretation  $\Mmf$ based on~$\Fmf$ iff $\varphi'$ is satisfied in an interpretation $\Mmf'$ based on $\Fmf$; moreover, $\Mmf$
and $\Mmf'$ share the same domains.
\end{proof}

\paragraph{Infinite branching models}
As the following example shows, we can easily enforce infinite branching using equality or counting quantifiers. 
\begin{example}\label{ex:running-example:1}\em%
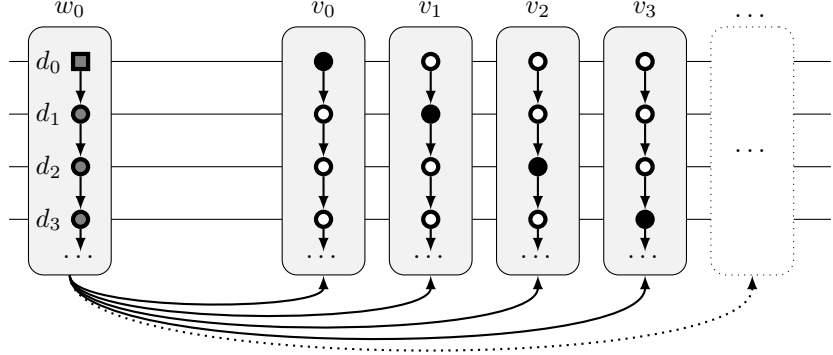
\begin{figure}[t]
\centering%
\begin{tikzpicture}[>=latex, xscale=0.95, yscale=0.7,
 	nd/.style={ultra thick,circle,draw,minimum size=2mm,inner sep=0pt},
	wn/.style={rectangle,rounded corners=2mm,draw,minimum width=11mm,minimum height=33mm}]
\draw (-2,0) -- +(11.5,0);
\draw (-2,-1) -- +(11.5,0);
\draw (-2,-2) -- +(11.5,0);
\draw (-2,-3) -- +(11.5,0);
%
\node[wn, fill=gray!10,label=above:{$w_0$}] (w0) at (-1.15,-1.7) {};	
%
\node[nd,fill=gray,rectangle,label=left:{$d_0$}] (c0) at (-1,0) {};
\node[nd,fill=gray,label=left:{$d_1$}] (c1) at (-1,-1) {};
\node[nd,fill=gray,label=left:{$d_2$}] (c2) at (-1,-2) {};
\node[nd,fill=gray,label=left:{$d_3$}] (c3) at (-1,-3) {};
\node[inner sep=1pt]  (c4) at (-1,-3.7) {\dots};
\begin{scope}[thick]
\draw[->] (c0) -- (c1);
\draw[->] (c1) -- (c2);
\draw[->] (c2) -- (c3);
\draw[->] (c3) -- (c4);
\end{scope}
%
\node[wn, fill=gray!10,label=above:{$v_0$}] (wp) at (2.4,-1.7) {};	
\node[wn, fill=gray!10,label=above:{$v_1$}] (w1) at (3.9,-1.7) {};	
\node[wn, fill=gray!10,label=above:{$v_2$}] (w2) at (5.4,-1.7) {};	
\node[wn, fill=gray!10,label=above:{$v_3$}] (w3) at (6.9,-1.7) {};	
\node[wn,thin,fill=white,dotted,label=above:{$\dots$}] (w4) at (8.4,-1.7) {\dots};
\begin{scope}[out=-90,in=-90,looseness=0.5,thick]
\draw[->] (w0.south) to (wp.south);
\draw[->] (w0.south) to (w1.south);
\draw[->] (w0.south) to (w2.south);
\draw[->] (w0.south) to (w3.south);
\draw[->,dotted] (w0.south) to (w4.south);
\end{scope}
%
\node[nd,fill=black] (c0-0) at (2.4,0) {};
\node[nd,fill=white] (c1-0) at (2.4,-1) {};
\node[nd,fill=white] (c2-0) at (2.4,-2) {};
\node[nd,fill=white] (c3-0) at (2.4,-3) {};
\node[inner sep=1pt]  (c4-0) at (2.4,-3.7) {\dots};
\begin{scope}[thick]
\draw[->] (c0-0) -- (c1-0);
\draw[->] (c1-0) -- (c2-0);
\draw[->] (c2-0) -- (c3-0);
\draw[->] (c3-0) -- (c4-0);
\end{scope}
%
\node[nd,fill=white] (c0-1) at (3.9,0) {};
\node[nd,fill=black] (c1-1) at (3.9,-1) {};
\node[nd,fill=white] (c2-1) at (3.9,-2) {};
\node[nd,fill=white] (c3-1) at (3.9,-3) {};
\node[inner sep=1pt]  (c4-1) at (3.9,-3.7) {\dots};
\begin{scope}[thick]
\draw[->] (c0-1) -- (c1-1);
\draw[->] (c1-1) -- (c2-1);
\draw[->] (c2-1) -- (c3-1);
\draw[->] (c3-1) -- (c4-1);
\end{scope}
%
\node[nd,fill=white] (c0-2) at (5.4,0) {};
\node[nd,fill=white] (c1-2) at (5.4,-1) {};
\node[nd,fill=black] (c2-2) at (5.4,-2) {};
\node[nd,fill=white] (c3-2) at (5.4,-3) {};
\node[inner sep=1pt]  (c4-2) at (5.4,-3.7) {\dots};
\begin{scope}[thick]
\draw[->] (c0-2) -- (c1-2);
\draw[->] (c1-2) -- (c2-2);
\draw[->] (c2-2) -- (c3-2);
\draw[->] (c3-2) -- (c4-2);
\end{scope}
%
\node[nd,fill=white] (c0-3) at (6.9,0) {};
\node[nd,fill=white] (c1-3) at (6.9,-1) {};
\node[nd,fill=white] (c2-3) at (6.9,-2) {};
\node[nd,fill=black] (c3-3) at (6.9,-3) {};
\node[inner sep=1pt]  (c4-3) at (6.9,-3.7) {\dots};
\begin{scope}[thick]
\draw[->] (c0-3) -- (c1-3);
\draw[->] (c1-3) -- (c2-3);
\draw[->] (c2-3) -- (c3-3);
\draw[->] (c3-3) -- (c4-3);
\end{scope}
\end{tikzpicture}
\caption{An infinitely branching interpretation satisfying $\varphi$ in Example~\ref{ex:running-example:1}.}\label{fig:running-example:1}
\end{figure}
Let $\varphi_0$ be a $\CT$-sentence that  has only infinite models,
e.g.,
\begin{equation*}
\forall y\exists^{=1}z\,P(y,z) \ \ \land \ \ \forall y\exists^{\leq 1}z\,P(z,y) \ \ \land \ \ \exists y\, \lnot \exists z\,P(z,y). 
\end{equation*}
Then consider the following $\CtwomonML$-sentence:
\begin{equation*}
\varphi \quad = \quad \varphi_0 \ \ \wedge \ \ \forall y\, \Diamond_a A(y) \ \ \wedge \ \ \Box_a \exists^{\leq 1}y\, A(y).
\end{equation*}
Figure~\ref{fig:running-example:1} illustrates an interpretation
satisfying $\varphi$: in world $w_0$,  sentence~$\varphi_0$ forces an
infinite chain $d_0,d_1,d_2,\dots$ of domain elements connected
by~$P$: every element in the chain has a unique $P$-successor, and every element but $d_0$ has a
unique $P$-predecessor, while $d_0$ has no $P$-predecessors. By the remaining two conjuncts of $\varphi$,
for each $d_i$, $i\in\mathbb{N}$, world~$w_0$ has an $R_a$-successor
world $v_i$ that contains one and only one element in~$A$, thus
resulting in infinite branching.
\end{example}


%% file: 3_counting.tex

\subsection{Non-Rigid Constants and Difference Operator}

We start by introducing $\MLDiff$, the language obtained from $\QoneMLc$ by replacing constants by the ``elsewhere'' quantifier. It was introduced and investigated in~\cite{DBLP:conf/aiml/HampsonK12,HamKur15,DBLP:conf/aiml/Hampson16}. Formulas in $\MLDiff$ are defined by the following grammar:
\begin{align*}
\varphi \ \ ::= \ \ & \
P(x)
\ \mid \
\neg \varphi \mid (\varphi \land \varphi) \ \mid \ \exists x\, \varphi \ \mid \ 
\exists^{\neq} x\, \varphi
 \ \mid \ \Diamond_{a} \varphi
\end{align*}
for unary
$P \in \Pred$
and a single fixed $x \in \Var$, where $\exists^{\neq}$ is called an \emph{elsewhere} or \emph{difference} \emph{quantifier}.
The semantics is given as in Section~\ref{sec:syntax}, with  the following additional clause:
\begin{multline*}
		\Mmf, w  \models^{\assign} \exists^{\neq} x\, \varphi \quad \text{iff}  \quad
		\Mmf, w \models^{\assign'} \varphi, \ \text{for some $x$-variant $\assign'$ of $\assign$ at $w$}\\[-2pt]\text{\emph{different} from $\assign$.} 
\end{multline*}
We emphasise that, unlike $\exists x\,\varphi$,  the truth value of $\exists^{\neq} x\, \varphi$ depends on $\assign(x)$, and so $x$ is \emph{free} in $\exists^{\ne}x\, \varphi$; \revised{R1.12}{in this respect,  $\exists^{\neq} x\, \varphi$ is not a true quantifier.}
%
%
%

%
\begin{theorem}\label{th:diff}
Let $\Cmc$ be any class of frames.
$\Cmc$-validity in $\QoneMLc$ and $\MLDiff$ are mutually polytime-reducible, with both constant and expanding domains.
The same applies to global $\Cmc$-consequence.
\end{theorem}
\begin{proof}
The proof lifts to our setting the results by
Gargov and Goranko~\cite[Section~4.1]{DBLP:journals/jphil/GargovG93} for propositional modal logic,
which show that the difference modality has the same expressive power as nominals with the universal modality.

\paragraph{Reasoning in $\QoneMLc$ is polytime-reducible to reasoning in $\MLDiff$}
Without loss of generality, we assume constants occur only in atoms of the form~$x = c$: replacing each $P(c)$  by $\exists x\,\bigl(P(x)\land (x = c)\bigr)$, each $c = d$ by $\exists x\,\bigl((x = c) \land (x = d)\bigr)$ and each $c = x$ by $x = c$ results in an equivalent formula of linear size.
We also consider only sentences, as validity of formulas coincides with validity of their universal closures.

Given a $\QoneMLc$-sentence $\varphi$, we take a fresh unary predicate $Q_{c}$, for every $c \in \Ind(\varphi)$, and define $\varphi^{\dagger}$ as the $\MLDiff$-sentence obtained by 
replacing every 
$x = c$  by
$Q_{c}(x) \land \lnot \exists^{\neq} x \, Q_{c}(x)$.
For a $\QoneMLc$ theory $\Gamma$, we
set 
$\Gamma^{\dagger} =  \{ \gamma^{\dagger} \mid \gamma \in \Gamma \}$.
We
now show that $(\varphi, \Gamma)\leq_{\Cmc}(\varphi^\dagger, \Gamma^\dagger)$ and $(\varphi^\dagger, \Gamma^\dagger)\leq_{\Cmc}(\varphi, \Gamma)$.

  $(\Rightarrow)$ Given $\Mmf$,
  we extend it to $\Mmf^\dagger$ that
  additionally interprets, for every~$w$,
  \begin{equation*}
    Q_{c}^{\smash{\Mmf^\dagger(w)}} =
    \begin{cases}
      \{ c^{\Mmf(w)} \}, & \text{if $c^{\Mmf(w)}$ is defined}, \\
      \emptyset, & \text{otherwise}.
    \end{cases}
  \end{equation*}
It is enough to show that, for any $\QoneMLc$-formula $\psi$, we have $\Mmf, w \models^{\assign} \psi$ iff
  $\Mmf^\dagger, w \models^{\assign} \psi^{\dagger}$, for every assignment~$\assign$ at every world $w$. We proceed  by induction on the structure of $\psi$.   For the base case of $x = c$, we have
  $\Mmf, w \models^{\assign} x = c$ iff $c^{\Mmf(w)}$ is defined and
  $\assign(x) = c^{\Mmf(w)}$.
  By the definition of $\Mmf^\dagger$,
  this is equivalent to
  $\Mmf^\dagger, w \models^{\assign} Q_{c}(x) \land \lnot \exists^{\neq} x \,
  Q_{c}(x)$, meaning $\Mmf^\dagger, w \models^{\assign} (x = c)^\dagger$.
  The remaining base case of $P(x)$ is trivial, and the inductive
  cases follow by a straightforward application of the induction
  hypothesis.  

  $(\Leftarrow)$ Given $\Mmf^\dagger$,
  we define $\Mmf$ as an extension of $\Mmf^\dagger$ such that, for
  every~$w$,\revised{R2.8}{}
  \begin{equation*}
    c^{\Mmf(w)} \ \text{is} \
    \begin{cases}
      d , & \text{if $Q_{c}^{\smash{\Mmf^\dagger(w)}} = \{ d \}$, for some $d \in \Delta_{w}$}, \\
      \text{undefined}, & \text{otherwise}.
    \end{cases}
  \end{equation*}
  As above, for any $\QoneMLc$-formula $\psi$, it can be seen that
  $\Mmf^\dagger, w \models^{\assign} \psi^{\dagger}$ iff
  $\Mmf,w \models^{\assign} \psi$, for every assignment $\assign$ at
  every world $w$. 

\paragraph{Global $\Cmc$-consequence in $\MLDiff$ is polytime-reducible to that in $\QoneMLc$}
Given a $\MLDiff$-sentence
$\varphi$, we recursively replace every subformula of the form $\exists^{\neq} x\, \psi$, where $\psi$ does not contain other occurrences of $\exists^{\neq}x$, by $(\exists^{\neq} x\, \psi)^\ddagger$ defined by
\begin{equation*}
(\exists^{\neq} x\, \psi)^\ddagger \ \ = \ \ \exists x \, P_\psi(x) \ \land \ \bigl(x = c_{\psi} \to \exists x\, \bigl(\lnot (x = c_{ \psi}) \land  P_\psi(x)\bigr)\bigr),
\end{equation*}
where $c_{\psi}$ is a fresh constant and $P_\psi$ is a fresh unary predicate symbol
(which plays the role of a surrogate for $\psi$ and is used to avoid an exponential blow-up in the transformation), and denote the resulting $\QoneMLc$-sentence
by $\varphi^{\ddagger}$; \revised{R1.12}{observe that $\exists^{\neq} x\, \psi$ and  $(\exists^{\neq} x\, \psi)^\ddagger$ have the same free variable, $x$.}
Moreover, we denote by \revised{R2.9}{$\Xi_{\varphi}$}
the  set of $\QoneMLc$-sentences $\singl$ defined by
\begin{equation*}
\singl \ \ = \ \ \forall x\,\bigl(\psi(x) \leftrightarrow P_\psi(x)\bigr) \ \ \land \ \ \bigl(\exists x\, P_\psi(x) \to
P_\psi(c_{\psi})
\bigr),
\end{equation*}%
for each subformula $\exists^{\neq} x \, \psi$ replaced in the
transformation of $\varphi$.

\begin{claim}\label{cla:difftoonecon}
For any $\MLDiff$ formula $\varphi$ and any $\Mmf^\ddagger$ such that $\Mmf^\ddagger \models \Xi_{\varphi}$, we have $\Mmf^\ddagger, w \models^{\assign} \varphi$ iff $\Mmf^\ddagger, w \models^{\assign} \varphi^{\ddagger}$, for 
any world $w$ of~$\Mmf^\ddagger$ and assignment $\assign$ at $w$.
\end{claim}
\begin{proof}
The proof is by induction on the structure of $\varphi$: we show only the case of $\exists^{\neq} x\, \psi$, where $\psi$ does not contain other occurrences of $\exists^{\neq} x$, since the others follow by straightforward applications of the inductive hypothesis.

($\Rightarrow$) Suppose $\Mmf^\ddagger, w \models^{\assign} \exists^{\neq} x\, \psi$, that is, $\Mmf^\ddagger, w \models^{\assign'} \psi(x)$, for some $x$-variant~$\assign'$ of~$\assign$ at~$w$ different from $\assign$. 
Since $\Mmf^\ddagger, w \models \singl$, 
we have both
$\Mmf^\ddagger, w \models^{\assign'} P_{\psi}(x)$
and $\Mmf^\ddagger, w \models^{\assign} \exists x\, P_{\psi}(x)$. The latter is the first conjunct of $(\exists^{\ne} x\,\psi)^\ddagger$.
For its second conjunct, suppose $\Mmf^{\ddagger}, w \models^{\assign} x = c_{\psi}$. So, $c_{\psi}$ designates at $w$, and\linebreak $\assign(x) = \smash{c_{\psi}^{\smash{\Mmf^{\ddagger}(w)}}}$.
Since $\assign'(x) \neq \assign(x) = \smash{c_{\psi}^{\smash{\Mmf^{\ddagger}(w)}}}$ and $\assign'(x)  \in \smash{P_{\psi}^{\smash{\Mmf^{\ddagger}(w)}}}$, the $x$-variant $\assign'$  of $\assign$ ensures $\Mmf^{\ddagger}, w \models^{\assign} \exists x \, \bigl(\lnot (x = c_{\psi}) \land P_{\psi}(x)\bigr)$.
Thus,
$\Mmf^\ddagger, w \models^{\assign} (\exists^{\neq} x\,\psi)^\ddagger$.

($\Leftarrow$)
Suppose $\Mmf^{\ddagger}, w \models^{\assign} (\exists^{\neq} x\,\psi)^\ddagger$.
From the first conjunct of $(\exists^{\neq} x\,\psi)^\ddagger$, $\smash{P_{\psi}^{\smash{\Mmf^{\ddagger}(w)}}} \neq \emptyset$.
So, by the second conjunct of $\singl$, we have
$\Mmf^{\ddagger}, w \models P_{\psi}(c_{\psi})$.
We distinguish two cases.
If $\assign(x) \neq \smash{c_{\psi}^{\smash{\Mmf^{\ddagger}(w)}}}$, then,
by the first conjunct of $\singl$, we obtain $\Mmf^{\ddagger}, w \models^{\assign} \exists^{\neq} x\, \psi$.
Otherwise, $\assign(x) = \smash{c_{\psi}^{\smash{\Mmf^{\ddagger}(w)}}}$,
and so $\Mmf^{\ddagger}, w \models^{\assign} x = c_{\psi}$.
By the second conjunct of $(\exists^{\neq} x \, \psi)^\ddagger$, there exists $d \in \smash{P_{\psi}^{\smash{\Mmf^{\ddagger}(w)}}}$ with $d \neq \assign(x)$.
So, by the first conjunct of $\singl$, we get $\Mmf^{\ddagger}, w \models^{\assign} \exists^{\neq} x\, \psi$.
\end{proof}

Now,
given a $\MLDiff$-theory $\Gamma$, we take 
\revised{R2.9}{$\Gamma^{\ddagger} = \{ \gamma^\ddagger \mid \gamma \in \Gamma\}$ and $\Xi_\Gamma = \smash{\bigcup_{\gamma\in\Gamma} \Xi_{\gamma}}$}
and
show that \revised{R2.9}{$(\varphi^\ddagger, \Gamma^\ddagger\cup \Xi_\Gamma \cup \Xi_\varphi)$} is a $\Cmc$-model-conservative extension of~$(\varphi,\Gamma)$.

$(\Rightarrow)$ 
Let $\Mmf, v \models \varphi$ and $\Mmf \models \Gamma$.
We extend $\Mmf$ to $\Mmf^\ddagger$ 
by interpreting predicates $\smash{P_{\psi}}$ and constants~$\smash{c_{\psi}}$, for subformulas $\exists^{\neq} x \, \psi$ of $\varphi$ and $\Gamma$: for every~$w$, we take
$\smash{P_{\psi}^{\smash{\Mmf^\ddagger(w)}}} = \bigl\{ d \in \Delta_{w} \bigm| \Mmf, w \models \psi[d] \bigr\}$
and
let $\smash{c_{\psi}^{\smash{\Mmf^\ddagger(w)}}}$ be an arbitrary element of $\smash{P_{\psi}^{\smash{\Mmf^\ddagger(w)}}}$, if $\smash{P_{\psi}^{\smash{\Mmf^\ddagger(w)}}} \neq \emptyset$, and let $\smash{c_{\psi}^{\smash{\Mmf^\ddagger(w)}}}$ be undefined, otherwise.
By the definition of $\Mmf^\ddagger$, we have $\Mmf^\ddagger \models \Xi_{\varphi}\cup \Xi_\Gamma$, whence, by Claim~\ref{cla:difftoonecon}, $\Mmf^\ddagger, v \models \varphi^{\ddagger}$ and $\Mmf^\ddagger \models \gamma^\ddagger$, for each $\gamma\in\Gamma$, and so
 $\Mmf^\ddagger \models \Gamma^{\ddagger}$.

$(\Leftarrow)$ 
Let 
$\Mmf^\ddagger, v \models \varphi^{\ddagger}$ and
$\Mmf^\ddagger \models\Gamma^{\ddagger}\cup \Xi_\Gamma \cup \Xi_{\varphi}$.
By Claim~\ref{cla:difftoonecon}, we obtain $\Mmf^\ddagger \models \Gamma$
and $\Mmf^\ddagger, v \models \varphi$.

\paragraph{$\Cmc$-validity in $\MLDiff$ is polytime-reducible to that in $\QoneMLc$}
We modify the proof above by carefully selecting sequences of boxes to avoid an exponential blowup of the formula.
To this end, for a sentence $\varphi$ and a subformula $\psi$ of~$\varphi$, we define the set of \emph{$\psi$-relevant paths in $\varphi$}, denoted by $\rpath(\varphi, \psi)$, as the set of sequences $a_1,\dots, a_n$ of the modalities under which $\psi$ occurs in $\varphi$.
For instance,
for $\varphi = \exists x\,\bigl(\Diamond_{1} \lnot P(x) \land \Diamond_{2}\Diamond_{3} P(x)\bigr)$, we have
$\rpath(\varphi, P(x)) = \{\langle1\rangle, \langle2, 3\rangle \}$ and $\rpath(\varphi, \lnot P(x)) = \{\langle 1\rangle \}$.
Clearly, the length of paths in $\rpath(\varphi, \psi)$ is at most~$\md(\varphi)$. Moreover,
for any path $\pi$,
we recursively define
$\Box^{\epsilon} \chi  = \chi  \text{ and } \Box^{a \cdot \pi} \chi =  \Box_{a} \Box^{\pi} \chi$,
where $\epsilon$ denotes the empty path.
Then, we can again replace any occurrence of 
$\exists^{\neq} x \, \psi$ in~$\varphi$ not containing another $\exists^{\neq}$
by~$(\exists^{\neq} x \, \psi)^\ddagger$,
while adding to $\varphi$ the conjunct
$\Box^{\pi} \singl$
for every $\pi \in \rpath(\varphi, \exists^{\neq} x\,\psi)$.
As above, by repeatedly applying this procedure in a bottom-up way, we obtain in polynomial time a sentence
$\varphi^{\ddagger}$ that is a $\Cmc$-model-conservative extension of $\varphi$.
\end{proof}


%% file: 4.1_quasimodel.tex


\section{Quasimodels for $\QmonMLc$}\label{sec:quasimodels}

In this section we present a straightforward generalisation of the standard quasimodel technique~(see, e.g., \cite{GabEtAl03}) to the languages with constants, equality and/or counting quantifiers. 

For every monodic formula of the form $\Diamond_a\psi(x)$ with one free variable $x$,
we reserve a unary predicate $R_{\Diamond_a\psi}(x)$, and, for every monodic sentence
of the form $\Diamond_a\psi$, a propositional variable $p_{\Diamond_a\psi}$.
Symbols $R_{\Diamond_a\psi}(x)$ and $p_{\Diamond_a\psi}$ are called the \emph{surrogates}
of~$\Diamond_a\psi(x)$ and~$\Diamond_a\psi$, respectively.  For clarity of
presentation, we assume that these surrogates are {\em not} in the
original signature.

Given a  $\QmonMLc$-formula $\varphi$, we denote by
$\overline{\varphi}$ 
the result of replacing all its
subformulas of the form $\Diamond_a\psi(x)$ or~$\Diamond_a\psi$, which are not in the scope of another modal operator,
by their surrogates. Thus,
$\overline{\varphi}$ contains no occurrences of modal operators at
all.  Observe that, for all
monodic formulas $\psi_1$ and $\psi_2$, we have
\begin{equation*}
	\overline{\psi_1\wedge\psi_2}=\overline{\psi_1}\wedge\overline{\psi_2},
	\qquad
	\overline{\neg\psi_1}=\neg\overline{\psi_1} \quad \mbox{ and }\quad
	\overline{\exists x\,\psi_1}=\exists x\,\overline{\psi_1}.
\end{equation*}
Let $x$ be a variable not occurring in $\varphi$. Put
\begin{equation*}
\sub[x]{\varphi} = \bigl\{ \psi\{x/y\}, \neg \psi\{x/y\} \mid \psi(y)\in \sub[]{\varphi}\bigr\},
\end{equation*}
where, as usual, we write $\psi(y)$ to indicate that $\psi$ has at most one free variable,~$y$, and $\psi\{x/y\}$ denotes the result of replacing all free occurrences of $y$ in $\psi$ with~$x$.
Clearly, $|\sub[x]{\varphi}| \leq 2|\varphi|$.
By a \emph{type}
for  a $\QmonMLc$-sentence $\varphi$ we mean any Boolean-saturated subset $\contp$ of $\sub[x]{\varphi}$, that is,
\begin{itemize}
	\item $\psi_1\wedge\psi_2\in \contp$ iff
	$\psi_1\in \contp$ and $\psi_2\in \contp$, for every
	$\psi_1\wedge\psi_2\in \sub[x]{\varphi}$;
	\item $\neg\psi\in \contp$ iff $\psi\notin \contp$,
	for every $\neg\psi\in \sub[x]{\varphi}$.
\end{itemize}
For a type $\contp$ for $\varphi$, we write $\overline{\contp}$ to denote the set $\{\overline{\psi} \mid \psi \in\contp\}$. We will often identify such a set with the conjunction of formulas in it.

We use quasistates  to describe worlds in interpretations. With the equality in the language, the description of each world needs to specify not only the types realised in that world, but also the number of domain elements that realise the type. To deal with this issue, multisets of types are useful. 

A multiset $\avec{X}$ is a pair, $(X, \textit{card}_X)$,
consisting of an underlying set, $X$, and a function,
$\textit{card}_X$ that assigns to each element its multiplicity. To
represent infinite domains of interpretations, the multiplicities need
to range over the set $\extN=\mathbb{N}\cup\{\aleph_0\}$, with the
usual assumptions that $\aleph_0 + m = m + \aleph_0 = \aleph_0$, for
each $m \in \mathbb{N}$, and $\aleph_0 + \aleph_0 = \aleph_0$. Note
that we assume $0 \in\mathbb{N}$, and so $\textit{card}_X(x)$ can in
principle be $0$, indicating that $x$ does \emph{not} belong to the multiset
$\avec{X}$. Equally, we can always take a larger underlying set and set $\textit{card}_X(x) = 0$ for any element $x$ outside the original underlying set. In the sequel, we often omit the underlying set, assuming that it is clear from the context, and slightly abuse notation by identifying a multiset~$\avec{X}$ with its multiplicity function $\textit{card}_X$ and write $\avec{X}(x)$ for $\textit{card}_X(x)$. As a shortcut, we also write $x\in \avec{X}$ whenever $\avec{X}(x) > 0$ and so that, for example, a multiset~$\avec{X}$ is non-empty if there is $x\in \avec{X}$. Given a multiset $\avec{X}$, we denote by $|\avec{X}|$ its cardinality, that is, $\sum_{x\in X} \avec{X}(x)$.

Formally, a \emph{quasistate candidate for $\varphi$} is a non-empty multiset $\avec{n}$ of types for~$\varphi$. We say that $\avec{n}$ is \emph{realised in} a  first-order (non-modal) structure
$\Bmf$ if
\begin{equation*}
\avec{n}(\contp)= |\{ b\in \Bmf \mid \Bmf\models \overline{\contp}[b]\}|, \text{ for all types } \contp \text{ for } \varphi.
\end{equation*}
If $\Bmf\models \overline{\contp}[b]$, then we say that $\contp$ is \emph{realised by $b$ in $\Bmf$}, and if $\avec{n}$ is realised in some~$\Bmf$, it is also said to be \emph{realisable} (without referring to~$\Bmf$), and  a realisable quasistate candidate is simply called a \emph{quasistate for $\varphi$}.
Note that the realisability condition allows us to deal with inconsistency and constants in a uniform way: for example, if a type $\contp$ contains $\forall y\,A(y)$ and $\neg A(x)$, then $\avec{n}(\contp) = 0$, for any quasistate $\avec{n}$; also, 
if $x = c$, for $c \in \Ind$, is a subformula of~$\varphi$, then, any
quasistate~$\avec{n}$ will contain a type~$\contp$ with $(x = c) \in
\contp$ and $\avec{n}(\contp) = 1$ because~$\forall x\,(x
\ne c)$ is unsatisfiable (we assume first-order structures
interpret all constant symbols) but no other type $\contp'$ with $\avec{n}(\contp') > 0$ can contain $x = c$  because only one element,~$c^{\Bmf}$, can satisfy $x = c$ in any $\Bmf$.

Let $\Fmf=(W,\{R_a\}_{a\in A})$ be a frame.
A \emph{basic structure for $\varphi$ based on~$\Fmf$} is a pair~$(\Fmf, \funcand)$, where
$\funcand$ is a function associating with
every $w \in W$ a quasistate~$\funcand(w)$ for~$\varphi$.
In the sequel, we write $\funcand(w, \contp)$  for the multiplicity $\funcand(w)(\contp)$ of type $\contp$ in~$\funcand(w)$. 

We say that $W'\subseteq W$ is \emph{upward-closed} if $W'$ contains all $v$ with $wR_av$ for $w\in W'$ and $a \in A$. In any expanding-domains interpretation $(\Fmf,\Delta,\cdot)$, the set of worlds where a domain element $d$ exists, $\{w\in W\mid d\in \Delta_w \}$, is upward-closed. Runs in quasimodels correspond to domain elements, which motivates the following definition. 
A \emph{run through $(\Fmf, \funcand)$} is a function $\rho$ mapping each world~$w$ in an upward-closed subset $W'$ of $W$ to a type $\rho(w)\in \funcand(w)$
satisfying
the following coherence and saturation conditions for every $w\in W'$:
\begin{enumerate}
	[label=\textbf{(r-coh)},leftmargin=*,series=run]
	\item $\Diamond_{a} \psi \in \rho(w)$ if there exists $v \in W$ with
	$wR_{a}v$ and $\psi \in \rho(v)$, for every $\Diamond_{a} \psi \in \sub[x]{\varphi}$;
	\label{rn:modal}
\end{enumerate}
\begin{enumerate}
	[label=\textbf{(r-sat)},leftmargin=*,series=run]
	\item if $\Diamond_{a} \psi \in \rho(w)$, then there exists $v \in W$ with
	$wR_{a}v$ and $\psi \in \rho(v)$, for every $\Diamond_{a} \psi \in \sub[x]{\varphi}$.
	\label{rn:modal2}
\end{enumerate}
The domain, $W'$, of $\rho$ is denoted by $\dom \rho$. A run $\rho$ is called \emph{full} if $\dom \rho = W$.

In our quasimodels we will need to count not only the number of times a type is realised at a certain world in an interpretation, but also the number of domain elements that give rise to the same run. 
In other words, instead of sets of runs in standard quasimodels, we use \emph{multisets} of runs. The following notation will be useful in the sequel. Let $\runs$ be a multiset  of runs through a basic structure $(\Fmf, \funcand)$. For a world $w\in W$ in $\Fmf$, we denote by $\runs_{w}$ the \emph{$w$-slice of $\runs$}, that is, the multiset of runs $\rho$ in $\runs$ such that $w\in\dom\rho$, with each run in the $w$-slice having the same multiplicity as in $\runs$: formally, 
\begin{equation*}
\runs_{w}(\rho) = \begin{cases}\runs(\rho), & \text{ if } w\in\dom\rho, \\ 0, &\text{ otherwise.}\end{cases}
\end{equation*}
Similarly, for a world $w\in W$ and a type $\contp$ for $\varphi$, we denote by $\runs_{w, \contp}$  the \emph{$(w, \contp)$-slice of $\runs$}, the multiset of runs $\rho$ in $\runs$ such that $w\in\dom\rho$  and $\rho(w) = \contp$:
\begin{equation*}
\runs_{w,\contp}(\rho) = \begin{cases}\runs(\rho), & \text{ if } w\in\dom\rho \text{  and } \rho(w) = \contp, \\ 0, &\text{ otherwise.}\end{cases}
\end{equation*}
An \emph{\textup{(}expanding-domains\textup{)} quasimodel for $\varphi$ based on $\Fmf$} is a triple $\quasimod =
(\Fmf, \funcand, \runs)$, where $(\Fmf, \funcand)$ is a basic structure for
$\varphi$ based on $\Fmf$ and $\runs$ is a multiset of runs through $(\Fmf, \funcand)$ such that
\begin{enumerate}
	[label=\textbf{(card)},leftmargin=*,series=run]
	\item\label{run:exists} $\funcand(w,\contp)= |\runs_{w,\contp}|$, for every $w \in W$ and every type $\contp$ for $\varphi$.
\end{enumerate}
We say that $\quasimod$ is a \emph{constant-domain quasimodel} if $\runs$  consists of full runs.  
A quasimodel $\quasimod =
(\Fmf, \funcand, \runs)$ is said to \emph{satisfy $\varphi$} if 
\begin{enumerate}[label=\textbf{(b)},leftmargin=*,series=run]
\item $\varphi \in \contp$ for some $w_0\in W$ and some type $\contp\in\funcand(w_0)$.
	\label{b1}
\end{enumerate}  
%
%
%

In the sequel it will often be  convenient to represent a multiset $\runs$ of runs  as a set $\hat{\runs}$ of \emph{indexed runs}:
\begin{equation*}
\hat{\runs} = \{ (\rho,\ell) \in \runs \times \mathbb{N} \mid 0 \leq \ell < \runs(\rho) \}; 
\end{equation*}
note that if $\rho$ has infinite  multiplicity ($\aleph_0$) in $\runs$, then by definition $\hat{\runs}$ contains all pairs of the form $(\rho, \ell)$, for $\ell \in \mathbb{N}$. 

\begin{example}\label{ex:running:quasimodel}\em
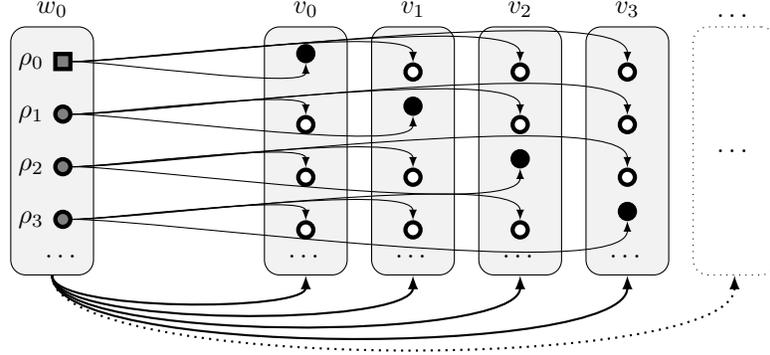
\begin{figure}[t]
\centering%
\begin{tikzpicture}[>=latex, xscale=0.95, yscale=0.7,
 	nd/.style={ultra thick,circle,draw,minimum size=2mm,inner sep=0pt},
	wn/.style={rectangle,rounded corners=2mm,draw,minimum width=11mm,minimum height=33mm}]
%
%
%
\node[wn, fill=gray!10,label=above:{$w_0$}] (w0) at (-1.15,-1.7) {};	
%
\node[nd,fill=gray,rectangle,label=left:{$\rho_0$}] (c0) at (-1,0) {};
\node[nd,fill=gray,label=left:{$\rho_1$}] (c1) at (-1,-1) {};
\node[nd,fill=gray,label=left:{$\rho_2$}] (c2) at (-1,-2) {};
\node[nd,fill=gray,label=left:{$\rho_3$}] (c3) at (-1,-3) {};
\node (c4) at (-1,-3.7) {\dots};
\node[wn, fill=gray!10,label=above:{$v_0$}] (v0) at (2.4,-1.7) {};	
\node[wn, fill=gray!10,label=above:{$v_1$}] (v1) at (3.9,-1.7) {};	
\node[wn, fill=gray!10,label=above:{$v_2$}] (v2) at (5.4,-1.7) {};	
\node[wn, fill=gray!10,label=above:{$v_3$}] (v3) at (6.9,-1.7) {};	
\node[wn,thin,fill=white,dotted,label=above:{$\dots$}] (v4) at (8.4,-1.7) {\dots};
\begin{scope}[out=-90,in=-90,looseness=0.5,thick]
\draw[->] (w0.south) to (v0.south);
\draw[->] (w0.south) to (v1.south);
\draw[->] (w0.south) to (v2.south);
\draw[->] (w0.south) to (v3.south);
\draw[->,dotted] (w0.south) to (v4.south);
\end{scope}
%
\node[nd,fill=black] (c0-0) at (2.4,0.15) {};
\node[nd,fill=white] (c1-0) at (2.4,-1.2) {};
\node[nd,fill=white] (c2-0) at (2.4,-2.2) {};
\node[nd,fill=white] (c3-0) at (2.4,-3.2) {};
\node (c4-1) at (2.4,-3.7) {\dots};
%
\node[nd,fill=white] (c0-1) at (3.9,-0.2) {};
\node[nd,fill=black] (c1-1) at (3.9,-0.85) {};
\node[nd,fill=white] (c2-1) at (3.9,-2.2) {};
\node[nd,fill=white] (c3-1) at (3.9,-3.2) {};
\node (c4-1) at (3.9,-3.7) {\dots};
%
\node[nd,fill=white] (c0-2) at (5.4,-0.2) {};
\node[nd,fill=white] (c1-2) at (5.4,-1.2) {};
\node[nd,fill=black] (c2-2) at (5.4,-1.85) {};
\node[nd,fill=white] (c3-2) at (5.4,-3.2) {};
\node (c4-2) at (5.4,-3.7) {\dots};
%
\node[nd,fill=white] (c0-3) at (6.9,-0.2) {};
\node[nd,fill=white] (c1-3) at (6.9,-1.2) {};
\node[nd,fill=white] (c2-3) at (6.9,-2.2) {};
\node[nd,fill=black] (c3-3) at (6.9,-2.85) {};
\node (c4-3) at (6.9,-3.7) {\dots};
\begin{scope}[looseness=0.4]
\draw[->,out=0,in=-90] (c0) to (c0-0.south);
\draw[->,out=0,in=90] (c0) to (c0-1.north);
\draw[->,out=0,in=90] (c0) to (c0-2.north);
\draw[->,out=0,in=90] (c0) to (c0-3.north);
\draw[->,out=0,in=90] (c1) to (c1-0.north);
\draw[->,out=0,in=-90] (c1) to (c1-1.south);
\draw[->,out=0,in=90] (c1) to (c1-2.north);
\draw[->,out=0,in=90] (c1) to (c1-3.north);
\draw[->,out=0,in=90] (c2) to (c2-0.north);
\draw[->,out=0,in=90] (c2) to (c2-1.north);
\draw[->,out=0,in=-90] (c2) to (c2-2.south);
\draw[->,out=0,in=90] (c2) to (c2-3.north);
\draw[->,out=0,in=90] (c3) to (c3-0.north);
\draw[->,out=0,in=90] (c3) to (c3-1.north);
\draw[->,out=0,in=90] (c3) to (c3-2.north);
\draw[->,out=0,in=-90] (c3) to (c3-3.south);
\end{scope}
\end{tikzpicture}%
\caption{A quasimodel for the interpretation in Fig.~\ref{fig:running-example:1}.}\label{fig:running:quasimodel}
\end{figure}
We continue with our Example~\ref{ex:running-example:1} and consider a quasimodel depicted in Fig.~\ref{fig:running:quasimodel}.
There are two types for $\varphi$ at $w_0$: $\contp_0$ contains
$\lnot \exists z\, P(z,x)$
and is shown by
a grey square, while $\contp_0'$ contains
$\exists z\, P(z,x)$
and is shown by grey circles; both types also include
$\exists^{=1} z\,P(x,z)$,
$\exists^{\leq 1} z\, P(z,x)$,
$\Diamond_a A(x)$, $\Box_a \exists^{\leq 1}
y\,A(y)$ and (the sentences equivalent modulo renaming of variables to) their universal closures along with suitable conjunctions, but no other positive formulas from~$\sub[x]{\varphi}$.
Then there are two types for $\varphi$ at each of the $v_i$, for $i\in\mathbb{N}$: $\contp_1$, depicted by black circles, contains $A(x)$ and $\exists^{\leq 1} y\,A(y)$, while $\contp_2$, depicted by white circles, contains $\neg A(x)$ and $\exists^{\leq 1} y\,A(y)$.

In the quasimodel in Fig.~\ref{fig:running:quasimodel}, quasistate $\funcand(w_0)$ contains one $\contp_0$ and $\aleph_0$-many~$\contp_0'$, while quasistate $\funcand(v_i)$, for $i\in\mathbb{N}$, contains one $\contp_1$ and $\aleph_0$-many~$\contp_2$. Each run $\rho_i$, for $i\in\mathbb{N}$, in the multiset $\runs$ has multiplicity~1 and is depicted as a tree with types in  quasistates being its vertices and the edges matching the accessibility relation of the frame: for each $i\in\mathbb{N}$, we have 
\begin{equation*}
\rho_i(w_0) = \begin{cases}\contp_0, & \text{if } i  = 0,\\  \contp_0', &\text{otherwise},\end{cases}\quad \rho_i(v_i) = \contp_1 \quad\text{ and }\quad \rho_i(v_j) = \contp_2, \text{ for } j \in \mathbb{N}\setminus \{ i\}.
\end{equation*} 
\end{example}

The next lemma provides a link between quasimodels and interpretations.

\begin{lemma}\label{lemma:quasimodel}
For both constant and expanding domains, a $\QmonMLc$-sentence~$\varphi$ is satisfiable in  a total interpretation $(\Fmf, \Delta, \cdot)$
iff there is a quasimodel $(\Fmf, \funcand, \runs)$ satisfying $\varphi$\textup{;} moreover,  $|\Delta_w| = |\runs_w|$, for each $w$ in $\Fmf$.
\end{lemma}
\begin{proof}
	$(\Rightarrow)$
	Let $\Mmf = (\Fmf, \Delta, \cdot)$ be a total interpretation based on a frame $\Fmf = (W, \{R_a\}_{a\in A})$ with $\Mmf, w_0\models \varphi$ for some $w_0\in W$. Let
\begin{equation*}
\contp^{\Mmf(w)}(d) = \bigl\{ \psi \in \sub[x]{\varphi} \bigm| \Mmf, w\models \psi[d] \bigr\}, 
\end{equation*}
for every $w \in W$ and $d \in \Delta_w$. Clearly,
	$\contp^{\Mmf(w)}(d)$ is Boolean-saturated and so is a type for $\varphi$. We now define a
	triple $\Qmf = (\Fmf, \funcand, \runs)$, where
	\begin{itemize}
		\item $\funcand$ is the function from $W$ to the set of
		quasistate candidates for $\varphi$ defined, for
		every $w\in W$, by setting
		$\funcand(w,\contp) = \min(\aleph_0, | \{ d \in \Delta_w \mid \contp^{\Mmf(w)}(d) = \contp \}|)$,  for every type $\contp$ for $\varphi$;
		\item $\runs$ is the multiset of functions $\rho_{d}$, 
		$d \in \bigcup_{w\in W}\Delta_w$, from upward-closed subsets of $W$ to
		the set of types for~$\varphi$ defined by
		$\rho_{d}(w)= \contp^{\Mmf(w)}(d)$, for
                each $w\in W$ with $d\in \Delta_w$. Note that if $k$ distinct $d_0,d_{1},\dots,d_{k-1}$ give rise to the same function $\rho$, that is, $\rho=\rho_{d_i}$, for $0\leq i < k$, then $\runs(\rho)=k$; if  infinitely many distinct $d_0,d_1,\dots$ give rise to the same function~$\rho$, that is, $\rho=\rho_{d_i}$, for all $i\in\mathbb{N}$, then $\runs(\rho)=\aleph_0$.
	\end{itemize}
	By the downward L\"owenheim-Skolem theorem,
	each $\funcand(w)$ is realisable.
	The elements of~$\runs$, by construction,
	satisfy~\mbox{\ref{rn:modal}} and~\ref{rn:modal2}. By the definition of $\funcand$ and~$\runs$, we have~\ref{run:exists}. Finally, $\Qmf$  satisfies $\varphi$
	as $\varphi\in \contp$, for all $\contp\in \funcand(w_0)$. 

\smallskip
	
\noindent $(\Leftarrow)$ Suppose there is a quasimodel
$\Qmf = (\Fmf, \funcand, \runs)$ satisfying $\varphi$. 
For each
$w\in W$, we consider
the set $\hat{\runs}_w$ of indexed runs in the $w$-slice of $\runs$. 
Since the
domains of runs are upward-closed, we have
$\hat{\runs}_w \subseteq \hat{\runs}_v$, for each $w,v\in W$ with $wR_av$, for
some $a\in A$.  By~\ref{run:exists}, we have
$|\hat{\runs}_w| = \sum_{\contp} \funcand(w,\contp)$, for each $w\in W$;
moreover, by the definition of quasistate,
$|\hat{\runs}_w| > 0$. So, each $\hat{\runs}_w$ is either countably infinite or
non-empty and finite, and $|\hat{\runs}_w| \leq |\hat{\runs}_v|$, for each
$w,v\in W$ with $wR_av$, for some $a\in A$. It then follows that, for
each $w\in W$, we can take a first-order structure~$\Bmf_w$ realising $\funcand(w)$
with the domain of cardinality $|\hat{\runs}_w|$. So,
for each $w\in W$, there is a bijection $f_w$ between the set $\hat{\runs}_w$ and the domain of~$\Bmf_w$ satisfying
\begin{equation}\label{eq:Bw}
\Bmf_w\models \overline{ \rho(w)}[f_w(\rho,\ell)], \text{ for every } w\in W \text{ and every } (\rho,\ell)\in\hat{\runs}_w.
\end{equation} 	
Define a total interpretation $\Mmf = (\Fmf, \Delta, \cdot)$, by setting
$\Delta_w = \hat{\runs}_w$, 
$c^{\Mmf(w)} = f^{-1}_w(c^{\Bmf_w})$, 
for any constant $c$ (recall that the\revised{R2.11}{} $\Bmf_w$ interpret all constant symbols),  
$P^{\Mmf(w)} = f^{-1}_w(P^{\Bmf_w})$,
for any $k$-ary
predicate symbol $P$, where $f^{-1}_w$
is extended to sets of $k$-tuples ($k \geq 0$) of indexed runs in a component-wise
way.
We show that, for each subformula $\psi$ of $\varphi$, we have
\begin{equation}\label{eq:quasimodel:induction}
\Bmf_w\models^{f_w \circ \mathfrak{b}} \overline{\psi} \quad \text{ iff }\quad
\Mmf, w\models^{\mathfrak{b}} \psi,
\end{equation}
for all assignments~$\mathfrak{b}$ and all $w\in W$, where
$f_w \circ \mathfrak{b}\colon x \mapsto f_w(\mathfrak{b}(x))$, for all
variables~$x$. The basis of induction, the case of atomic formulas, including equalities, is immediate from the definition (for equalities, we use the fact that $f_w$ is a bijection).
The cases of the Boolean connectives $\land$ and $\neg$ and the
quantifier $\exists$ easily follow from IH. 

So, it remains to consider the case of subformulas of the form $\Diamond_a\psi$.
First, note that $\psi$ contains at most one free variable, say, $x$
(for a sentence, pick any variable). Consider
$(\rho,\ell) = \mathfrak{b}(x)$ and $w\in \dom\rho$.  Suppose first
$\Bmf_w\models^{f_w\circ \mathfrak{b}} \overline{\Diamond_a\psi}$,
that is, $\Bmf_w \models \overline{\Diamond_a\psi}[f_w(\rho,\ell)]$.
By~\eqref{eq:Bw}, we have $\Diamond_a\psi\in\rho(w)$, whence,
by~\revised{R1.13}{\ref{rn:modal2}}, $\psi\in\rho(v)$, for some $v\in W$ with $wR_a
v$.
By~\eqref{eq:Bw}, we have $\Bmf_v\models
\overline{\rho(v)}[f_v(\rho,\ell)]$ and thus
$\Bmf_v\models \overline{\psi}[f_v(\rho,\ell)]$.  By IH,
$\Mmf, v\models \psi[(\rho,\ell)]$, whence
$\Mmf, w\models \Diamond_a\psi[(\rho,\ell)]$.
Conversely, suppose $\Mmf, w\models
\Diamond_a\psi[(\rho,\ell)]$. Then
$\Mmf, v\models \psi[(\rho,\ell)]$, for some $v\in W$ with $wR_av$. By IH, we
obtain $\Bmf_v\models \overline{\psi}[f_v(\rho,\ell)]$.
By~\eqref{eq:Bw}, we have $\psi\in\rho(v)$, whence, by~\revised{R1.14}{\ref{rn:modal}},
$\Diamond_a\psi\in\rho(w)$.
By~\eqref{eq:Bw}, we obtain
$\Bmf_w\models \overline{\Diamond_a\psi}[f_w(\rho,\ell)]$, as required.
	
	Now we can easily finish the proof of Lemma~\ref{lemma:quasimodel} for expanding domains by observing 
	that, by~\ref{b1}, there is a world $w_0\in W$ and a type
	$\contp$ with  $\funcand(w_0, \contp) > 0$ such that $\varphi\in \contp$. Thus,
	by~\ref{run:exists}, there exists $(\rho_0,\ell_0) \in \hat{\runs}_{w_0}$ with
	$\rho_0(w_0) = \contp$ and $\Bmf_{w_0}\models \overline{\varphi}[f_{w_0}(\rho_0,\ell_0)]$, whence, 
	by~\eqref{eq:quasimodel:induction}, we obtain
	$\Mmf, w_0 \models \varphi$.
	
	For the case of constant domains, observe that ($i$) constant domains mean that the runs are full and ($ii$) full runs give rise to constant domains.
\end{proof}

\begin{example}\label{ex:running-example:quasimodel:2}\em
It can be seen that any quasimodel satisfying $\varphi$ from our running Example~\ref{ex:running-example:1} will have infinite branching. Indeed, it should be clear that a type containing the same positive formulas as~$\contp'_0$
in Example~\ref{ex:running:quasimodel} has infinite multiplicity~($\aleph_0$) in the quasistate, $w_0$, required by~\ref{b1}. Then, by~\ref{run:exists}, the multiset $\runs$ of runs  contains $\aleph_0$ runs that go through $\contp'_0$ at $w_0$. By~\ref{rn:modal2}, there is an $a$-successor, $v_1$, of $w_0$ with a type, say $\contp_1$, containing $A(x)$. But since, by~\ref{run:exists}, only one run can go through~$\contp_1$ at $v_1$, each of the $\aleph_0$ runs in $\runs$ requires a separate $a$-successor, as depicted in Fig.~\ref{fig:running:quasimodel}. 
\end{example}


%% file: 4.2_constant.tex
 
\section{Weak Quasimodels  for $\QmonMLc$ over $\Kn$ Frames}\label{sec:kn:weak:quasimodels}

As we have seen in Section~\ref{sec:quasimodels}, standard quasimodels can require infinite frames in the presence of equality. We now define a weaker notion of quasimodel that will allow us to consider only finite frames. To this end, we first need to define a class of frames with simpler structure.

A $\Kn$ frame $\Fmf = (W, \{R_a\}_{a\in A})$ is  called \emph{tree-shaped} if there exists a $w_{0}\in W$, called the \emph{root of~$\Fmf$},
such that the domain $W$ of~$\Fmf$ is a prefix-closed set of 
words of the form 
\begin{equation}\label{eq:world:k}
	\avec{w}=w_0a_0w_1a_1\cdots a_{m-1}w_{m},
\end{equation}
where  $a_j\in A$, for $0 \leq j < m$, and each $R_a$, $a\in A$, is 
the smallest relation containing all pairs of the form\revised{R2.12}{} $(\avec{w},\avec{w}aw)\in W\times W$. For $\avec{w}\in W$ of the form~\eqref{eq:world:k}, its last component, $w_m$,  is denoted by $\textit{tail}(\avec{w})$; we denote $m$ by~$d(\avec{w})$ and call it the \emph{depth of~$\avec{w}$}.
The \emph{depth of $\Fmf$} as the maximum depth $d(\avec{w})$ of~$\avec{w}\in W$. 
Let $\Tree{d}$ denote the set of all tree-shaped $\Kn$ frames of depth bounded by~$d$.  The following result is shown by using the standard modal logic unfolding technique lifted to the first-order modal language with constants and equality. 

\begin{lemma}\label{lem:k:tree-shaped}
For both constant and expanding domains, 
every $\Kn$-satisfiable $\QMLc$-formula $\varphi$  is also $\Tree{d(\varphi)}$-satisfiable.
\end{lemma}
\begin{proof}
Let $\Mmf = (\Fmf, \Delta, \cdot)$ with $\Fmf = (W, \{R_a\}_{a\in A})$ and  $\Mmf, w_0 \models^{\mathfrak{a}} \varphi$, for some $w_0\in W$ and $\mathfrak{a}$, be given. Unfold $\Fmf$ into a tree-shaped  $\Kn$ frame
\mbox{$\Fmf^* = (W^*, \{R_a^*\}_{a\in A})$},
where $W^*$ is the set of all words $\avec{w}$ of the form~\eqref{eq:world:k} such that $d(\avec{w}) \leq d(\varphi)$ and $w_jR_{a_j} w_{j+1}$, for each $0 \leq j < d(\avec{w})$, and $R_a^*$ is $\{ (\avec{w}, \avec{w}aw)\in W^*\times W^*\}$, for each $a\in A$.  Clearly, $d(\Fmf^*)\leq d(\varphi)$.  Define an  interpretation $\Mmf^* = (\Fmf^*, \Delta^*, \cdot^*)$ by taking $\Delta^*_{\smash{\avec{w}}} = \smash{\Delta_{\textit{tail}(\avec{w})}}$, $c^{\smash{\Mmf^*(\avec{w})}} = c^{\smash{\Mmf(\textit{tail}(\avec{w}))}}$ and $P^{\smash{\Mmf^*(\avec{w})}} = P^{\smash{\Mmf(\textit{tail}(\avec{w}))}}$ for every constant $c$, predicate~$P$ and $\avec{w}\in W^*$. Observe that $\Mmf^*$ has constant domains whenever $\Mmf$ also has constant domains. By induction on the structure of $\varphi$, one can show that 
\begin{equation}\label{eq:unfolding:k:induction}
\Mmf^*, \avec{w} \models^{\mathfrak{b}} \psi \quad\text{ iff }\quad \Mmf, \textit{tail}(\avec{w}) \models^{\mathfrak{b}}\psi,
\end{equation}
for all assignments~$\mathfrak{b}$, all subformulas $\psi$ of
$\varphi$ and all words $\avec{w}\in W^*$ such that $d(\avec{w}) + d(\psi)
\leq d(\varphi)$. The claim of the lemma will then be immediate. So,
we proceed with the proof of~\eqref{eq:unfolding:k:induction}. 

The base case of atomic formulas, including equalities, which are
subformulas of modal depth 0, is immediate from the definition
of~$\Mmf^*$. The cases of the Boolean connectives $\land$ and $\neg$
and the quantifier~$\exists$ are standard.  It remains to consider the
case of $\Diamond_a \psi$. First, note that
$d(\psi) = d(\Diamond_{a}\psi) - 1$. Suppose
$\Mmf^*, \avec{w} \models^{\mathfrak{b}} \Diamond_{a}\psi$. Then there
is $\avec{v}\in W^*$ with $\avec{w}R^*_a\avec{v}$ and
\mbox{$\Mmf^*, \avec{v} \models^{\mathfrak{b}} \psi$}. Since
$d(\avec{v}) = d(\avec{w}) + 1$, by IH, we have
$\Mmf, \textit{tail}(\avec{v}) \models^{\mathfrak{b}} \psi$, whence,
due to $\textit{tail}(\avec{w})R_a\textit{tail}(\avec{v})$, we obtain
$\Mmf, \textit{tail}(\avec{w})
\models^{\mathfrak{b}}\Diamond_{a}\psi$. Conversely, suppose that
\mbox{$\Mmf, \textit{tail}(\avec{w})
\models^{\mathfrak{b}}\Diamond_{a}\psi$}.
Then there is $v\in W$ with $\textit{tail}(\avec{w})R_a v$ and
\mbox{$\Mmf, v \models^{\mathfrak{b}}\psi$}. Consider $\avec{v} = \avec{w}av\in W^*$. Since $d(\avec{v}) = d(\avec{w}) + 1$, we have, by IH, $\Mmf^*, \avec{v}\models^{\mathfrak{b}}\psi$ and so $\Mmf^*, \avec{w} \models^{\mathfrak{b}}\Diamond_{a}\psi$.
\end{proof}

Lemmas~\ref{lem:k:tree-shaped} and~\ref{lemma:quasimodel} provide a bound on the depth of frames in quasimodels, but such frames can still have infinite branching; see Example~\ref{ex:running-example:quasimodel:2}.
We now define weak quasimodels, which will allow us to deal with this issue. 
Let $\varphi$ be a $\QmonMLc$-sentence and $\Fmf = (W, \{R_a\}_{a\in A})$ a tree-shaped $\Kn$  frame. 
\revised{R1.4}{We introduce the coherence relations $\rightarrow_a$, for $a\in A$,  on types $\contp,\contp'$ for $\varphi$ by taking  
\begin{equation*}
\contp\rightarrow_a\contp' \qquad\text{ if }\qquad  \psi \in \contp' \text{ implies } \Diamond_a\psi\in \contp, \text{ for all } \Diamond_a \psi \in \sub[x]{\varphi}.
\end{equation*}}%
Then, for a basic structure $(\Fmf, \funcand)$, a \emph{weak run through $(\Fmf, \funcand)$  is} a function $\rho$ mapping each world $\avec{w}$ in an upward-closed subset of $W$ to a type $\rho(\avec{w}) \in \funcand(\avec{w})$ such that, for all $a\in A$,
\begin{enumerate}
	[label={\bf ($a$-r-coh)},leftmargin=*,series=run]	
	\item $\rho(\avec{w}) \rightarrow_a \rho(\avec{v})$, for every  $\avec{w},\avec{v}\in \dom \rho$ with $\avec{w}R_a\avec{v}$.\label{a-r-coh}
\end{enumerate}
Note that these conditions, for $a\in A$, taken together coincide with the run coherence condition~\ref{rn:modal},
but weak runs are not required to be saturated in general. A weak run $\rho$ is said to be \emph{$a$-saturated at $\avec{w}\in W$}, for $a\in A$, if 
\begin{enumerate}
	[label={\bf($a$-$\avec{w}$-r-sat)},leftmargin=*,series=run]	
\item for every $\Diamond_a \psi \in  \rho(\avec{w})$, 
	there is $\avec{v} \in W$ with
	$\avec{w}R_a\avec{v}$ and $\psi \in \rho(\avec{v})$.\label{run:wsat}
\end{enumerate}
A \emph{prototype function on $(\Fmf,\funcand)$} is a function $\prset{}$ that maps every world $\avec{w}\in W$ and
        type~$\contp\in\funcand(\avec{w})$ to a weak run  $\prfun{\avec{w}, \contp}$ through $(\Fmf, \funcand)$ such that $\prfun{\avec{w}, \contp}(\avec{w}) = \contp$ and
\begin{enumerate}[label=\textbf{(wq-sat)},leftmargin=*,series=run]
      \item 
        $\prfun{\avec{w}, \contp}$ is  
        $a$-saturated at $\avec{w}$, for each
        $a\in A$;\label{approx:R2}
\end{enumerate}
note that $\prfun{\avec{w}, \contp}$ is undefined whenever $\contp\notin\funcand(\avec{w})$. A prototype function $\prset$ on $(\Fmf, \funcand)$ defines a \emph{set} $\runsp$   of weak runs through $(\Fmf, \funcand)$: 
\begin{equation*}
\runsp(\rho) = \begin{cases}
1, & \text{if } \rho = \prfun{\avec{w},\contp}, \text{ for some } \avec{w}\in W \text{ and } \contp\in\funcand(\avec{w}),\\
0, & \text{otherwise.}
\end{cases}
\end{equation*}%
An (expanding-domains) \emph{weak $\Kn$ quasimodel for~$\varphi$ based on $\Fmf$} is a quadruple $\quasimod =
(\Fmf, \funcand, \runs, \prset{})$, where $(\Fmf, \funcand)$ is a basic structure for
$\varphi$ based on $\Fmf$,  $\runs$ is a multiset of weak runs through $(\Fmf, \funcand)$ such that~\ref{run:exists} holds for the weak runs in $\runs$, and $\prset{}$ is a prototype function on  $(\Fmf,\funcand)$ with $\runsp\subseteq\runs$.
As before, we say that $\Qmf$ is \emph{constant-domain} if
every weak run in $\runs$  is full. A weak quasimodel $\Qmf$ \emph{satisfies} $\varphi$ if~\ref{b1} holds; \revised{R2.13}{see p.~\pageref{b1}.}

It follows that, on tree-shaped frames, any quasimodel satisfying~$\varphi$  trivially gives rise to a weak \revised{R1.15}{quasimodel} satisfying $\varphi$; moreover, as Lemma~\ref{lem:weakprequasi:k} will show, from any quasimodel satisfying~$\varphi$ we can extract a weak quasimodel satisfying~$\varphi$ based on a frame of small (exponential) size. Conversely, Lemma~\ref{lemma:k:finite-approx} will show that any weak quasimodel can be ``saturated'' to obtain a quasimodel.

\begin{example}\label{ex:running:weak-quasimodel}\rm%
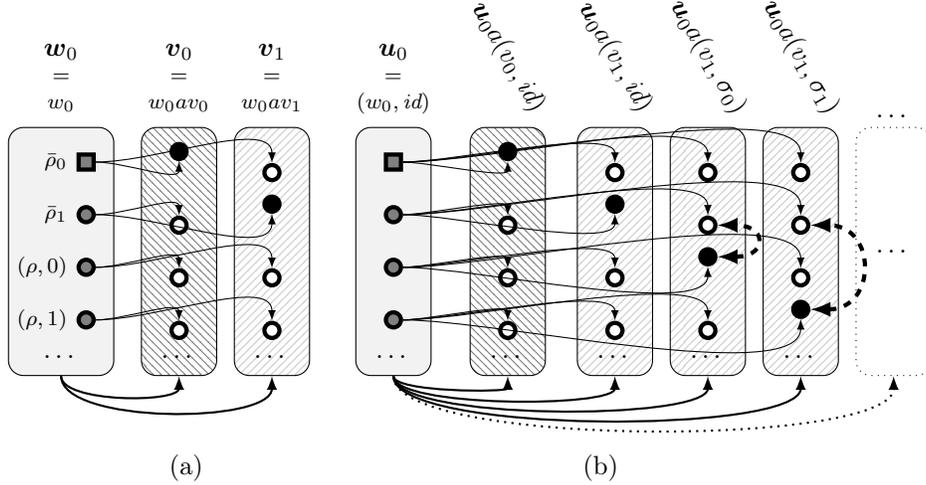
\begin{figure}[t]
\centering%
\begin{tikzpicture}[>=latex, xscale=0.95, yscale=0.7,
 	nd/.style={ultra thick,circle,draw,minimum size=2mm,inner sep=0pt},
	wn/.style={rectangle,rounded corners=2mm,draw,minimum width=10mm,minimum height=33mm}]
\node at (1.4, -5.8) {(a)};
\node[wn, minimum width=14mm, fill=gray!10,label=above:{\begin{tabular}{c}$\avec{w}_0$\\[-2pt]\footnotesize$=$\\[-2pt]\footnotesize $w_0$\end{tabular}}] (w0) at (-0.35,-1.7) {};	
%
\node[nd,fill=gray,rectangle,label=left:{\footnotesize $\bar{\rho}_0$}] (c0) at (0,0) {};
\node[nd,fill=gray,label=left:{\footnotesize$\bar{\rho}_1$}] (c1) at (0,-1) {};
\node[nd,fill=gray,label=left:{\footnotesize$(\rho,0)$}] (c2) at (0,-2) {};
\node[nd,fill=gray,label=left:{\footnotesize$(\rho,1)$}] (c3) at (0,-3) {};
\node (c4) at (-0.35,-3.7) {\dots};
\node[wn, pattern=north west lines, pattern color=gray,label=above:{\begin{tabular}{c}$\avec{v}_0$\\[-2pt]\footnotesize$=$\\[-2pt]\footnotesize $w_0av_0$\end{tabular}}] (wz) at (1.3,-1.7) {};	
\node[wn, pattern=north east lines, pattern color=gray!40,label=above:{\begin{tabular}{c}$\avec{v}_1$\\[-2pt]\footnotesize$=$\\[-2pt]\footnotesize $w_0av_1$\end{tabular}}] (w1) at (2.6,-1.7) {};	
\begin{scope}[out=-90,in=-90,looseness=0.8,thick]
\draw[->] (w0.south) to (wz.south);
\draw[->] (w0.south) to (w1.south);
\end{scope}
%
\node[nd,fill=black] (c0-0) at (1.3,0.2) {};
\node[nd,fill=white] (c1-0) at (1.3,-1.2) {};
\node[nd,fill=white] (c2-0) at (1.3,-2.2) {};
\node[nd,fill=white] (c3-0) at (1.3,-3.2) {};
\node (c4-0) at (1.3,-3.7) {\dots};
\node[nd,fill=white] (c0-1) at (2.6,-0.2) {};
\node[nd,fill=black] (c1-1) at (2.6,-0.8) {};
\node[nd,fill=white] (c2-1) at (2.6,-2.2) {};
\node[nd,fill=white] (c3-1) at (2.6,-3.2) {};
\node (c4-1) at (2.6,-3.7) {\dots};
\begin{scope}[looseness=0.9]
\draw[->,out=0,in=-90] (c0) to (c0-0.south);
\draw[->,out=0,in=90] (c0) to (c0-1.north);
\draw[->,out=0,in=90] (c1) to (c1-0.north);
\draw[->,out=0,in=-90] (c1) to (c1-1.south);
\draw[->,out=0,in=90] (c2) to (c2-0.north);
\draw[->,out=0,in=90] (c2) to (c2-1.north);
\draw[->,out=0,in=90] (c3) to (c3-0.north);
\draw[->,out=0,in=90] (c3) to (c3-1.north);
\end{scope}
\begin{scope}[xshift=29mm]
\node at (4.3, -5.8) {(b)};
\node[wn, fill=gray!10,label=above:{\begin{tabular}{c}$\avec{u}_0$\\[-2pt]\footnotesize$=$\\[-2pt]\footnotesize$(w_0,\textit{id})$\end{tabular}}] (w0) at (1.4,-1.7) {};	
%
\node[nd,fill=gray,rectangle] (c0) at (1.4,0) {}; 
\node[nd,fill=gray,] (c1) at (1.4,-1) {}; 
\node[nd,fill=gray] (c2) at (1.4,-2) {}; 
\node[nd,fill=gray] (c3) at (1.4,-3) {}; 
\node (c4) at (1.4,-3.7) {\dots};
\node[wn, pattern=north west lines, pattern color=gray] (wp) at (3,-1.7) {}; \node[rotate=-60] at (3, 2) {$\avec{u}_0a(v_0,\textit{id})$};
\node[wn, pattern=north east lines, pattern color=gray!40] (w1) at (4.5,-1.7) {};	 \node[rotate=-60] at (4.5, 2) {$\avec{u}_0a(v_1,\textit{id})$};
\node[wn, pattern=north east lines, pattern color=gray!40] (w2) at (5.8,-1.7) {}; \node[rotate=-60] at (5.8, 2) {$\avec{u}_0a(v_1,\sigma_0)$};
\node[wn,pattern=north east lines, pattern color=gray!40] (w3) at (7.1,-1.7) {};	 \node[rotate=-60] at (7.1, 2) {$\avec{u}_0a(v_1,\sigma_1)$};
\node[wn,thin,fill=white,dotted,label=above:{$\dots$}] (w4) at (8.4,-1.7) {\dots};
\begin{scope}[out=-90,in=-90,looseness=0.5,thick]
\draw[->] (w0.south) to (wp.south);
\draw[->] (w0.south) to (w1.south);
\draw[->] (w0.south) to (w2.south);
\draw[->] (w0.south) to (w3.south);
\draw[->,dotted] (w0.south) to (w4.south);
\end{scope}
\node[nd,fill=black] (c0-0) at (3,0.2) {};
\node[nd,fill=white] (c1-0) at (3,-1.2) {};
\node[nd,fill=white] (c2-0) at (3,-2.2) {};
\node[nd,fill=white] (c3-0) at (3,-3.2) {};
\node (c4-0) at (3,-3.7) {\dots};
%
\node[nd,fill=white] (c0-1) at (4.5,-0.2) {};
\node[nd,fill=black] (c1-1) at (4.5,-0.8) {};
\node[nd,fill=white] (c2-1) at (4.5,-2.2) {};
\node[nd,fill=white] (c3-1) at (4.5,-3.2) {};
\node (c4-1) at (4.5,-3.7) {\dots};
%
\node[nd,fill=white] (c0-2) at (5.8,-0.2) {};
\node[nd,fill=white] (c1-2) at (5.8,-1.2) {};
\node[nd,fill=black] (c2-2) at (5.8,-1.8) {};
\node[nd,fill=white] (c3-2) at (5.8,-3.2) {};
\node (c4-2) at (5.8,-3.7) {\dots};
%
\node[nd,fill=white] (c0-3) at (7.1,-0.2) {};
\node[nd,fill=white] (c1-3) at (7.1,-1.2) {};
\node[nd,fill=white] (c2-3) at (7.1,-2.2) {};
\node[nd,fill=black] (c3-3) at (7.1,-2.8) {};
\node (c4-3) at (7.1,-3.7) {\dots};
\begin{scope}[looseness=0.6]
\draw[->,out=0,in=-90] (c0) to (c0-0.south);
\draw[->,out=0,in=90] (c0) to (c0-1.north);
\draw[->,out=0,in=90] (c0) to (c0-2.north);
\draw[->,out=0,in=90] (c0) to (c0-3.north);
\draw[->,out=0,in=90] (c1) to (c1-0.north);
\draw[->,out=0,in=-90] (c1) to (c1-1.south);
\draw[->,out=0,in=90] (c1) to (c1-2.north);
\draw[->,out=0,in=90] (c1) to (c1-3.north);
\draw[->,out=0,in=90] (c2) to (c2-0.north);
\draw[->,out=0,in=90] (c2) to (c2-1.north);
\draw[->,out=0,in=-90] (c2) to (c2-2.south);
\draw[->,out=0,in=90] (c2) to (c2-3.north);
\draw[->,out=0,in=90] (c3) to (c3-0.north);
\draw[->,out=0,in=90] (c3) to (c3-1.north);
\draw[->,out=0,in=90] (c3) to (c3-2.north);
\draw[->,out=0,in=-90] (c3) to (c3-3.south);
\end{scope}
\begin{scope}[ultra thick,dashed]
\draw[<->,out=0,in=0,looseness=1.5] (c3-3) to (c1-3);
\draw[<->,out=0,in=0,looseness=3] (c2-2) to (c1-2);
\end{scope}
\end{scope}
\end{tikzpicture}
\caption{(a)  A weak quasimodel  in Example~\ref{ex:running:weak-quasimodel} and (b) its saturated quasimodel constructed in the proof of Lemma~\ref{lemma:k:finite-approx}, where $\sigma_i$ swaps $(\rho,i)$ with the only copy of $\bar{\rho}_1$ and preserves all other indexed runs.}\label{fig:running:weak-quasimodel}
\end{figure}
In our running
Example~\ref{ex:running:quasimodel}, the frame is already tree-shaped and so we use $\avec{w}_0$ and $\avec{v}_i$ to refer to its worlds:  let $\avec{w}_0 = w_0$ and $\avec{v}_i = w_0av_i$, for~$i\in\mathbb{N}$. First, we pick, say, $\rho_0$ from $\runs_{\avec{w}_0,\contp_0}$ and $\rho_1$ from $\runs_{\avec{w}_0,\contp'_0}$, with~$\avec{v}_0$ and~$\avec{v}_1$ being the $\Diamond_aA(x)$-witness at $\avec{w}_0$ for $\rho_0$ and $\rho_1$, respectively (note that $\rho_0$ is uniquely determined, while $\rho_1$ could be any of the $\rho_i$, for $i \geq 1$).
Then the set of worlds~$W'$ comprises~$\avec{w}_0$, $\avec{v}_0$
and~$\avec{v}_1$, and the multiset $\runs'$  of weak runs consists of
the restrictions~$\bar{\rho}_0$ and ~$\bar{\rho}_1$ of the chosen
$\rho_0$ and $\rho_1$ to $W'$ together with $\aleph_0$ copies of the
``unsaturated'' weak run~$\rho$ that maps $\avec{w}_0$ to
$\contp'_0$ and both $\avec{v}_0$ and $\avec{v}_1$ to~$\contp_2$ (these unsaturated weak runs are the restrictions of the~$\rho_i$, for~$i \geq 2$, to $W'$). 
For our prototype weak runs, $\prfun{\avec{w}_0,\contp_0}$ and $\prfun{\avec{w}_0,\contp_0'}$, we take $\bar{\rho}_0$ and $\bar{\rho}_1$, respectively; for $\prfun{\avec{v}_i,\contp_1}$ and $\prfun{\avec{v}_i,\contp_2}$, with $i = 0, 1$, we take any of the weak runs---the choice is irrelevant as $\contp_1$ and $\contp_2$ have no positive $\Diamond_a$-formulas.  
As Lemma~\ref{lem:weakprequasi:k} below shows, this is a finitely representable weak quasimodel satisfying the formula; see \revised{R1.16; see the caption}{Fig.~\ref{fig:running:weak-quasimodel}}~(a). Note, however, that the weak quasimodel is not finite as, for example, $\contp_0'$ still has infinite multiplicity in $\avec{w}_0$.

On the other hand, given this small weak quasimodel, we can restore the original quasimodel by creating a separate $a$-successor of $\avec{w}_0$ for each of the copies of the unsaturated weak run $\rho$ and using the prototype weak run~$\bar{\rho}_1$ as a template for witnessing $\Diamond_a A(x)$ in that successor; more details are given in the proof of Lemma~\ref{lemma:k:finite-approx} below.
\end{example}

\begin{lemma}\label{lem:weakprequasi:k}
Let  $\varphi$ be a $\QmonMLc$-sentence.
For both constant and expanding domains, 
for every quasimodel satisfying $\varphi$ based on a $\Tree{d}$ frame, there is a weak quasimodel satisfying $\varphi$  based on a $\Tree{d}$ frame  of size  $2^{O(d\cdot |\varphi|)}$.
\end{lemma}
\begin{proof}
  Assume a quasimodel \revised{R1.5}{$\Qmf = (\Fmf,\funcand,\runs)$}
  satisfying $\varphi$ with tree-shaped $\Fmf=(W,\{R_{a}\}_{a\in A})$
  of finite depth $d$ is given.  We inductively define a sequence
  $W_{0},W_{1},\cdots W_{d}\subseteq W$ of sets of worlds with
  $d(\avec{w})=i$ for all $\avec{w}\in W_{i}$ as follows. Set
  $W_{0}=\{w_{0}\}$, where $w_{0}$ is the root of~$\Fmf$, and assume
  $W_{i}$ has been defined. For each $\avec{w}\in W_{i}$ and each
  $\contp\in \funcand(\avec{w})$,
   take a run \mbox{$\rho \in \runs_{\avec{w},\contp}$}.
  By~\ref{rn:modal2}, for every $\Diamond_a \psi \in \rho(\avec{w})$,
  there exists $\avec{v}\in W$ such that $\avec{w}R_{a}\avec{v}$ and
  $\psi\in \rho(\avec{v})$, so we add $\avec{v}$ to $W_{i+1}$; we also set  $\prfun{\avec{w},\contp}= \rho$. Thus,
  $|W_{i+1}| \leq |W_i| \cdot 2^{|\varphi|} \cdot |\varphi|$. Let
  $\Qmf'=(\Fmf',\funcand',\runs',\prset')$ be the restriction of
  $(\Fmf,\funcand,\runs,\prset)$ to $W'=\bigcup_{i = 0}^{d}W_{i}$;
note
  that, if distinct runs in $\runs$ 
  have the same restriction to $W'$, their multiplicities are added in $\runs'$.
  It can be seen that
  $\Qmf'$ is a weak quasimodel required by the lemma: in particular,
  by construction, each weak run $\prfun[']{\avec{w},\contp}$ is $a$-saturated
  at~$\avec{w}$, for each $a\in A$. Moreover, weak quasimodel $\Qmf'$
  contains only full weak runs whenever quasimodel $\Qmf$ contains only full runs.
\end{proof}

\begin{lemma}\label{lemma:k:finite-approx}
Let  $\varphi$ be a $\QmonMLc$-sentence.
For both constant and expanding domains, 
for any  weak quasimodel $(\Fmf,\funcand,\runs,\prset{})$ satisfying $\varphi$ based on a $\Tree{d}$ frame, 
there is a quasimodel $(\Fmf',\funcand',\runs')$ satisfying $\varphi$ based on a $\Tree{d}$ frame\textup{;} moreover, $|\runs'| = |\runs|$ and
$|\Fmf'| = O(|\Fmf|\cdot |\runs|^d)$ provided that $\Fmf$ and $\runs$ are finite. 
\end{lemma}
\begin{proof}
  Let $\Qmf = (\Fmf,\funcand,\runs,\prset{})$ be a weak
  quasimodel satisfying $\varphi$ based on a tree-shaped frame
  $\Fmf=(W,\{R_a\}_{a\in A})$ of depth
  $d$. 
  We construct a new frame using copies of the worlds and extend the
  existing weak runs to the world copies in such a way that each run
  becomes saturated, that is, satisfying~\mbox{\ref{rn:modal2}}. More
  precisely, by~\ref{approx:R2}, for each $\avec{w}\in W$ and each
  $\contp\in \funcand(\avec{w})$, we take the prototype weak run
  $\prfun{\avec{w},\contp} \in \runs_{\avec{w},\contp}$, which is
  $a$-saturated at~$\avec{w}$, for all $a\in A$: for every
  $\Diamond_a \psi\in\contp$, there is a $\Diamond_a \psi$-witness
  world $\avec{v}_{\avec{w},\contp,\Diamond_a\psi}\in W$ such that
  $\avec{w}R_a\avec{v}_{\avec{w},\contp,\Diamond_a\psi}$ and
  $\psi\in\prfun{\avec{w},\contp}(\avec{v}_{\avec{w},\contp,\Diamond_a\psi})$. Next,
  for each indexed weak run~$(\rho,\ell)$ that coincides with
  $\prfun{\avec{w},\contp}$ on $\avec{w}$, we create a copy
  $\avec{v}_{\rho,\ell}$ of 
  $\avec{v}_{\avec{w},\contp,\Diamond_a\psi}$ as an $a$-successor for
  $\avec{w}$. We extend $(\rho,\ell)$ to $\avec{v}_{\rho,\ell}$ so that it
  coincides with the prototype~$\prfun{\avec{w},\contp}$ starting from the
  $\Diamond_a \psi$-witness world, thus making the
  extended~$(\rho,\ell)$ saturated for $\Diamond_a \psi$ at
  $\avec{w}$; all other runs, except $(\prfun{\avec{w},\contp},0)$, are extended to~$\avec{v}_{\rho,\ell}$ in
  exactly the same way as they are defined on
  $\avec{v}_{\avec{w},\contp,\Diamond_a\psi}$; however, in order to preserve
  the multiplicity constraints,  \revised{R2.14}{$(\prfun{\avec{w},\contp},0)$ starting from
  $\avec{v}_{\rho,\ell}$ is defined to be identical to~$(\rho,\ell)$ from~$\avec{v}_{\avec{w},\contp,\Diamond_a\psi}$. In other words, $\avec{v}_{\rho,\ell}$ is an identical copy of  $\avec{v}_{\avec{w},\contp,\Diamond_a\psi}$ except that the
runs $(\rho,\ell)$ and $(\prfun{\avec{w},\contp},0)$ `swap' places.  We use bijections, denoted by~$\sigma$, to describe the swaps.
In Example~\ref{ex:running:weak-quasimodel}, for each weak run $(\rho, i)$, which is unsaturated at $w_0$, we create a copy $\avec{u}_0 a(v_1,\sigma_i)$ of $w_0av_1$, where $\avec{u}_0 = (w_0,\textit{id})$ is the new root and $\sigma_i$ denotes the swap between $(\rho, i)$ and its matching prototype run~$\overline{\rho}_1$. Next, all runs except $(\rho, i)$ and $\overline{\rho}_1$ are extended to the new copy as they are defined on $w_0av_1$; on the other hand, $(\rho, i)$ and $(\overline{\rho}_1,0)$ are swapped: $(\rho, i)$  gets its types from~$\overline{\rho}_1$ thus making it $a$-saturated at $\avec{u}_0$, while $(\overline{\rho}_1,0)$ gets its types from $(\rho, i)$ thus preserving the multiplicity of types; see Fig.~\ref{fig:running:weak-quasimodel}~(b).
}
	 
  Formally, we consider the set $\hat{\runs}$ of indexed weak runs
  $(\rho,\ell)$ associated with the multiset~$\runs$ and denote the
  set of all bijections on~$\hat{\Rmf}$ by~$B(\hat{\Rmf})$.  Then, for
  each $\avec{w}\in W$, we define the \emph{$\avec{w}$-repair set}
  $\textrm{Rep}(\avec{w})$ by taking
\begin{equation*}	
\textrm{Rep}(\avec{w}) = \bigl\{ \sigma_{\avec{w},\rho, \ell}\in B(\hat{\Rmf}) \mid (\rho, \ell) \in \hat{\Rmf}_{\avec{w}}\bigr\}, 
\end{equation*}
where $\sigma_{\avec{w},\rho, \ell}$ is the bijection on $\hat{\Rmf}$
that swaps $(\rho, \ell)$ with the indexed prototype weak run
$(\prfun{\avec{w},\contp}, 0)$ for $\avec{w}$ and
$\contp=\rho(\avec{w})$ and maps the remaining indexed weak runs to
themselves.
\revised{R2.14}{Observe that, for any $\sigma\in \textrm{Rep}(\avec{w})$
and any indexed weak run $(\rho,\ell) \in \hat{\runs}$, either $\avec{w}$ belongs to the domains of both $\rho$ and its image, $\sigma(\rho,\ell)$, or neither domain contains $\avec{w}$; moreover, in the latter case $\sigma(\rho, \ell) = (\rho, \ell)$.}
Note that, as each $\hat{\Rmf}_{\avec{w}}$ is non-empty, the set $\textrm{Rep}(\avec{w})$ is non-empty and contains the identity function~$\textit{id}$ on $\hat{\runs}$, which `swaps' \emph{any} indexed weak run with itself.
	
Now, for each word 
\begin{equation*}
\avec{w} = w_0a_0w_1a_1\cdots a_{m-1}w_{m} \in W, 
\end{equation*}
we construct words of the form 
\begin{equation}\label{eqveru:k}
		\avec{u}=(w_{0},\sigma_0)a_{0}(w_{1},\sigma_1)a_{1} \cdots a_{m-1}(w_{m},\sigma_m),
\end{equation}
where $\sigma_0= \textit{id}$ and $\sigma_{j+1}\in
  \textrm{Rep}(\avec{w}_j)$
with $\avec{w}_j = w_0a_0w_1a_1\cdots
a_{j-1}w_{j}$, for each $0\leq j<m$.
We denote $\avec{w}$ by $\gen{\avec{u}}$ and $\sigma_m \circ \cdots
\circ \sigma_{0}\in B(\hat{\runs})$ by $\rep{\avec{u}}$.  Since domains of weak runs are upward-closed, $\rep{\avec{u}}(\rho, \ell) = (\rho, \ell)$, for any $\gen{\avec{u}}\notin\dom\rho$.
Let~$W'$ be the set of all words of the form~\eqref{eqveru:k} for all $\avec{w}\in W$.
Consider a tree-shaped frame $\Fmf' = (W', \{R_a'\}_{a\in A})$, where
$R_a' = \{ (\avec{u},\avec{u}a(w,\sigma)) \in W'\times W' \}$, for
$a\in A$. Clearly, it is as required by the lemma. Define 
$\funcand'(\avec{u})= \funcand(\gen{\avec{u}})$ for any
$\avec{u}\in W'$.
	
We now define a multiset $\runs'$ of runs through $(\Fmf',\funcand')$. 
Denote by $\textit{rep}$ the function mapping any $(\rho,\ell)\in
\smash{\hat{\runs}}$ to a weak run $\rho'$
on $W'$ defined by setting 
\begin{equation*}
\rho'(\avec{u}) = \rep{\avec{u}}(\rho, \ell)(\gen{\avec{u}}), \text{ for any } \avec{u}\in W' \text{ with } \gen{\avec{u}} \in \dom \rho.
\end{equation*}

By induction on the depth of $\avec{u}\in W'$, we can show that
$\rho'(\avec{u})$ is well-defined.
Indeed, this is trivial for the
only $\avec{u}$ of depth~0 as it is of the form~$(w_0,\textit{id})$.
Suppose $\rho'(\avec{u})$ is well-defined,
for all
$\avec{u}\in W'$ of depth $m$. Consider any $\avec{v}\in W'$ of depth
$m+1$, which is of the form $\avec{u}a(w,\sigma)$
for some $\avec{u}\in W'$, $w$ and
$\sigma\in \textrm{Rep}(\gen{\avec{u}})$.  By definition, we have
$\gen{\avec{v}} = \gen{\avec{u}}aw$ and
$\rep{\avec{v}} = \sigma \circ \rep{\avec{u}}$.
If $\gen{\avec{u}}\in\dom\rho$, then,
by IH,
$\rep{\avec{u}}(\rho,\ell)$ is defined on
$\gen{\avec{u}}$.
As we observed above, for any
$\sigma\in \textrm{Rep}(\gen{\avec{u}})$, the weak run
$\sigma(\rep{\avec{u}}(\rho,\ell))$ is defined on
  $\gen{\avec{u}}$ and so
on $\gen{\avec{v}}$ (as $\gen{\avec{u}}R_a\gen{\avec{v}}$).
\revised{R1.6}{If $\gen{\avec{u}}\notin\dom \rho$, then
$\rep{\avec{u}}(\rho,\ell) = (\rho,\ell)$,
whence $\rep{\avec{v}}(\rho,\ell)  =
\sigma(\rep{\avec{u}}(\rho,\ell)) =
\sigma(\rho,\ell)$. As $\sigma\in \textrm{Rep}(\gen{\avec{u}})$, we have $\sigma(\rho,\ell) =
(\rho,\ell)$.
Thus, $\rep{\avec{v}}(\rho,\ell) = (\rho,\ell)$ and so $\rep{\avec{v}}(\rho,\ell)(\gen{\avec{v}})$ is defined whenever
$\gen{\avec{v}} \in \dom \rho$.}

Let $\Rmf'$ be the multiset of weak runs $\rho'$ such that 
\begin{equation*}
\Rmf'(\rho') = \bigl|\bigl\{ (\rho,\ell)\in\hat{\runs} \mid \textit{rep}(\rho,\ell) = \rho' \bigr\}\bigr|.
\end{equation*}
Due to this multiplicity condition, $\textit{rep}$ can be extended to a bijection from~$\hat{\Rmf}$ onto the set $\hat{\Rmf}'$ of indexed weak runs in $\Rmf'$ (note, however,  that the indexes of $\rho' = \textit{rep}(\rho,\ell)$ will in general be unrelated to $\ell$).
	
We show that $\Qmf' = (\mathfrak{F}',\funcand', \Rmf')$ is a required
quasimodel. Condition~\ref{b1} is straightforward.  We
prove~\ref{run:exists} for (weak) runs in $\runs'$
using~\ref{run:exists} for weak runs in $\runs$.  To this end, we show
that
the composition $\rep{\avec{u}}\circ \textit{rep}^{-1}$
is a bijection from $\hat{\runs}'_{\avec{u},\contp}$ onto
$\hat{\runs}_{\gen{\avec{u}},\contp}$, for any $\avec{u}\in
W'$ 
and any $\contp\in\funcand'(\avec{u})$. First, for any
$(\rho',\ell') \in \hat{\runs}'_{\avec{u},\contp}$, we have
$\rep{\avec{u}}(\textit{rep}^{-1}(\rho',\ell'))\in
\hat{\runs}_{\gen{\avec{u}},\contp}$: denote
$(\rho,\ell) = \textit{rep}^{-1}(\rho',\ell')$; then
$\rep{\avec{u}}(\textit{rep}^{-1}(\rho',\ell'))(\gen{\avec{u}}) =
\rep{\avec{u}}(\rho,\ell)(\gen{\avec{u}}) = \rho'(\avec{u}) = \contp$,
as required. The composition is injective as both $\rep{\avec{u}}$ and
$\textit{rep}$ are bijections. Finally, the composition is surjective:
for any $(\rho,\ell)\in \hat{\runs}_{\gen{\avec{u}},\contp}$, consider
$(\rho',\ell') = \textit{rep}(\rep{\avec{u}}^{-1}(\rho,\ell))$; we
have
$\rho'(\avec{u}) =
\rep{\avec{u}}(\rep{\avec{u}}^{-1}(\rho,\ell))(\gen{\avec{u}}) =
\rho(\gen{\avec{u}}) = \contp$, and so
$(\rho',\ell')\in\hat{\runs}'_{\avec{u},\contp}$.
	
We next show that the elements of $\Rmf'$ are indeed runs, that is,
they satisfy~\mbox{\ref{rn:modal}} and~\ref{rn:modal2}. Let
$\Diamond_a \psi \in \sub[x]{\varphi}$. Consider $\rho'\in \Rmf'$ and
$\avec{u}\in \dom\rho'$ of the form~\eqref{eqveru:k}. We have
$\rho' = \textit{rep}(\rho,\ell)$ for some
$(\rho,\ell)\in\hat{\runs}$.
	
For~\ref{rn:modal}, assume 
there exists $\avec{v} \in W'$
such that $\avec{u}R_a'\avec{v}$ and
$\psi \in \rho'(\avec{v})$. By construction,
$\gen{\avec{u}}R_a \gen{\avec{v}}$.  Since
$\rho'(\avec{v}) = \rep{\avec{v}}(\rho,\ell)(\gen{\avec{v}})$, by
coherence of the weak run $\rep{\avec{v}}(\rho, \ell)$, we obtain
$\Diamond_a\psi\in \rep{\avec{v}}(\rho, \ell)(\gen{\avec{u}})$.
By definition,
$\rep{\avec{v}} = \sigma \circ \rep{\avec{u}}$,
for some
$\sigma\in\textrm{Rep}(\gen{\avec{u}})$, and so
$\rep{\avec{u}}(\rho,\ell)$ coincides with $\rep{\avec{v}}(\rho,\ell)$
on $\gen{\avec{u}}$. It follows that
$\Diamond_a\psi\in \rep{\avec{u}}(\rho,\ell)(\gen{\avec{u}}) =
\rho'(\avec{u})$.
	
For~\ref{rn:modal2}, assume $\Diamond_a\psi \in \rho'(\avec{u})$.
Let $\contp = \rep{\avec{u}}(\rho,\ell)(\gen{\avec{u}})$. We have
\mbox{$\Diamond_a\psi\in \contp$}.
Condition~\ref{approx:R2} provides a prototype weak run
$\prfun{\gen{\avec{u}}, \contp}$ for $\gen{\avec{u}}$ and $\contp$ and
a witness $\gen{\avec{w}}\in W$ such that
$\gen{\avec{u}} R_a \gen{\avec{w}}$ and
$\psi \in \prfun{\gen{\avec{u}}, \contp}(\gen{\avec{w}})$. By
construction, we have $\gen{\avec{w}} = \gen{\avec{u}}a w$, for
some~$w$.  Consider $\avec{v}= \avec{u}a(w,\sigma)$, where
$\sigma \in\textrm{Rep}(\gen{\avec{u}})$ swaps
$\rep{\avec{u}}(\rho,\ell)$ with
$(\prfun{\gen{\avec{u}},\contp},0)$. By definition, we have
$\avec{u}R_a'\avec{v}$ and
$\rho'(\avec{v})=\sigma(\rep{\avec{u}}(\rho,\ell))(\gen{\avec{w}})=\prfun{\gen{\avec{u}},\contp}(\gen{\avec{w}})$. So,
$\psi \in\rho'(\avec{v})$, as required.

Finally, if every weak run in $\Rmf$ is full, then so is every $\rho'\in\Rmf'$. Hence, $\Qmf'$ is a constant-domain quasimodel whenever $\Qmf$ is.
\end{proof}

\section{Weak Quasimodels  for $\QmonMLc$ over $\Sfiven$ Frames}\label{sec:s5n:weak:quasimodels}

In this section, we consider interpretations constructed from $\Sfiven$ frames.
Recall that a frame $\Fmf=(W,\{R_a\}_{a\in A})$ is an  $\Sfiven$ frame if all $R_a$ are equivalence relations. Since the $R_a$ are symmetric, interpretations over an $\Sfiven$ frame can be assumed to have constant domains (so, we shall omit the constant-domain qualifier for the rest of the section). 

An $\Sfiven$ frame is 
called \emph{tree-shaped} if there exists a $w_{0}\in W$, the root of~$\Fmf$,
such that the domain $W$ of~$\Fmf$ is a prefix-closed set of 
words of the form 
\begin{equation}\label{eq:world}
	\avec{w}=w_0a_0w_1a_1\cdots a_{m-1}w_{m},
\end{equation}
where $a_j\in A$, $a_{j}\ne a_{j+1}$, and each $R_a$ is 
the smallest equivalence relation containing all pairs of the form $(\avec{w},\avec{w}aw)\in W\times W$. 
We define the immediate $a$-successor relation $\prec_a$, for $a\in A$, by taking $\avec{w} \prec_a \avec{w}'$ iff $\avec{w}'$ is of the form $\avec{w}aw$, for some $w$. We also write $\avec{w}\preceq_a \avec{w'}$ iff either $\avec{w} = \avec{w}'$ or $\avec{w}\prec_a \avec{w}'$.  
Observe that, by definition, such frames have no $\prec_a$-chains of length greater than one; in other words, each world has no $\prec_a$-predecessor or has no $\prec_a$-successors (or has neither). It follows that
\begin{multline}\label{eq:s5n:frame}
\avec{w}R_a \avec{w}' \quad \text{ iff } \quad  \avec{w}\prec_a \avec{w}' \quad \text{ or } \quad \avec{w}' \prec_a \avec{w}  \quad\text{ or }\quad \avec{w} = \avec{w}'\\ \text{ or there is } \avec{v}\in W \text{ with } \avec{v}\prec_a \avec{w}  \ \text{ and } \avec{v}\prec_a \avec{w}'. 
\end{multline}
An example is shown in Fig.~\ref{fig:s5:frame}, where relations $\prec_a$ and $\prec_b$ are depicted by solid and dotted arrows, respectively.  Then $\avec{w}$, $\avec{w}aw_1$ and $\avec{w}aw_2$ are in the same $a$-equivalence class, with $\avec{w}$ having no $\prec_a$-predecessor and the other two having no $\prec_a$-successors, but $\avec{w}bw_3$, $\avec{w}bw_3aw_4$ and $\avec{w}bw_3aw_5$ form another $a$-equivalence class, with $\avec{w}bw_3$ having no $\prec_a$-predecessor and the other two no $\prec_a$-successors; $\avec{w}$ and $\avec{w}bw_3$ form a $b$-equivalence class. 
\begin{figure}[t]
\centering%
\begin{tikzpicture}[>=latex,
nd/.style={draw,circle,thick,inner sep=0pt,minimum size=2mm},
ndb/.style={fill=gray!20,circle,inner sep=0pt,minimum size=5mm}
]
\begin{scope}
\node[ndb] (bw0) at (0,0) {};
\node[ndb] (bw1) at (-2.4,1) {};
\node[ndb] (bw2) at (0,1) {};
\node[ndb] (bw3) at (2.4,1) {};
\node[ndb] (bw4) at (1.2,2) {};
\node[ndb] (bw5) at (3.6,2) {};
\fill[gray!20] (bw0.south) -- (bw1.south) -- (bw1.north) -- (bw2.north) -- (bw2.east) -- (bw0.east) -- cycle;
\fill[gray!20] (bw3.south west) -- (bw4.south west) -- (bw4.north) -- (bw5.north) -- (bw5.south east) -- (bw3.south east) -- cycle;
\end{scope}
\node[nd,fill=white,label=below:{$\avec{w}$}] (w0) at (bw0) {};
\node[nd,fill=white,label=left:{$\avec{w}aw_1$}] (w1) at (bw1) {};
\node[nd,fill=white,label=right:{$\avec{w}aw_2$}] (w2) at (bw2) {};
\node[nd,fill=white,label=below:{$\avec{w}bw_3$}] (w3) at (bw3) {};
\node[nd,fill=white,label=left:{$\avec{w}bw_3aw_4$}] (w4) at (bw4) {};
\node[nd,fill=white,label=right:{$\avec{w}bw_3aw_5$}] (w5) at (bw5) {};
\begin{scope}[thick]
\draw[->] (w0) -- (w1);
\draw[->] (w0) -- (w2);
\draw[->,dotted] (w0) -- (w3);
\draw[->] (w3) -- (w4);
\draw[->] (w3) -- (w5);
\end{scope}
\end{tikzpicture}
\caption{An example of a tree-shaped $\Sfiven$ frame.}\label{fig:s5:frame}
\end{figure}

Let $\SfTree{d}$ denote the set of all tree-shaped $\Sfiven$ frames of depth bounded by $d$, where the notion of depth is defined as in Section~\ref{sec:kn:weak:quasimodels}.
 We then have the following counterpart of Lemma~\ref{lem:k:tree-shaped}.

\begin{lemma}\label{lem:s5:tree-shaped}
Every $\Sfiven$-satisfiable  $\QMLc$-formula $\varphi$ is $\SfTree{d(\varphi)}$-satisfiable.
\end{lemma}
\begin{proof}
Let $\Mmf = (\Fmf, \Delta, \cdot)$ with $\Fmf = (W, \{R_a\}_{a\in A})$ and  $\Mmf, w_0 \models^{\mathfrak{a}} \varphi$, for some $w_0\in W$ and $\mathfrak{a}$, be given. Unfold $\Fmf$ into 
$\Fmf^* = (W^*, \{R_a^*\}_{a\in A})$,
where $W^*$ is the set of all words $\avec{w}$ of the form~\eqref{eq:world} with $d(\avec{w}) \leq d(\varphi)$ and $(w_j, w_{j+1})\in R_{a_j}$, $w_j \ne w_{j+1}$ and $a_j \ne a_{j+1}$, for all $0 \leq j < m$, and $R_a^*$ is the smallest equivalence relation containing all pairs $(\avec{w}, \avec{w}aw)\in W^*\times W^*$.  Clearly, $\Fmf^*$ is a tree-shaped $\Sfiven$ frame of depth bounded by $d(\varphi)$. As in Lemma~\ref{lem:k:tree-shaped}, define $\Mmf^* = (\Fmf^*, \Delta^*, \cdot^*)$ by taking $\Delta^*_{\smash{\avec{w}}} = \Delta_{\smash{\textit{tail}(\avec{w})}}$, $c^{\smash{\Mmf^*(\avec{w})}} = c^{\smash{\Mmf(\textit{tail}(\avec{w}))}}$ and $P^{\smash{\Mmf^*(\avec{w})}} = P^{\smash{\Mmf(\textit{tail}(\avec{w}))}}$ for every constant $c$, predicate~$P$ and $\avec{w}\in W^*$.  By induction on the structure of~$\varphi$, one can show that 
\begin{equation}\label{eq:unfolding:induction}
\Mmf^*, \avec{w} \models^{\mathfrak{b}} \psi \quad\text{ iff }\quad \Mmf, \textit{tail}(\avec{w}) \models^{\mathfrak{b}}\psi,
\end{equation}
for all assignments~$\mathfrak{b}$, subformulas $\psi$ of
$\varphi$ and $\avec{w}\in W^*$ with $d(\avec{w}) + d(\psi)
\leq d(\varphi)$. The claim of the lemma will then be immediate. 

The base case and the cases of the Boolean connectives and the quantifier in~\eqref{eq:unfolding:induction} are as in the proof ofLemma~\ref{lem:k:tree-shaped}.
It remains to consider the case of~$\Diamond_a \psi$. Suppose first $\Mmf^*, \avec{w} \models^{\mathfrak{b}} \Diamond_{a}\psi$. Then there is $\avec{v}\in W^*$ with $\avec{w}R^*_a\avec{v}$ and $\Mmf^*, \avec{v} \models^{\mathfrak{b}} \psi$. By~\eqref{eq:s5n:frame},  $d(\avec{v}) \leq d(\avec{w}) + 1$.
Then, as $d(\psi) = d(\Diamond_{a}\psi) - 1$, by IH, $\Mmf,
\textit{tail}(\avec{v}) \models^{\mathfrak{b}} \psi$.
Since $\avec{w}R^*_a\avec{v}$, then, by~\eqref{eq:s5n:frame}
  and $R_a$ being  an equivalence relation,
we have $\textit{tail}(\avec{w})R_a\textit{tail}(\avec{v})$, whence  $\Mmf, \textit{tail}(\avec{w}) \models^{\mathfrak{b}}\Diamond_{a}\psi$. Conversely, assume that $\Mmf, \textit{tail}(\avec{w}) \models^{\mathfrak{b}}\Diamond_{a}\psi$. 
Then there is $v\in W$ such that $\textit{tail}(\avec{w})R_a v$ and
\mbox{$\Mmf, v \models^{\mathfrak{b}}\psi$}. If $\textit{tail}(\avec{w})=v$, then $\Mmf^*, \avec{w} \models^{\mathfrak{b}}\psi$ by IH and $\Mmf^*, \avec{w} \models^{\mathfrak{b}}\Diamond_{a}\psi$ due to reflexivity of $R_a^*$. 
If $\textit{tail}(\avec{w})\ne v$, there are two further cases to consider. If  $\avec{w}$ has no $\prec_a$-predecessor, then
$\avec{w}R_a^*\avec{w}av$, whence $\Mmf^*, \avec{w} \models^{\mathfrak{b}}\Diamond_{a}\psi$ follows again by IH as \mbox{$d(\psi) = d(\Diamond_{a}\psi) - 1$}. Otherwise, there is
$\avec{w}'$ with $\avec{w}'\prec_a \avec{w}$. 
It follows that $\textit{tail}(\avec{w}') R_a \textit{tail}(\avec{w})$, which together with $\textit{tail}(\avec{w})R_a v$ by transitivity implies $\textit{tail}(\avec{w}')R_a v$.
Thus, we have either $\textit{tail}(\avec{w}')=v$ or $\avec{w}' R_a^* \avec{w}'av$. In the former case, by symmetry of
$R_{a}^{\ast}$, we have $\avec{w}R_a^*\avec{w}'$, whence, by IH, \mbox{$\Mmf^*, \avec{w}' \models^{\mathfrak{b}}\psi$} and thus $\Mmf^*, \avec{w} \models^{\mathfrak{b}}\Diamond_{a}\psi$.
In the latter case, by symmetry of $R^*_a$, we obtain $\avec{w} R_a^* \avec{w}'$, whence, by transitivity, $\avec{w}R_a^* \avec{w}' av$. Then, by IH, $\Mmf^*, \avec{w}' av\models^{\mathfrak{b}}\psi$ and so $\Mmf^*, \avec{w} \models^{\mathfrak{b}}\Diamond_{a}\psi$.
\end{proof}

We now define a notion of $\Sfiven$ weak quasimodels as a variant of $\Kn$ weak quasimodels.
Let $\varphi$ be a $\QmonMLc$-sentence and $\Fmf = (W, \{R_a\}_{a\in A})$ a tree-shaped $\Sfiven$   frame. We say a Boolean-saturated subset $\contp$ of $\sub[x]{\varphi}$ is an \emph{$\Sfiven$ type} if
\begin{enumerate}
	[label=\textbf{(s5-type)},leftmargin=*,series=run]
	\item $\Diamond_a\psi\in\contp$, for each $\psi \in \contp$ with $\Diamond_a\psi\in\sub[x]{\varphi}$.\label{s5-type}
\end{enumerate}
We say that $(\Fmf, \funcand)$ is an \emph{$\Sfiven$ basic structure for $\varphi$ based on $\Fmf$ if each~$\funcand(\avec{w})$, for $\avec{w} \in W$, consists of $\Sfiven$ types for $\varphi$.}
\revised{R2.15}{The notion of coherence between $\Sfiven$ types $\contp,\contp'$ reflects the fact that the $R_a$ are equivalence relations:  we write  
\begin{gather*}
\contp \leftrightarrow_a \contp' \qquad\text{if}\qquad \Diamond_a\psi\in \contp  \ \ \text{ iff }\ \  \Diamond_a\psi\in \contp', \ \ \text{ for all } \Diamond_a \psi \in \sub[x]{\varphi}.
\end{gather*}
It follows that each $\leftrightarrow_a$ is an equivalence relation on $\Sfiven$ types for $\varphi$. A \emph{weak run} through an $\Sfiven$ basic structure $(\Fmf, \funcand)$  is a function $\rho$ mapping each $\avec{w}\in W$ to $\rho(\avec{w})\in\funcand(\avec{w})$ such that, for all $a\in A$,
\begin{enumerate}
	[label={\bf ($a$-r-coh-s5)},leftmargin=*,series=run]	
\item $\rho(\avec{w}) \leftrightarrow_a \rho(\avec{v})$,  for every $\avec{w},\avec{v}\in W$ with $\avec{w}R_a\avec{v}$,
\end{enumerate}
which is the $\Sfiven$ counterpart of the run coherence condition~\ref{a-r-coh}.} Note that weak runs through an $\Sfiven$ basic structure are full by definition.
An \emph{$\Sfiven$ prototype function on $(\Fmf,
    \funcand)$} is a function $\prset{}$ that maps every $\avec{w}\in W$ and $\contp\in\funcand(\avec{w})$ to a weak run $\prfun{\avec{w}, \contp}$ through $(\Fmf, \funcand)$ such that $\prfun{\avec{w}, \contp}(\avec{w}) = \contp$ and
\begin{enumerate}
	[label=\textbf{(wq-sat-s5)},leftmargin=*,series=run]
	\item  $\prfun{\avec{w}, \contp}$  
        is $a$-saturated at $\avec{w}$, for every $a\in A$ such that $\avec{w}$ has  no $\prec_a$-predecessor.\label{approx:R2:s5}
\end{enumerate}
Observe that~\ref{approx:R2:s5} 
does \emph{not} require witnesses for $\Diamond_a\psi$ at worlds $\avec{w}$ that have $\prec_a$-predecessors, which means that, due to the definition of $R_a$, each witness~$\avec{v}$ for a $\Diamond_a\psi$ satisfies $\avec{w}\preceq_a \avec{v}$ (rather than just $\avec{w}R_a\avec{v}$ as in $\Kn$ quasimodels, which due to symmetry of $R_a$ may mean $\avec{v}\prec_a\avec{w}$). 

An \emph{$\Sfiven$ weak  quasimodel for~$\varphi$ based on $\Fmf$} is a quadruple $\quasimod =
(\Fmf, \funcand, \runs,\prset)$, where $(\Fmf, \funcand)$ is an $\Sfiven$ basic structure for
$\varphi$ based on $\Fmf$,  $\runs$ is a multiset of weak runs through $(\Fmf, \funcand)$ such that~\ref{run:exists} holds for the weak runs in $\runs$, and $\prset$ is an $\Sfiven$ prototype function on $(\Fmf, \funcand)$ satisfying $\runsp\subseteq\runs$.
When clear from the context, we will omit the $\Sfiven$ qualifier for these notions.
As before, on tree-shaped frames, every quasimodel gives rise to
a weak \revised{R1.17}{quasimodel}, but we can also extract a weak quasimodel of small (exponential) size. And, conversely, any weak quasimodel can be ``saturated'' to obtain a quasimodel.

\begin{lemma}\label{lem:weakprequasi:s5}
Let  $\varphi$ be a $\QmonMLc$-sentence.
For every quasimodel satisfying $\varphi$ based on an $\SfTree{d}$ frame, there is a weak quasimodel satisfying~$\varphi$ based on an $\SfTree{d}$ frame  of size 
$2^{O(d \cdot |\varphi|)}$.
\end{lemma}
\begin{proof}
Let $\Qmf = (\mathfrak{F},\funcand,\runs)$, for a tree-shaped $\Sfiven$ frame $\Fmf=(W,\{R_{a}\}_{a\in A})$ of finite depth $d$, be a quasimodel satisfying $\varphi$.
The construction is nearly identical to the proof of Lemma~\ref{lem:weakprequasi:k}: we construct a sequence $W_0,\dots,W_{d}\subseteq W$ of sets of worlds by including witnesses $\avec{v}$ for subformulas of the form $\Diamond_a\psi$, except that a $\Diamond_a\psi$-witness is picked only for a $\avec{w}$ without a $\prec_a$-predecessor \revised{R2.16}{(so, the same frame size bounds apply)}. It follows then that $d(\avec{w}) = i$ for all $\avec{w} \in W_i$, which is required to ensure that all relevant worlds are included after $d$ steps.
\end{proof}

\begin{lemma}\label{lemma:s5:finite-approx}
Let  $\varphi$ be a $\QmonMLc$-sentence.
For every  weak quasimodel $(\Fmf,\funcand,\runs,\prset)$ satisfying $\varphi$ based on an $\SfTree{d}$ frame, 
there is a quasimodel $(\Fmf',\funcand',\runs')$ satisfying $\varphi$ based on an $\SfTree{d}$ frame\textup{;} moreover, $|\runs'| = |\runs|$ and
  $|\Fmf'| = O(|\Fmf|\cdot |\runs|^d)$ provided that $\Fmf$ and $\runs$ are finite.
\end{lemma}
\begin{proof}
Let $\Qmf = (\Fmf, \funcand, \runs,\prset{})$ be a weak quasimodel satisfying $\varphi$ based on a tree-shaped $\Sfiven$ frame $\Fmf=(W,\{R_a\}_{a\in A})$ of depth $d$. The construction is essentially the same as in  Lemma~\ref{lemma:k:finite-approx}:  
let~$W'$ be the set of all words
	\begin{equation}\label{eqveru}
		\avec{u}=(w_0,\sigma_0)a_{0}(w_1,\sigma_1)a_1 \cdots a_{m-1}(w_m,\sigma_m),
	\end{equation}
	for $\avec{w} = w_0a_0w_1a_1\cdots a_{m-1}w_m\in W$, $\sigma_0= id$ and $\sigma_{j+1}\in \textrm{Rep}(\avec{w_j})$ with $\avec{w}_j$ denoting the $j$-prefix of $\avec{w}$, for each $0\leq j<m$. 
For $\avec{u}, \avec{u}'\in W'$, we write $\avec{u} \prec'_a \avec{u}'$  if $\avec{u}'$ is an immediate $a$-successor of $\avec{u}$, that is, if $\avec{u}'$ is of the form $\avec{u}a(\avec{w}, \sigma)$.  Consider $\Fmf' = (W', \{R_a'\}_{a\in A})$, where each $R_a'$ is the smallest equivalence relation containing $\prec_a'$, for $a\in A$. Clearly, it is as required by the lemma.
	Define~$\funcand'$ by setting $\funcand'(\avec{u})= \funcand(\gen{\avec{u}})$ for any $\avec{u}\in W'$.
	The multiset~$\Rmf'$ of \emph{full} runs through $\Fmf'$ is defined as in the proof of Lemma~\ref{lemma:k:finite-approx}:
	denote by $\textit{rep}$ the function mapping any $(\rho,\ell)\in \hat{\Rmf}$ to a weak run $\rho'$  
	on $W'$ defined by 
\begin{equation*}
\rho'(\avec{u}) = \rep{\avec{u}}(\rho, \ell)(\gen{\avec{u}}), \text{ for any } \avec{u}\in W'.
\end{equation*}
Let $\Rmf'$ be  such that  $\Rmf'(\rho') = |\{ (\rho,\ell)\in\hat{\runs} \mid \textit{rep}(\rho,\ell) = \rho' \}|$, for each $\rho'$.
	
	We show that $\Qmf' = (\mathfrak{F}',\funcand', \Rmf')$  is a required quasimodel. Conditions~\ref{b1} and~\ref{run:exists} are as in the proof of Lemma~\ref{lemma:k:finite-approx}.
	It remains to prove that the elements of $\Rmf'$ are indeed runs, that is, they satisfy~\ref{rn:modal} and~\ref{rn:modal2}. Let $\Diamond_a \psi \in \sub[x]{\varphi}$. Consider $\rho'\in \Rmf'$ and $\avec{u}\in W'$ of the form~\eqref{eqveru}. We have $\rho' = \textit{rep}(\rho,\ell)$ for some $(\rho,\ell)\in\hat{\runs}$.
	
	For~\ref{rn:modal}, assume there exists $\avec{v} \in W'$ such that $\avec{u}R_a'\avec{v}$ and $\psi \in \rho'(\avec{v})$. By~\ref{s5-type}, $\Diamond_a \psi \in \rho'(\avec{v})$.
	We have to show that $\Diamond_a\psi \in \rho'(\avec{u})$.
	We distinguish the following four cases.
	\begin{enumerate}
		\item If $\avec{v}=\avec{u}$, then we are done. 
		\item If $\avec{u}\prec'_a\avec{v}$, then $\gen{\avec{u}} \prec_a \gen{\avec{v}}$.
		Since $\rep{\avec{v}}(\rho, \ell)(\gen{\avec{u}})\leftrightarrow_a \rep{\avec{v}}(\rho,\ell)(\gen{\avec{v}}) = \rho'(\avec{v})$, we get $\Diamond_a\psi\in \rep{\avec{v}}(\rho, \ell)(\gen{\avec{u}})$.  As $\rep{\avec{v}} = \sigma\circ\rep{\avec{u}}$, for some $\sigma\in\textrm{Rep}(\gen{\avec{u}})$, weak runs $\rep{\avec{v}}(\rho, \ell)$ and $\rep{\avec{u}}(\rho,\ell)$ coincide on $\gen{\avec{u}}$. So, \mbox{$\Diamond_a\psi\in \rep{\avec{u}}(\rho,\ell)(\gen{\avec{u}}) = \rho'(\avec{u})$}.
		\item If $\avec{v}\prec'_a\avec{u}$, then $\gen{\avec{v}}\prec_a\gen{\avec{u}}$.
		As $\rep{\avec{u}} = \sigma\circ\rep{\avec{v}}$, for some $\sigma\in\textrm{Rep}(\gen{\avec{v}})$, weak runs $\rep{\avec{v}}(\rho,\ell)$ and $\rep{\avec{u}}(\rho,\ell)$ coincide on~$\gen{\avec{v}}$. So, we obtain $\Diamond_a\psi\in \rep{\avec{u}}(\rho,\ell)(\gen{\avec{v}})$ from $\Diamond_a\psi\in \rho'(\avec{v}) = \rep{\avec{v}}(\rho,\ell)(\gen{\avec{v}})$.
		Since $\rep{\avec{u}}(\rho,\ell)(\gen{\avec{v}}) \leftrightarrow_a \rep{\avec{u}}(\rho,\ell)(\gen{\avec{u}})$, we have $\Diamond_a\psi\in \rep{\avec{u}}(\rho,\ell)(\gen{\avec{u}}) = \rho'(\avec{u})$. 
		\item If $\avec{v}'\prec'_a\avec{v}$ and $\avec{v}'\prec'_a\avec{u}$, for some $\avec{v}'$, then we can use the argument in Point~2 to show that $\Diamond_a\psi\in \rho'(\avec{v}')$; next, as in Point~3, we can show that $\Diamond_a\psi\in \rho'(\avec{u})$.  
	\end{enumerate}
	
	For~\ref{rn:modal2}, assume $\Diamond_a\psi \in \rho'(\avec{u})$. We have to show that there exists $\avec{v} \in W'$ such that $\avec{u}R_a'\avec{v}$ and $\psi \in \rho'(\avec{v})$.
	We have $\rho'(\avec{u})= \rep{\avec{u}}(\rho,\ell)(\gen{\avec{u}})$ and distinguish the following two cases. 	
	\begin{enumerate}
		\item If $\gen{\avec{u}}$ has no $\prec_a$-predecessor, then, by~\ref{approx:R2:s5}, there is  a prototype weak run $\prfun{\gen{\avec{u}}, t}$ for $\gen{\avec{u}}$ and $\contp = \rep{\avec{u}}(\rho,\ell)(\gen{\avec{u}})$ and a  witness 
		$\avec{w}\in W$ such that $\gen{\avec{u}} \preceq_a \avec{w}$ and $\psi \in \prfun{\gen{\avec{u}}, \contp}(\avec{w})$.
		If $\gen{\avec{u}} = \avec{w}$, then $\psi\in \contp = \rho'(\avec{u})$ and so we take $\avec{v}=\avec{u}$.	
		Otherwise, $\gen{\avec{u}} \prec_a \avec{w}$. So, let $\avec{v}\in W'$ be such that $\gen{\avec{v}} = \avec{w}$		
		 and $\rep{\avec{v}} = \sigma\circ\rep{\avec{u}}$ for 
		$\sigma\in\textrm{Rep}(\gen{\avec{u}})$ that swaps $\rep{\avec{u}}(\rho,\ell)$ with~$(\prfun{\gen{\avec{u}},\contp},0)$. We have 
		$\avec{u}R_a'\avec{v}$ and $\rep{\avec{v}}(\rho,\ell)(\gen{\avec{v}})=\prfun{\gen{\avec{u}}, \contp}(\avec{w})$, whence $\psi \in\rho'(\avec{v})$.
		\item If $\avec{w} \prec_a \gen{\avec{u}}$, for some $\avec{w}\in W$, then consider $\avec{u}'\in W'$ with $\gen{\smash{\avec{u}'}} = \avec{w}$. Since $\avec{u}' \prec_a' \avec{u}$, we have 
		 $\rep{\avec{u}}(\rho,\ell)(\gen{\smash{\avec{u}'}}) \leftrightarrow_a \rep{\avec{u}}(\rho,\ell)(\gen{\avec{u}})$, whence  $\Diamond_a\psi\in \rep{\avec{u}}(\rho,\ell)(\gen{\smash{\avec{u}'}})$. As $\rep{\avec{u}}(\rho,\ell)$   coincides with  $\rep{\smash{\avec{u}'}}(\rho,\ell)$ on~$\gen{\smash{\avec{u}'}}$, we then obtain $\Diamond_a\psi\in \rep{\smash{\avec{u}'}}(\rho,\ell)(\gen{\smash{\avec{u}'}})$.
		We now apply the argument in Point~1 to~$\gen{\smash{\avec{u}'}}$, which has no $\prec_a$-predecessor,  and obtain
		$\avec{v}\in W'$ with $\avec{u}'R'_a\avec{v}$ and $\psi\in \rho'(\avec{v})$.
		Since $R'_a$ is an equivalence relation, 
		$\avec{u}R_a'\avec{v}$, as required. 
\end{enumerate}
This completes the proof of Lemma~\ref{lemma:s5:finite-approx}.	
\end{proof}

\section{Decidability for Monodic Fragments in Constant Domains}

We have reduced the problem of deciding the existence of a model satisfying formula 
$\varphi$ to the problem of deciding the existence of a weak quasimodel satisfying $\varphi$ 
based on a frame of exponential size. In general, however, the latter problem is not yet trivially decidable, since we have no bound on the size of quasistates (and so also no bound on the number of weak runs). In fact, as~$\CT$ does not have the finite model property, no finite bound exists for  $\CtwomonMLc$, for example. 

For some languages with limited counting, in particular the one-variable and guarded fragments, we show that one can directly 
use weak quasimodels to obtain tight complexity bounds. For the two-variable fragment with counting we develop further machinery. 

We remind the reader that, due to Theorem~\ref{thm:reductions}, the results in this and the following section are applicable to the validity and global consequence problems in total and partial interpretations and also to the languages with closed definite descriptions, even if it is not stated in the claims explicitly.

\subsection{One-Variable Fragment}\label{sec:one-variable}

We first consider a $\QoneMLc$-sentence $\varphi$ and begin with a characterisation of its quasistates. We assume that it contains a subformula $x=c$ for any constant~$c\in\Ind(\varphi)$; this can trivially be achieved by adding conjuncts of the form $\exists x\,(x = c)$ to the sentence. We then say that a type~$\contp$ for $\varphi$ \emph{contains a constant} if it contains $x = c$, for some $c\in\Ind$; note that the same type can contain multiple $x = c_i$. Then quasistate candidates $\avec{n}$ and $\avec{n}'$ for $\varphi$ (not necessarily realisable) are \emph{simply compatible} (written $\avec{n} \sim_0 \avec{n}'$) if 
	\begin{itemize}
	\item  $\avec{n}'(\contp)=\avec{n}(\contp)=1$, for all $\contp$ containing a constant and 
	\item  $\avec{n}'(\contp)=0$ iff $\avec{n}(\contp)=0$, for all types $\contp$ for $\varphi$. 
	\end{itemize}
It can be easily seen that this condition defines a closure property for quasistates in $\QoneMLc$:
\begin{lemma}\label{lem:qs:closure:1}
 If $\avec{n}$ is a quasistate for a $\QoneMLc$-sentence $\varphi$, then any quasistate candidate~$\avec{n}'$ for $\varphi$ with $\avec{n}\sim_0\avec{n}'$ is realisable \textup{(}thus, is a quasistate for $\varphi$\textup{)}. 
\end{lemma}	
It follows from Lemma~\ref{lem:qs:closure:1}, that in any weak quasimodel   $(\Fmf,\funcand,\Rmf,\prset{})$ for a \QoneMLc-sentence $\varphi$,
we can assume that the multiplicity of each type in quasistates matches the multiplicity required by the set $\runsp$ of weak runs  induced by the prototype function $\prset{}$.
Since $|\runsp|$ is linear in the size of the frame and exponential in the size of~$\varphi$,
we obtain a tight complexity result for  \QoneMLc.
\begin{theorem}\label{thm:complexity:onevar}
 For constant domains, $\Kn$-validity and $\Sfiven$-validity in $\QoneMLc$ are \textup{co\NExpTime}-complete. In fact, every satisfiable sentence is satisfiable in an interpretation of exponential size.
 \end{theorem}
\begin{proof}
We consider $\Kn$, the proof for $\Sfiven$ is similar. Assume $\varphi$ is satisfiable in some $\Kn$ frame. By Lemmas~\ref{lem:k:tree-shaped},~\ref{lemma:quasimodel}  and~\ref{lem:weakprequasi:k},
	there is a weak quasimodel $\Qmf = (\Fmf,\funcand,\runs,\prset)$ satisfying $\varphi$ based on a $\Tree{d(\varphi)}$ frame $\Fmf = (W, \{R_a\}_{a\in A})$ with~$|W| \leq 2^{O(d(\varphi)\cdot |\varphi|)}$.
	\revised{R2.17}{We set  $\funcand^{\mathfrak{p}}(\avec{w},\contp)=|\runsp_{\avec{w},\contp}|$, for each $\avec{w}\in W$ and type~$t$. As $\funcand(\avec{w}) \sim_0 \funcand^{\mathfrak{p}}(\avec{w})$, by Lemma~\ref{lem:qs:closure:1}, the~$\funcand^{\mathfrak{p}}(\avec{w})$ are quasistates for~$\varphi$.} Therefore, $\Qmf^{\mathfrak{p}} = (\Fmf,\funcand^{\mathfrak{p}},\runsp,\prset{})$ is a weak quasimodel satisfying~$\varphi$. Since \mbox{$|\runsp| \leq |W| \cdot 2^{|\varphi|}$}, by Lemmas~\ref{lemma:k:finite-approx} and~\ref{lemma:quasimodel}, we obtain an interpretation satisfying $\varphi$ based  on an exponential-size frame with exponential-sized domains.
	The co\NExpTime upper bound for validity is then immediate. The matching lower bound follows from  \NExpTime-hardness of satisfiability of the one-variable modal logics $\K$ and $\Sfive$ in constant domains, even without equality and constants~\cite{DBLP:journals/logcom/Marx99}.
	\end{proof}

We note that the \coNExpTime-hardness result~\cite{DBLP:journals/logcom/Marx99}, proved in the context of product modal logics, is contingent on the domains of the models being constant and does not translate to expanding domains. Indeed, we show in Section~\ref{sec:pspace} that $\QoneMLc$-validity on $\Kn$ frames in \emph{expanding} domains is \PSpace-complete. 

\revised{R1.2}{We use Theorem~\ref{thm:complexity:onevar} to resolve an open problem posed by Fitting~\cite{Fit02}, who studies decidability of modal first-order logics equipped with a \emph{$\lambda$-abstraction} operator but without first-order quantifiers. The language allows predicates of arbitrary arity as well as non-rigid constants. The $\lambda$-abstraction operator $\langle \lambda x.\varphi\rangle(\tau)$ serves as a scoping device that distinguishes between \emph{de re} and \emph{de dicto} readings of modalities; see also~\cite{FitMen12,StaTho68}. It is shown~\cite{Fit02} that validity in the resulting system $\mathbf{S5}\lambda$ is undecidable, while the decidability of validity in its one-variable fragment is left open.}
	
\revised{R1.2}{This question can be answered positively, with a $\coNExpTime$ upper bound, by embedding the one-variable fragment of $\mathbf{S5}\lambda$ into $\QoneMLc$. Indeed, the natural translation
\begin{equation*}
\langle \lambda x.\varphi\rangle(c)
\quad\mapsto\quad
\exists x\,\bigl((x=c) \wedge \varphi\bigr)
\end{equation*}
from the one-variable fragment of $\mathbf{S5}\lambda$ to $\QoneMLc$ is easily seen to preserve and reflect validity. Note that  predicates of arbitrary arity can, without loss of generality, be replaced by unary predicates.}

\subsection{Guarded Fragment}\label{sec:guarded}

We next show that $\Kn$- and $\Sfiven$-validity in $\QGuardmonMLc$  is 2\ExpTime-complete.
Let~$\varphi$ be a $\QGuardmonMLc$-sentence. We again assume it contains a subformula $x=c$ for any constant $c\in\Ind(\varphi)$.
For quasistate candidates $\avec{n}$ and $\avec{n}'$ for~$\varphi$, define the product ordering as usual by setting $\avec{n}\leq \avec{n}'$ if $\avec{n}(\contp)\leq \avec{n}'(\contp)$ for all types $\contp$.
We set $\avec{n}\leq_{\textit{gf}}\avec{n}'$ if $\avec{n}\leq \avec{n}'$ and $\avec{n}\sim_0 \avec{n}'$.
We require the following properties of the guarded fragment~\cite{DBLP:journals/corr/BaranyGO13,DBLP:journals/jsyml/Gradel99,DBLP:journals/apal/Hodkinson06}.

\begin{lemma}[Guarded Fragment]\label{lem:guardedfragment} 
Let $\varphi$ be a $\QGuardmonMLc$-sentence.

	\textup{(1)} If $\avec{n}$ is a quasistate for $\varphi$, then any quasistate candidate $\avec{n}'$ with 
	$\avec{n}\leq_{\textit{gf}}\avec{n}'$ is also a quasistate for $\varphi$.
	
	\textup{(2)} Deciding if a quasistate candidate $\avec{n}$ for $\varphi$ is realisable is in \textup{2\ExpTime} \textup{(}in~$|\varphi|$\textup{)} provided that either $\avec{n}(\contp)\leq \smash{2^{|\varphi|^{O(1)}}}$ or $\avec{n}(\contp)=\aleph_0$, for all~$\contp\in\avec{n}$.
\end{lemma}
\begin{proof}
	(1) It suffices to show that, for every $\contp$ not containing a constant and  with $\avec{n}(\contp)>0$, the quasistate candidate $\avec{n}'$ defined in the same way as~$\avec{n}$ except that $\avec{n}'(\contp)=\avec{n}(\contp)+1$ is also a quasistate for $\varphi$.  Consider a first-order structure $\Bmf$ realising $\avec{n}$ and, in particular, $\contp$ by domain element $d$. Define a new first-order structure $\Bmf'$ by adding a fresh domain element $d'$ to it and adding a tuple to any $R$ if it is obtained from a tuple in $R$ by replacing uniformly $d$ by $d'$. The first-order structures $\Bmf'$ and $\Bmf$ are guarded bisimilar~\cite{ANvB98}, and therefore, $\avec{n}'$ is realised in~$\Bmf'$. One can deal with the case of $\avec{n}'(\contp) = \aleph_0$ in a similar way.
			
\smallskip	
		
	(2) We assume that no type $\contp$ with $\avec{n}(\contp) = \aleph_0$ contains constants (otherwise, such $\avec{n}$ would trivially be non-realisable). Then we consider the following guarded sentence~$\psi$: 
\begin{equation*}
\psi \ \ = \ \ \textsf{only}_{\avec{n}} \ \ \land \bigwedge_{\begin{subarray}{c}\contp\in\avec{n}\\\avec{n}(\contp) < \aleph_0\end{subarray}} \bigl(\textsf{atleast}^{\avec{n}(\contp)}_{\contp}\land \textsf{atmost}^{\avec{n}(\contp)}_{\contp}\bigr), 
\end{equation*}	
where
\begin{align}	
	\label{eq:gf:only}
	 \textsf{only}_{\avec{n}} \ \ & = \ \ \forall x\,\bigvee_{\avec{n}(\contp) > 0}\bar{\contp}(x),\\
	\label{eq:gf:atleast}
	\textsf{atleast}^{n}_{\contp} \ \ & = \ \  \bigwedge_{1\leq i \leq n} \bar{\contp}(c_{\contp,i}) \ \land \hspace*{0em}\bigwedge_{1 \leq i < j \leq n} \hspace*{-0.5em}(c_{\contp, i} \ne c_{\contp, j}),\\
	\label{eq:gf:atmost}
	\textsf{atmost}^{n}_{\contp} \ \ &= \ \ \forall x\,\bigl(\bar{\contp}(x) \to \hspace*{-0.5em}\bigvee_{1 \leq i \leq n} (x = c_{\contp,i})\bigr),	 
\end{align}
	 for fresh constant symbols $c_{\contp,i}$.
	 It should be clear that $\psi$ is satisfiable iff $\avec{n}$ is realisable; in particular, due to  Point (1) we can always ensure that types with $\avec{n}(\contp) = \aleph_0$ are realised by infinitely many domain elements.
	 Observe that $\psi$ contains $2^{\smash{|\varphi|^{O(1)}}}$ constants, but the same variables and the same predicate symbols as~$\varphi$. It follows that its satisfiability can be decided in time $2^{\smash{O(|\varphi|) \cdot (l +k)^{k}}}$, where $k$ is the number of variables and $l$ is the number of constants in $\psi$; see, e.g.,~\cite[Proposition~3.5]{DBLP:journals/apal/Hodkinson06}. Thus, we obtain the required complexity bound.
\end{proof}

%
%
%
%
\begin{theorem}\label{thm:guarded:complexity}
	For constant domains, $\Kn$-validity and $\Sfive_{n}$-validity in $\QGuardmonMLc$ are \textup{2\ExpTime}-complete. 
\end{theorem}
\begin{proof}
	We consider validity on $\Kn$ frames with constant domains; the proof for $\Sfiven$ is similar. 2\ExpTime-hardness follows from the complexity of~$\textsf{GF}$~\cite{DBLP:journals/jsyml/Gradel99}.

	Assume $\varphi$ is given. By Lemmas~\revised{R1.18}{\ref{lem:k:tree-shaped}, \ref{lemma:quasimodel}, \ref{lem:weakprequasi:k}} and~\ref{lemma:k:finite-approx}, 
	$\varphi$ is satisfiable in some $\Kn$ frame iff there is a weak quasimodel $\Qmf = (\Fmf,\funcand,\runs,\prset{})$ satisfying $\varphi$ based on a $\Tree{d(\varphi)}$ frame  of size~$2^{O(d(\varphi)\cdot |\varphi|)}$.
	Let $\runs_{c}$ be the multiset of weak runs $\rho\in\runs$ such that $\rho(\avec{w})$ contains a constant for some $\avec{w}\in W$, in which case, as $\funcand(\avec{w})$ is realisable, we have $\funcand(\avec{w}, \rho(\avec{w})) = 1$. It follows that $|\runs_c|\leq |W| \times |\varphi|$.  Consider also a set $\runs_\infty$ that contains a \emph{single} weak run $\rho\in\runs_{\avec{w},\contp}\setminus \runs_c$ if it exists, for each $\avec{w}\in W$ and $\contp\in\funcand(\avec{w})$. Observe that  $|\runsp|, |\runs_\infty| \leq |W| \times 2^{|\varphi|}$. By Lemma~\ref{lem:guardedfragment}~(1), the tuple $\Qmf' = (\Fmf,\funcand',\runs',\prset{})$ defined by setting
\begin{align*}
		\funcand'(\avec{w},\contp) & = \begin{cases}
			 \funcand(\avec{w},\contp),& \text{if $\runs_{\avec{w},\contp}\subseteq \runs_{c}$},\\
			\aleph_{0}, & \text{otherwise},
		\end{cases}
	     \text{ for each } \avec{w}\in W \text{ and each type } \contp,\\
		\runs'(\rho) & = \begin{cases}
			\runs(\rho), & \text{for each $\rho\in \runs_{c}$},\\
			\aleph_{0}, & \text{for each $\rho\in \runs_\infty\cup\runsp$ with $\rho\notin \runs_c$},\\
			0, & \text{otherwise},
		\end{cases}
\end{align*}
	is again a weak quasimodel
	satisfying $\varphi$. In $\Qmf'$, we have, for any $\avec{w}\in W$ and type~$\contp$, either $\funcand'(\avec{w},\contp)\leq |W| \times |\varphi|$ or $\funcand'(\avec{w},\contp)=\aleph_{0}$. We can enumerate in double exponential time all such
	$\Qmf'$ if we admit $\funcand'(\avec{w})$ to be (not necessarily realisable) quasistate candidates for~$\varphi$. Next, we can use Lemma~\ref{lem:guardedfragment}~(2) to check in double exponential time for each $\avec{w}$ whether $\funcand'(\avec{w})$ is realisable and so $\Qmf'$ is a quasimodel. Thus, we obtain a decision procedure in double exponential time.
	%
	%
	%
\end{proof}

\subsection{Two-Variable Fragment with Counting and Presburger Arithmetic}\label{sec:2var}

\revised{R2.18}{To prove decidability and obtain tight complexity bounds for 
the two-variable fragment with counting, we now define weak pre-quasimodels, which, similarly to Section~\ref{sec:one-variable}, replace the multiset of weak runs $\runs$ with the \emph{set}~$\runsp$ of  locally saturated weak runs. This, however, gives us only a \emph{lower bound} on the number of runs required by~\ref{run:exists}, and so, for the two-variable fragment with counting, we will also require some additional local constraints on the quasistates in order to reconstruct the multiset of weak runs satisfying~\ref{run:exists}. The existence of a weak pre-quasimodel  
together with known, or easily derived, characterisations of quasistates 
can then be translated into constraints in decidable fragments of arithmetic over natural numbers extended with infinity~$(\aleph_0)$.}

Let $\varphi$ be a $\QmonMLc$-sentence and $\Fmf = (W, \{R_a\}_{a\in A})$ a tree-shaped $\Kn$ frame with $|\Fmf| > 1$. 
For uniformity with $\Sfiven$ tree-shaped frames, we shall use the $a$-successor relation~$\prec_a$, for $a\in A$, defined by taking $\avec{w} \prec_a \avec{w}'$ iff $\avec{w}'$ is of the form $\avec{w}aw$, for some~$w$; clearly, $R_a$ coincides with $\prec_a$  for tree-shaped $\Kn$ frames; in $\Sfiven$, each $R_a$ is the reflexive, symmetric and transitive closure of~$\prec_a$.

\revised{R2.18}{Consider a  \emph{constant-domain}  weak quasimodel   $(\Fmf,\funcand,\runs,\prset{})$ for $\varphi$ based on~$\Fmf$. We replace the multiset $\runs$ of weak runs with a more local description in terms of links.
A \emph{link} between $\avec{w}\in W$ and $\avec{v}\in W$ with $\avec{w}\prec_a\avec{v}$ for $a\in A$
is a function $L_{\avec{w},\avec{v}}$ that assigns a multiplicity  $L_{\avec{w},\avec{v}}(\contp_1,\contp_{2})\in\extN$ to any pair $\contp_1,\contp_2$ of types for~$\varphi$; note that $a$ is uniquely determined.
Intuitively, it says that $L_{\avec{w},\avec{v}}(\contp_1,\contp_{2})$-many weak runs simultaneously have type $t_1$ in $\funcand(\avec{w})$  and type $t_2$ in $\funcand(\avec{v})$. So, given a multiset $\runs$ of weak runs, we can define links between $\avec{w}$ and $\avec{v}$ by taking 
\begin{equation}\label{eq:links:canonical}
 L_{\avec{w},\avec{v}}(t_1,t_2) \ \ = \ \ 
 |\runs_{\avec{w},\contp_1,\avec{v},\contp_2}|,\qquad \text{for all types $t_1$, $t_2$ for  $\varphi$,}
 \end{equation}
where $\runs_{\avec{w},\contp_1,\avec{v},\contp_2}$ denotes the multiset of weak runs $\rho\in\runs$ such that $\rho(\avec{w})=\contp_1$ and~\mbox{$\rho(\avec{v})=\contp_2$}. In Lemma~\ref{lem:weakpreweak} below we show the converse: given a set of links matching a basic structure $(\Fmf, \funcand)$ and a prototype function $\prset{}$ on $(\Fmf, \funcand)$, we can construct a multiset $\runs$ of weak runs in a weak quasimodel $(\Fmf, \funcand, \runs, \prset{})$ for~$\varphi$.}

A \emph{$\Kn$ weak pre-quasimodel for~$\varphi$} is a quadruple $\quasimod =
(\Fmf, \funcand, \prset{}, \Lmf)$ such that $(\Fmf, \funcand)$ is a basic structure for $\varphi$,
$\prset$ is a prototype function on $(\Fmf, \funcand)$ with full weak runs, and $\Lmf$ is a set of links $L_{\avec{w},\avec{v}}$ between~$\avec{w}\in W$ and~$\avec{v}\in W$ with $\avec{w}\prec_{a} \avec{v}$ for~$a\in A$,  satisfying the following:
\begin{enumerate}
	[label=\textbf{\bfseries (link$_{\theenumi}$)},leftmargin=*,series=run]
\item\label{eq:links-card-1} $\sum_{\contp}L_{\avec{w},\avec{v}}(\contp_{1},\contp)=\funcand(\avec{w}, \contp_1)$,  for all types $\contp_1$ for $\varphi$,
\item\label{eq:links-card-2} $\sum_{\contp}L_{\avec{w},\avec{v}}(\contp,\contp_{2})=\funcand(\avec{v}, \contp_2)$,  for all types $\contp_2$ for $\varphi$,
	\item\label{eq:links-coh} $\contp_1\rightarrow_a\contp_2$ whenever $L_{\avec{w},\avec{v}}(\contp_{1},\contp_{2})>0$, for all types $\contp_1, \contp_2$ for $\varphi$,
	\item\label{eq:links}
$L_{\avec{w},\avec{v}}(\contp_{1},\contp_{2}) \geq |\runsp_{\avec{w},\contp_1,\avec{v},\contp_2}|$, for all types $\contp_1, \contp_2$ for $\varphi$;
\end{enumerate}
note that $\runsp_{\avec{w},\contp_1,\avec{v},\contp_2}$ 
contains no duplicates since $\runsp$ is a set. We say $\Qmf$ satisfies~$\varphi$ if~\ref{b1} holds.

\begin{lemma}\label{lem:weakpreweak}
	Let $\Fmf$ be a finite tree-shaped~$\Kn$ frame with $|\Fmf| > 1$.  For any basic structure $(\Fmf,\funcand)$ based on $\Fmf$ and any prototype function $\prset{}$ on $(\Fmf, \funcand)$ with full runs, the following conditions are equivalent\textup{:}
	\begin{enumerate}
		\item there is a constant-domain $\Kn$ weak quasimodel $(\Fmf,\funcand,\runs,\prset{})$ satisfying~$\varphi$\textup{;}
		\item there is a $\Kn$ weak pre-quasimodel $(\Fmf,\funcand,\prset{},\Lmf)$ satisfying $\varphi$.
	\end{enumerate}
\end{lemma}
\begin{proof}
 (1) $\Rightarrow$ (2) Assume a weak quasimodel $(\Fmf,\funcand,\Rmf,\prset{})$ is given. To define a weak pre-quasimodel, we take $\Lmf$ that contains, for each $\avec{w},\avec{v}\in W$ with $\avec{w}\prec_{a}\avec{v}$ and $a\in A$, a link $L_{\avec{w},\avec{v}}$ defined by~\eqref{eq:links:canonical}.
 It is readily checked that $(\Fmf,\funcand,\prset{},\Lmf)$ is a $\Kn$ weak pre-quasimodel satisfying $\varphi$.
 
 \smallskip

 (2) $\Rightarrow$ (1) Let $(\Fmf,\funcand,\prset{},\Lmf)$ be a weak pre-quasimodel \revised{R1.19}{satisfying} $\varphi$.
 To define a weak quasimodel $(\Fmf,\funcand,\Rmf,\prset{})$ satisfying $\varphi$, we extend the set $\runsp$ to a multiset~$\runs$ of weak runs so that~\ref{run:exists} holds.
 To this end, we take an enumeration $(\avec{w}_1,\contp_1,\ell_1), (\avec{w}_2,\contp_2,\ell_2),\dots$ of the required weak runs through each type, where 
 $\avec{w}_i\in W$, $\contp_i\in\funcand(\avec{w}_i)$ and $0 \leq \ell_i< \funcand(\avec{w}_{i},\contp_{i})$, for all $i > 0$.
  
 Assume we have constructed $\smash{\Rmf^0}\subseteq \cdots \subseteq \smash{\Rmf^{i-1}}$, where \mbox{$\smash{\Rmf^0}= \runsp$, so that $\smash{\Rmf^{i-1}}$} contains
sufficiently many weak runs $\rho_j$ through each $(\avec{w}_j, \contp_j)$ of the first $(i-1)$ members of the enumeration and, for all $\avec{w},\avec{v}\in W$ with~$\avec{w}\prec_a\avec{v}$, $a\in A$,
\begin{enumerate}
	[label=\textbf{\bfseries (link$^{i-1}_{4}$)},leftmargin=*,series=run]
	\item\label{eq:links:im1}
$L_{\avec{w},\avec{v}}(\contp_{1},\contp_{2}) \geq |\Rmf^{i-1}_{\avec{w},\contp_1,\avec{v},\contp_2}|$, for all types $\contp_1, \contp_2$ for $\varphi$.
\end{enumerate}%
Observe that, as $|\Fmf| > 1$, by~\ref{eq:links-card-1} and~\ref{eq:links-card-2}, condition~\ref{eq:links:im1} implies the following weakening of~\ref{run:exists}: 
\begin{equation*}
\funcand(\avec{w},\contp) \geq  |\Rmf^{i-1}_{\avec{w}, \contp}|, \text{ for every } \avec{w}\in W \text{ and every type } \contp \text{ for }\varphi.
\end{equation*}
Consider the $i$th member $(\avec{w}_i,\contp_i,\ell_i)$ of the enumeration. 
If  $\funcand(\avec{w}_i,\contp_i) = |\Rmf^{i-1}_{\avec{w}_i,t_i}|$, then we set $\Rmf^i=\Rmf^{i-1}$, thus satisfying~\textbf{(link$_4^i$)}.
Otherwise,  
 $\funcand(\avec{w}_i,\contp_i) > |\Rmf^{i-1}_{\avec{w}_i, \contp_i}|$, and we define a weak run $\rho_i$ with $\rho_i(\avec{w}_i)=\contp_i$ so that the result $\smash{\Rmf^i}$ of extending~$\smash{\Rmf^{i-1}}$ with a fresh copy of $\rho_i$ satisfies~\textbf{(link$_4^i$)}. 
 We define $\rho_i$ by induction on the structure of the $\prec_a$, for $a\in A$, starting from $W^0=\{\avec{w}_i\}$ and $\rho_i(\avec{w}_i)=\contp_i$. At every step of the construction, we ensure that
 \begin{equation}\label{run:exists3-strict}
\funcand(\avec{w},\contp) >  |\Rmf^{i-1}_{\avec{w}, \contp}|, \text{ where }  \contp = \rho_i(\avec{w}), \text{ for every $\avec{w}\in W^j$}.
 \end{equation}  
This is clearly the case for $W^0$.  Suppose $\rho_i$ has been defined on $W^j$, for $j \geq 0$. We have the following cases.
\begin{itemize}
\item If $\avec{w}\in W^j$ but $\avec{v}\notin W^j$ with \mbox{$\avec{w}\prec_{a}\avec{v}$}, then we consider $\contp = \rho_i(\avec{w})$.
By~\eqref{run:exists3-strict} and~\ref{eq:links-card-1}, we can find a type $\contp'$ for~$\varphi$ with 
$L_{\avec{w},\avec{v}}(\contp,\contp') > |\Rmf^{i-1}_{\avec{w},\contp, \avec{v},\contp'}|$.
We extend $\rho_i$ to $W^{j+1} = W^j \cup \{\avec{v}\}$ by setting $\rho_i(\avec{v})=\contp'$. Then~\eqref{run:exists3-strict} holds for $\avec{v}$ due to~\ref{eq:links:im1} and~\ref{eq:links-card-2}.
 
\item If  $\avec{w}\in W^j$ but $\avec{v}\notin W^j$ with $\avec{v}\prec_{a}\avec{w}$, then the definition of $\rho_i$ on $\avec{v}$ is similar and left to the reader.
\end{itemize}
Observe that when $\rho_i$ is defined on $W$, it is a weak run due to~\ref{eq:links-coh}.
 
It should be clear that~\ref{run:exists} holds for the limit $\Rmf$ of $\Rmf^0\subseteq\Rmf^1\subseteq\cdots$, and
so  $(\Fmf,\funcand,\Rmf,\prset{})$ is a weak quasimodel satisfying $\varphi$.
 \end{proof}
 
 The above construction can be easily adapted for an $\Sfiven$ tree-shaped frame $\Fmf = (W, \{R_a\}_{a\in A})$. 
An \emph{$\Sfiven$ weak pre-quasimodel for~$\varphi$} is a quadruple $\quasimod =
(\Fmf, \funcand, \prset{}, \Lmf)$ such that $(\Fmf, \funcand)$ is an $\Sfiven$ basic structure for $\varphi$ based on $\Fmf$,
$\prset{}$ is an $\Sfiven$ prototype function on $(\Fmf,\funcand)$, and
$\Lmf$ is a set of links $L_{\avec{w},\avec{v}}$ between $\avec{w}\in W$ and~\mbox{$\avec{v}\in W$} with $\avec{w}\prec_{a} \avec{v}$ and $a\in A$, satisfying~\ref{eq:links-card-1}, \ref{eq:links-card-2}, \ref{eq:links} and  
the following counterpart of~\ref{eq:links-coh}: 
\begin{enumerate}
	[label=\textbf{\bfseries (link$_3$-s5)},leftmargin=*,series=run]
	\item\label{eq:links-coh:s5} $\contp_1\leftrightarrow_a\contp_2$ whenever $L_{\avec{w},\avec{v}}(\contp_{1},\contp_{2})>0$, for all types $\contp_1, \contp_2$ for $\varphi$.
\end{enumerate}
We then have the following analogue of Lemma~\ref{lem:weakpreweak}:
\begin{lemma}\label{lem:weakpreweak:s5}
Let $\Fmf$ be a finite tree-shaped~$\Sfiven$ frame with $|\Fmf| > 1$.  For any $\Sfiven$ basic structure~$(\Fmf,\funcand)$ based on $(\Fmf, \funcand)$ and any $\Sfiven$ prototype function $\prset{}$ on~$(\Fmf, \funcand)$, the following conditions are equivalent\textup{:}
	\begin{enumerate}
		\item there is an $\Sfiven$ weak quasimodel $(\Fmf,\funcand,\runs,\prset{})$ satisfying~$\varphi$\textup{;}
		\item there is an $\Sfiven$ weak pre-quasimodel $(\Fmf,\funcand,\prset{},\Lmf)$ satisfying $\varphi$.
	\end{enumerate}
\end{lemma}

\revised{R2.19}{We are in a position now to give rather quick decidability proofs for monodic fragments  over $\Kn$ and $\Sfiven$ frames that admit encoding of quasistates in a decidable but powerful fragment of arithmetic such as Presburger Arithmetic with infinity~\cite{DBLP:conf/casc/LasarukS09}.} We consider $\Kn$ with constant domains; the approach for~$\Sfiven$ is similar. Fix a sentence $\varphi$. The idea is to decide the existence of a model of $\varphi$ by encoding the existence of its weak pre-quasimodel in Presburger Arithmetic with infinity. To implement this idea, we require some notation for numbers. By $\extN^{k}$ we denote the set of $k$-tuples over~$\extN$. 
We fix an ordering $\contp_{1},\ldots,\contp_{k_{\varphi}}$ of the types for $\varphi$. Then a quasistate candidate $\avec{n}$ for $\varphi$ can be identified with a $k_\varphi$-tuple $\avec{n}=(n_{1},\ldots,n_{k_{\varphi}})\in \extN^{\smash{k_{\varphi}}}$ by setting $n_{i}=\avec{n}(\contp_{i})$, for $1\leq i \leq k_{\varphi}$. In this case we also say the tuple \emph{represents} the quasistate candidate; the same applies to quasistates.
Now, by Lemmas~\revised{R1.18}{\ref{lem:k:tree-shaped}, \ref{lemma:quasimodel}, \ref{lem:weakprequasi:k}} and~\ref{lemma:k:finite-approx},
$\varphi$ is satisfiable in some $\Kn$ frame iff there is a weak
quasimodel $\Qmf = (\Fmf,\funcand,\runs,\prset{})$
satisfying~$\varphi$ based on a $\Tree{d(\varphi)}$ frame $\Fmf = (W,
\{R_a\}_{a\in A})$ with~$|W| \leq 2^{O(d(\varphi)\cdot
  |\varphi|)}$. If $|W| > 1$, then, by Lemma~\ref{lem:weakpreweak}, this is the case iff
there is a weak pre-quasimodel $\quasimod' =
(\Fmf, \funcand, \prset{}, \Lmf)$ satisfying $\varphi$. Observe that $\prset{}$ contains at most
$|W|\times 2^{|\varphi|}$ weak runs. It follows that the only objects in~$\quasimod'$ that do not have 
an exponential upper bound are the quasistates~$\funcand(\avec{w})$ and the links~$\Lmf$.

Assume now that we have a monodic fragment $\Lmc$ such that we can construct, for any sentence $\varphi\in \Lmc$, a formula $\textit{Pres}_{\varphi}(x_{1},\ldots,x_{k_{\varphi}})$ in Presburger Arithmetic with infinity (PAI) such that
\begin{enumerate}
	[label=\textbf{\bfseries (encode)},leftmargin=*,series=run]
	\item\label{eq:encodePres} a $k_{\varphi}$-tuple $\avec{n}\in\extN^{\smash{k_{\varphi}}}$ represents a quasistate for $\varphi$ iff PAI $\models \textit{Pres}_{\varphi}(\avec{n})$.  
\end{enumerate}
\revised{R2.19}{This is known, for instance, if $\Lmc$ is $\CtwomonMLc$, its extension with a single equivalence relation, or its extension with ultimately periodic counting~\cite{pratt2023fragments,DBLP:journals/siamcomp/BenediktKT24}.} 
To decide satisfiability of $\varphi$ by checking the existence of a weak pre-quasimodel~$\quasimod$, guess finite $\Fmf$ and~$\prset{}$, 
and take, for each $\avec{w}\in W$, individual variables
$x_{1}^{\avec{w}},\ldots,x_{\smash{k_{\varphi}}}^{\avec{w}}$ and the formulas $\textit{Pres}_{\varphi}( x_{1}^{\avec{w}},\ldots,x_{\smash{k_{\varphi}}}^{\avec{w}})$, together with the conjunction~$\textit{pqm}_{\Fmf,\prset}$ of
\begin{align*}
\sum_{\begin{subarray}{c}1 \leq i \leq k_\varphi \\ \varphi\in\contp_i\end{subarray}}  x^{\avec{w}_0}_i & \geq 1, \tag{\textbf{f-b}}\\
x^{\avec{w}}_{i} & = 0, && \text{ if } \prfun{\avec{w},\contp_i} \text{ is undefined}, \text{ for } \avec{w} \in W,\ 1\le i\le k_\varphi, \tag{\textbf{f-wq-sat}}
\end{align*}
and the following equalities and inequalities with individual variables $x_{\avec{w},i,\avec{v},j}$ representing $L_{\avec{w},\avec{v}}(\contp_{i},\contp_{j})$, for every $\avec{w}, \avec{v}\in W$ with~$\avec{w}\prec_a \avec{v}$ and $a\in A$:
\begin{align*}
	\sum_{1 \leq i \leq k_{\varphi}}x_{\avec{w},i,\avec{v},j} & =x_{j}^{\avec{v}}, && \text{for } 1\leq j \leq k_{\varphi}, \tag{\textbf{f-link$_1$}}\\
	\sum_{1 \leq j \leq k_{\varphi}}x_{\avec{w},i,\avec{v},j} & =x_{i}^{\avec{w}}, &&\text{for } 1\leq i \leq k_{\varphi}, \tag{\textbf{f-link$_2$}} \\
	x_{\avec{w},i,\avec{v},j} & =0, && \text{if } \contp_{i} \not\rightarrow_{a} \contp_{j}, \text{ for } 1\leq i,j \leq k_{\varphi},\tag{\textbf{f-link$_3$}}\\
	x_{\avec{w},i,\avec{v},j} & \geq |\runsp_{\avec{w},\contp_i,\avec{v},\contp_j}|, && \text{for } 1\leq i,j \leq k_{\varphi};\tag{\textbf{f-link$_4$}}
\end{align*}
note that the labels of these formulas refer to the respective conditions. In particular, the condition in~\textbf{(f-wq-sat)} would be equivalent to $\contp_i\notin\funcand(\avec{w})$ in~\ref{approx:R2}, if we had a quasistate $\funcand(\avec{w})$  description available (but we only have $\prfun{\avec{w},\contp_i}$).
Then, by the definition of the above formulas, we have the following equivalence.

\begin{lemma}\label{lem:pres-based}
There exists a weak pre-quasimodel for $\varphi$ based on a given frame $\Fmf$ with $|\Fmf| > 1$ and a prototype function $\prset{}$ iff $\bigwedge_{\avec{w}\in W}\textit{Pres}_{\varphi}( x_{1}^{\avec{w}},\ldots,x_{k_{\varphi}}^{\avec{w}}) \wedge \textit{pqm}_{\Fmf,\prset}$ is satisfiable.
\end{lemma} 
It follows that $\Kn$- and $\Sfiven$-validity in $\CtwomonMLc$ and the extensions mentioned above are decidable.
\revised{R2.19}{This straightforward general encoding, however, is not sufficient to obtain tight complexity upper bounds.
In fact, to show a tight bound for the two-variable fragment with \revised{R2.20}{counting}, a more subtle combination with the upper bound proofs for the first-order fragments is needed.} We introduce some notation for inequalities.
\revised{R1.7}{\emph{Linear extended-Diophantine equations and inequalities}~(see e.g.,~\cite[Section 7.4]{pratt2023fragments}) take the form
\begin{equation*}
a_{1}x_{1}+\cdots+a_{k}x_{k}=b, \quad a_{1}x_{1}+\cdots + a_{k}x_{k}\leq b, \quad
a_{1}x_{1}+\cdots + a_{k}x_{k}\geq b
\end{equation*}
with coefficients $a_{1},\ldots,a_k,b\in\extN$; as usual, $m \leq \aleph_0$, for all $m\in\extN$, $\aleph_0 \cdot 0 = 0 \cdot \aleph_0  = 0$,  $\aleph_0 \cdot m = m \cdot \aleph_0  = \aleph_0$, for all $m\in\mathbb{N}\setminus\{0\}$, and $\aleph_0 \cdot \aleph_0 = \aleph_0$, which, together with the assumptions on $+$ stated in Section~\ref{sec:quasimodels},  reflect the standard counting principles for (countably) infinite sets.}
We use the following bounds from the proof of~\cite[Corollary 7.10]{pratt2023fragments}, with minor extensions to inequalities and~$\aleph_0$:
\begin{theorem}[\cite{pratt2023fragments}]\label{thm:sol}
Let $\mathcal{E}$ be a system of $n$ linear extended-Diophantine equations and inequalities
with variables $x_1,\dots,x_k$, and let $M$ be the maximum value of all the finite 
coefficients of $\mathcal{E}$. 
If $\mathcal{E}$ has a solution $\boldsymbol{m}$, then it has
a solution~$\avec{m}'$ such that either $\avec{m}'(x_i)\leq(2(k+3n+1)M-1)^{2n}$ or $\avec{m}'(x_i) =\avec{m}(x_i)=\aleph_0$, for every variable~$x_i$.
\end{theorem}
Now, instead of assuming the computability of a PAI formula $\textit{Pres}_{\varphi}$ in 
a blackbox manner, we exploit directly the \NExpTime upper bound proof for the two-variable fragment with counting 
in~\cite[Chapter 8]{pratt2023fragments}. Namely, we use the following consequence of the results presented in~\cite[Chapter 8]{pratt2023fragments}.
\begin{lemma}\label{lem:iancard}
	Let $\varphi$ be a $\CtwomonMLc$-sentence.
	There is a set $\Cmc$ of sets $\Emc$ of linear extended-Diophantine equations with variables $x_{1},\ldots,x_{k_{\varphi}}$, $|\Emc|$ at most exponential in $|\varphi|$, and coefficients of at most double exponential value in $|\varphi|$ such that a vector $\avec{n}\in \smash{\extN^{k_{\varphi}}}$ represents a quasistate for $\varphi$ iff $\avec{n}$ is a solution to some $\Emc\in \Cmc$. Moreover, $\Emc\in \Cmc$ can be decided in exponential time.
\end{lemma}
\begin{proof} We give a short sketch.
The proof of this lemma is based on~\cite[Sections~8.4 and 8.5]{pratt2023fragments} and  is rather cumbersome, but straightforward in principle. The following two observations should be sufficient. (1)
The normal form used in~\cite{pratt2023fragments} does not preserve satisfiability in small models. This can be treated by introducing new sets of equations in $\Cmc$ dealing with satisfiability in small models. (2) The equations in~\cite{pratt2023fragments} encode satisfiability in terms of the number of domain elements satisfying star-types and not the number of domain elements satisfying ``our'' types. This can be dealt with by expressing ``our'' types as disjunctions of star-types. The solutions to the equations 
in~\cite{pratt2023fragments}  are then used as coefficients in our equations. This leads to coefficients of double exponential size.
\end{proof}

\begin{theorem}\label{thm:c2:conexptime}
	For constant domains, $\Kn$-validity and $\Sfiven$-validity in $\CtwomonMLc$ are \textup{co\NExpTime}-complete. 
\end{theorem}
\begin{proof}
	We consider $\Kn$; the proof for $\Sfiven$ is similar. Hardness follows from the complexity of~$\textsf{C}^2$. So, we show the upper complexity bound. 
	
	Assume $\varphi$ is given. 
	We follow the proof schema above but define and use $\textit{Pres}_{\varphi}$ differently.
	By Lemmas~\revised{R1.18}{\ref{lem:k:tree-shaped}, \ref{lemma:quasimodel}, \ref{lem:weakprequasi:k}}, \ref{lemma:k:finite-approx}
	and~\ref{lem:weakpreweak},
	$\varphi$ is satisfiable in some $\Kn$ frame with constant domain iff $\varphi$ is satisfiable in a 1-element tree-shaped frame or
	there is a weak pre-quasimodel $\quasimod =
	(\Fmf, \funcand, \prset{}, \Lmf)$ satisfying~$\varphi$ based on a $\Tree{d(\varphi)}$ frame $\Fmf = (W, \{R_a\}_{a\in A})$ with~$|W| \leq 2^{O(d(\varphi)\cdot |\varphi|)}$ and  $|\runsp| \leq |W|\times 2^{|\varphi|}$. The algorithm checking the existence of such a weak pre-quasimodel $\quasimod$ now guesses $\Fmf$ and $\prset{}$ and, for each $\avec{w}\in W$, a set $\mathcal{E}_{\avec{w}}$ of linear extended-Diophantine equations with variables $x_{1}^{\avec{w}},\ldots,x_{k_{\varphi}}^{\avec{w}}$ satisfying the conditions of Lemma~\ref{lem:iancard} and finally guesses a solution of the set $\mathcal{E}$ of equations and inequalities comprising all~$\mathcal{E}_{\avec{w}}$, for $\avec{w}\in W$, and $\textit{pqm}_{\Fmf,\prset{}}$. By our definitions, a solution for $\mathcal{E}$ exists iff the weak pre-quasimodel $\quasimod$ introduced above exists.
	By Theorem~\ref{thm:sol}, a solution for~$\mathcal{E}$ exists iff there exists one in which the finite values are bounded by a double exponential function in $|\varphi|$ (and so can be represented using exponential space in $|\varphi|$). It can also be checked in exponential time whether the guessed values actually represent a solution, as required. 
\end{proof}


%% file: 4.3_expanding.tex

\section{Decidability for Monodic Fragments in Expanding Domains}

In this section, we study the complexity of reasoning in expanding domains. We begin with some general observations on the relationship between reasoning in constant and expanding domains. \revised{R1.3}{The standard relativisation technique (see, e.g.,~\cite[Proposition~3.20~(ii)]{GabEtAl03}) shows that global $\Cmc$-consequence in expanding domains is no more complex than in constant domains, even in the presence of equality, non-rigid constants and counting quantifiers.
\begin{theorem}\label{thm:exp:reduction:gc}
	Let $\Cmc$ be any class of frames.
	In $\QoneMLc$, $\CtwomonMLc$, $\QGuardmonMLc$ and $\QmonMLc$, global $\Cmc$-consequence for expanding domains is polytime-reducible to global $\Cmc$-consequence for constant domains. 
\end{theorem}
For $\Cmc$-validity, the situation is less clear. Whether an analogous reduction exists in general remains open. Indeed, the proof of the corresponding claim for equality-free languages in the above reference is incorrect, and the standard relativisation construction fails for $\Cmc$-validity in languages with equality over confluent intransitive frames. Nevertheless, the reduction holds for $\Kn$-validity.%
\begin{theorem}\label{thm:exp:validity:reduction}
	In $\QoneMLc$, $\CtwomonMLc$, $\QGuardmonMLc$ and $\QmonMLc$, $\Kn$-validity for expanding domains is polytime-reducible to $\Kn$-validity for constant domains.
\end{theorem}
\begin{proof}
	By Lemma~\ref{lem:k:tree-shaped}, every formula satisfiable in an interpretation based on a $\Kn$ frame with either expanding or constant domains is also satisfiable in such an interpretation based on a tree-shaped $\Kn$ frame. Moreover, the standard relativisation technique recalled above allows us to construct, for any formula $\varphi$, a formula $\varphi^{\textsf{Ex}}$ such that, for every tree-shaped $\Kn$ frame~$\Fmf$, $\varphi$ is satisfiable in an expanding-domain interpretation based on $\Fmf$ iff $\varphi^{\textsf{Ex}}$ is satisfiable in a constant-domain interpretation based on $\Fmf$.
\end{proof}}%
Theorems~\ref{thm:guarded:complexity} and~\ref{thm:c2:conexptime}, combined with the known hardness
results for~$\textsf{GF}$ and~$\textsf{C}^2$, yield matching complexity bounds for $\Kn$-validity in expanding domains:
\begin{corollary}\label{thm:expanding}
	For expanding domains, $\Kn$-validity is \textup{2\ExpTime}-complete in $\QGuardmonMLc$ and \textup{co\NExpTime}-complete in $\CtwomonMLc$.
\end{corollary}
We now turn to logics for which expanding domains make a difference compared to constant domains. For various logics with transitive closure, we show that reasoning is decidable but not primitive recursive over expanding domains, whereas it is undecidable over constant domains. For the one-variable fragment $\QoneMLc$, we show that $\Kn$-validity is \PSpace-complete over expanding domains, compared to co\NExpTime-complete over constant domains.

Both of our main constructions, quasimodels (Lemma~\ref{lemma:quasimodel}) and weak quasimodels (Lemmas~\ref{lem:k:tree-shaped}, \ref{lem:weakprequasi:k} and \ref{lemma:k:finite-approx}), also work for expanding domains.

The following additional assumption on (weak) quasimodels will simplify the constructions in this section: whenever the underlying frame $\Fmf = (W, \{R_a\}_{a\in A})$ is tree-shaped, we assume that all (weak) runs are rooted in the following sense.

For any $w\in W$, we set $W_{\downarrow\avec{w}}=\{\avec{w}\}\cup \{\avec{v} \in W\mid \avec{w}R_{\ast}\avec{v}\}$,
where $R_{\ast}$ is the transitive closure of $\bigcup_{a\in A}R_{a}$,
and denote the restriction of a (weak) run $\rho\in\runs_{\avec{w}}$ to $W_{\downarrow\avec{w}}$ by $\rho_{\downarrow\avec{w}}$.
A (weak) run  through $(\Fmf,\funcand)$ is called \emph{rooted} if there is $\avec{w}\in W$, called the \emph{root} of $\rho$, such that $W_{\downarrow\avec{w}}=\dom \rho$.  By the following lemma, we can always assume that (weak) runs are rooted.

\begin{lemma}\label{lem:rooted:runs}
For any quasimodel $\Qmf = (\Fmf, \funcand, \runs)$ based on a tree-shaped frame $\Fmf$,  there is a quasimodel $\Qmf' = (\Fmf, \funcand, \runs')$ with rooted runs $\runs'$ that satisfies the same formulas and such that $|\runs'| \leq |\runs| \cdot |W|$.

Similarly, for any weak quasimodel $\Qmf = (\Fmf, \funcand, \runs, \prset)$, there is a weak quasimodel $\Qmf' = (\Fmf, \funcand, \runs', \prset')$ with rooted weak runs $\runs'$ and $\prset'$ that satisfies the same formulas and such that $|\runs'| \leq |\runs| \cdot |W|$.
\end{lemma}
\begin{proof}
Let $\rho\in \runs$ have no root. Let $V \subseteq W$ be a (possibly infinite) minimal set such that $\dom\rho=\biguplus_{v\in V}W_{\downarrow\avec{v}}$; note that since the frame is tree-shaped, the union is disjoint. Then we replace, in $\runs$, the (weak) run $\rho$ by its restrictions~$\rho_{\downarrow\avec{v}}$ to $W_{\downarrow\avec{v}}$, for $v\in V$. Formally,  $\runs'$ is defined as~$\runs$ except that 
\begin{equation*}	
\runs'(\rho)=0 \quad\text{ and }\quad \runs'(\rho_{\downarrow\avec{v}})=\runs(\rho_{\downarrow\avec{v}})+\runs(\rho),\ \text{ for } v\in V;
\end{equation*}
note that $\rho_{\downarrow\avec{v}}$ could already exist in $\runs$ as a separate run---this is taken into account by the first component in the sum above, \revised{R2.21}{which preserves the number of runs through each type in every $\avec{w}\in W_{\downarrow\avec{v}}$.} 

For a weak quasimodel, if $\prfun{\avec{w},\contp}=\rho$, for some $\avec{w}\in W$ and~$\contp\in\funcand(\avec{w})$, then $\prfun{\avec{w},\contp}$ is replaced with the~$\rho_{\downarrow\avec{v}}$ such that $\avec{w}\in W_{\downarrow\avec{v}}$, for $v\in V$ (which is uniquely determined as the frame is tree-shaped).
\end{proof}

\subsection{Logics with Transitive Closure}

We consider modal logics with sets of modalities of the form \mbox{$A=A_{0}\cup \{\ast\}$}, where $\ast$ denotes the modality interpreted by the transitive closure of the union of the remaining modalities in $A_{0}$. Formally, a $\K_{\ast n}$ frame is of the form $\Fmf=(W,\{R_{a}\}_{a\in A})$, where $|A_0| = n$ and
$R_{\ast}$ is the transitive closure of $\bigcup_{a\in A_{0}}R_{a}$. If, in addition, $R_{\ast}$ contains no infinite ascending chains, that is,
no infinite sequence $w_{0},w_{1},\ldots \in W$ with $w_{i}R_{\ast}w_{i+1}$ for all $i\geq 0$, then $\Fmf$ is called a $\Kfn$ frame.  Observe that in contrast to $\Kn$ and $\Sfiven$, with $n \geq 2$,
global $\K_{\ast n}$-consequence and global $\Kfn$-consequence are polytime-reducible to
$\K_{\ast n}$- and, respectively, $\Kfn$-validity since $\varphi$ is a global $\mathcal{C}$-consequence of $\Gamma$ iff $(\bigwedge\Gamma \wedge \Box_{\ast}\bigwedge \Gamma) \rightarrow \varphi$ is $\mathcal{C}$-valid for $\mathcal{C}\in \{\K_{\ast n},\Kfn\}$; see also Remark~\ref{rem:s5:global}.
\begin{lemma}
	\label{lem:redglobloc}
	For both constant and expanding domains, for $\Cmc\in \{\K_{\ast n},\Kfn\}$ and all
	fragments $\Lmc$ considered, 
	global $\Cmc$-consequence in $\mathcal{L}$ is polytime-reducible to $\Cmc$-validity in $\mathcal{L}$.
\end{lemma} 

A $\Kfn$ frame is called \emph{tree-shaped} if the restriction $\Fmf_{|A_{0}}=(W,\{R_{a}\}_{a\in A_{0}})$ of $\Fmf$ to $A_{0}$ 
is a tree-shaped $\K_{n}$ frame. We use the notation from Section~\ref{sec:kn:weak:quasimodels}
for these frames, in particular worlds in tree-shaped frames take the form~\eqref{eq:world:k} with $a_{j}\in A_{0}$ for $0\leq j <m$. 
\begin{lemma}\label{lem:k:tree-shaped-exp}
For both constant and expanding domains, 
every $\Kfn$-satisfiable $\QMLplusc$-formula
is also satisfiable in tree-shaped $\Kfn$ frames.
\end{lemma}
\begin{proof}
	The construction is exactly the same as in the proof of Lemma~\ref{lem:k:tree-shaped} except that we use the modalities in $A_{0}$ only and do not impose the bound $d(\varphi)$ on the depth $d(\avec{w})$  of worlds $\avec{w}$.
\end{proof}

When working with (weak) quasimodels based on $\Kfn$ frames, it is useful to extend the set of subformulas $\sub[x]{\varphi}$ that occur in types. Similar extensions have been introduced for many modal logics with operators for the transitive closure, for instance the Fischer-Ladner 
closure for PDL~\cite{DBLP:conf/stoc/FischerL77}. In detail, define $\subs[x]{\varphi}$ by extending $\sub[x]{\varphi}$ with all formulas of the form $\Diamond_{a}\psi$ and $\Diamond_{a}\Diamond_{\ast}\psi$, for $\Diamond_{\ast}\psi\in \sub[x]{\varphi}$  and $a\in A_{0}$.
In what follows we always work with $\subs[x]{\varphi}$ instead of $\sub[x]{\varphi}$. We may assume that $|\subs[x]{\varphi}|\leq s(\varphi)$, where $s$ is defined by taking $s(\varphi)=2|\varphi|^{2}$.
 
A \emph{$\Kfn$ type for $\varphi$} is a Boolean-saturated subset $\contp$  of $\subs[x]{\varphi}$ satisfying
\begin{enumerate}
	[label=\textbf{(kf$_*$-type)},leftmargin=*,series=run]
	\item $\Diamond_{\ast}\psi\in \contp$ iff there is $a\in A_{0}$ with either $\Diamond_{a}\psi\in \contp$ or
$\Diamond_{a}\Diamond_{\ast}\psi\in \contp$, for all $\Diamond_{\ast}\psi\in \subs[x]{\varphi}$.\label{kfs-type}
\end{enumerate}
Let $\Mmf = (\Fmf, \Delta, \cdot)$ be an interpretation with a $\Kfn$ frame $\Fmf$. Then it is straightforward
to show that $\contp^{\Mmf(w)}(d) = \{ \psi \in \subs[x]{\varphi} \mid \Mmf, w\models \psi[d] \}$ is a $\Kfn$-type for $\varphi$, for every $w \in W$ and $d \in \Delta_w$. 
A \emph{$\Kfn$ basic structure for~$\varphi$} is a basic structure for $\varphi$ based on  a $\Kfn$ frame containing $\Kfn$ types for $\varphi$.
A run and  quasimodel are now defined as before and called a $\Kfn$ run and $\Kfn$ quasimodel, respectively. Then Lemma~\ref{lemma:quasimodel} still holds for $\Kfn$ quasimodels.

We could  define $\Kfn$  weak quasimodels similarly. However, to simplify the subsequent constructions, it will be useful to provide an equivalent description of (weak) $\Kfn$ runs 
in which the modality $\Diamond_{\ast}$ is considered only implicitly. 
\revised{R2.22}{\begin{lemma}
Let $\rho$
be a function mapping each world~$\avec{w}$ in an upward-closed subset $W'$ of $W$ to a $\Kfn$ type $\rho(\avec{w})$. Then
\begin{enumerate}  
\item $\rho$ satisfies~\ref{a-r-coh} for $a = *$ if it satisfies~\ref{a-r-coh}  for all $a\in A_0$\textup{;}
\item $\rho$ satisfies~\ref{run:wsat} for $a = *$ and all $\avec{w}\in \dom\rho$ if it satisfies~\ref{run:wsat}
for all $a\in A_{0}$ and all  $\avec{w}\in \dom\rho$.
\end{enumerate}
\end{lemma}}
\begin{proof}
Take any $\avec{w}\in\dom\rho$.

\smallskip

(1)  Let $\psi \in \rho(\avec{v})$ for some $\avec{v}\in W$ with $\avec{w}R_*\avec{v}$. By definition, there is a finite path $\avec{w} = \avec{w}_{0}R_{a_{1}}\cdots R_{a_{m}}\avec{w}_m =\avec{v}$ with $m > 0$ and $a_{1},\ldots,a_{m}\in A_{0}$. By~\ref{a-r-coh} for $a = a_{m}$, we get $\Diamond_{a_m} \psi\in \rho(\avec{w}_{m-1})$ and so, by~\ref{kfs-type}, $\Diamond_* \psi\in \rho(\avec{w}_{m-1})$. By induction on $i$, we show  $\Diamond_{\ast}\psi\in \rho(\avec{w}_{i})$ for all $0\leq i<m$. The basis, $i = m - 1$, is done. For the inductive step, let $\Diamond_*\psi\in\rho(\avec{w}_i)$. Since $\Diamond_{a_i}\Diamond_*\psi \in \subs[x]{\varphi}$, by~\ref{a-r-coh} for $a = a_i$, we have $\Diamond_{a_i}\Diamond_*\psi\in\rho(\avec{w}_{i-1})$, whence, by~\ref{kfs-type},   $\Diamond_*\psi\in\rho(\avec{w}_{i-1})$. Thus, $\Diamond_*\psi\in \rho(\avec{w})$, as required.

\smallskip

(2) Let $\Diamond_{\ast}\psi\in \rho(\avec{w})$. By~\ref{kfs-type}, $\Diamond_{a}\psi\in \rho(\avec{w})$ or $\Diamond_{a}\Diamond_{\ast}\psi\in \rho(\avec{w})$, for some $a \in A_{0}$. In the former case, by~\ref{run:wsat}, there is $\avec{v}\in W$ with $\avec{w}R_{a}\avec{v}$ and $\psi\in \rho(\avec{v})$, as required.  In the latter case, we can argue inductively and construct a finite path $\avec{w} = \avec{w}_{0}R_{a_1}\cdots R_{a_{m}}\avec{w}_m=\avec{v}$ with $a_{1},\ldots,a_{m}\in A_{0}$ such that $\psi\in \rho(\avec{v})$. This is possible since the frame has no infinite paths.
\end{proof}

We can thus define $\Kfn$  weak quasimodels by restricting conditions~\ref{a-r-coh} and~\ref{approx:R2} to $A_0$ only.
\begin{lemma}\label{lem:equivweaknotweak}
	Let  $\varphi$ be a $\QmonMLplusc$-sentence. 
	
	\textup{(a)} For every $\Kfn$ quasimodel satisfying~$\varphi$ based on a tree-shaped frame, there exists a $\Kfn$ weak quasimodel  satisfying~$\varphi$ based on a \textup{(}finite\textup{)} frame of outdegree bounded by $s(\varphi)2^{s(\varphi)}$. 
	
\textup{(b)} Conversely, for every $\Kfn$  weak quasimodel satisfying $\varphi$ based on a finite frame, there exists a $\Kfn$ quasimodel satisfying $\varphi$.
\end{lemma}
\begin{proof}
By straightforward adaptation of the proofs of Lemmas~\ref{lem:weakprequasi:k} and~\ref{lemma:k:finite-approx}.
\end{proof}

By Lemma~\ref{lem:equivweaknotweak}~(a), to show decidability of a fragment of $\QmonMLplusc$ over $\Kfn$ frames, it suffices to show that satisfiable sentences 
$\varphi$ are satisfied in $\Kfn$  weak quasimodels based on Kripke frames and quasistates of recursive size in $|\varphi|$. To this end, we aim to show that it suffices to consider weak quasimodels with a recursive bound (depending on $\varphi$) on the length of sequences $\avec{n}_{1}R_{a_{1}}\avec{n}_{2}R_{a_{2}}\ldots$ of quasistates. 

To obtain such a bound, we apply Dickson's Lemma to the quasistates with the product ordering. Call a pair $\avec{n},\avec{n}'\in\extN^k$ with $\avec{n} \leq \avec{n}'$ an \emph{increasing pair}. A (finite or infinite) sequence $\avec{n}_{1},\avec{n}_{2},\ldots\in \extN^{k}$ is \emph{bad} if it contains no increasing pair $\avec{n}_i,\avec{n}_j$ with $i<j$. By Dickson's Lemma, no infinite sequence $\avec{n}_{1},\avec{n}_{2},\ldots\in \extN^{k}$ is bad. It is easy to see that there are arbitrarily long finite bad sequences, however. To bound also the length of bad sequences we introduce \emph{control functions}~\cite{DBLP:conf/icalp/SchmitzS11} providing a bound on the size of the $\avec{n}_{i}$. 
Define a \emph{norm} $|\avec{n}|_{f}\in \mathbb{N}$ on $\extN^{k}$ by setting $|\avec{n}|_{f}=\sum_{n_{i}<\aleph_{0}}n_{i}$, for $\avec{n}\in \extN^{k}$.
This norm is \emph{proper} since $\{ \avec{n}\in\extN^{k} \mid |\avec{n}|_{f}<N\}$ is finite for all $N\in \mathbb{N}$. The following can be proved along the lines of 
proofs of Dickson's~Lemma~\cite{DBLP:conf/lics/FigueiraFSS11,DBLP:conf/icalp/SchmitzS11}.
\begin{theorem}\label{thm:bound}
	Let $g\colon\mathbb{N} \rightarrow \mathbb{N}$ be a recursive function, and let 
	$\length(m)$ be the maximal length of a bad sequence $\avec{n}_{1},\avec{n}_{2},\ldots\in \extN^{k}$ with $|\avec{n}_{i}|_{f} \leq g^{i}(m)$ for all $i$, where $g^i$ denotes the $i$ iterated applications of $g$. Then $\length(m)$ is bounded by a recursive function that can be obtained from~$g$.
\end{theorem}
Theorem~\ref{thm:bound} can be refined in various \revised{R1.20}{ways} by providing bounds on $\length(m)$
depending on $g$ if $g$ is primitive recursive. We note, however, that $\length(m)$  grows very fast. In our application of Theorem~\ref{thm:bound}, $m$ is given by the size $|\varphi|$ of the input formula $\varphi$.

\begin{lemma}\label{lem:squeeze}
There is a primitive recursive function $g$ such that, for each $\QGuardmonMLplusc$-sentence $\varphi$,  each quasistate $\avec{m}$ and each \emph{finite} quasistate candidate $\avec{n}\leq \avec{m}$  for $\varphi$, there exists a \emph{finite} quasistate $\avec{m}'$ for $\varphi$ such that $\avec{n} \leq \avec{m}' \leq \avec{m}$ and $|\avec{m}'|\leq g(|\avec{n}| + |\varphi|)$. 
	
	The statement above also holds for $\CtwomonMLplusc$-sentences  without the condition that $\avec{n}$ and $\avec{m}'$ are finite but with $|\cdot|$ replaced by $|\cdot|_{f}$. 
\end{lemma}
\begin{proof}
	Assume $\avec{m}$ and $\avec{n}$ are given and assume first that $\varphi$ is a $\QGuardmonMLplusc$-sentence. 
	It is known~\cite[Theorem~1.2]{DBLP:journals/corr/BaranyGO13} that there is a double exponential function $f$ such that  every satisfiable $\textsf{GF}$-sentence $\psi$ built from the~$\overline{\contp}$, for types $\contp$ for $\varphi$, can be satisfied in a first-order structure where every type for $\varphi$ is realised at most~$f(|\psi|)$ times  (strictly speaking, each type for $\psi$ determines a type for~$\varphi$, but the factor given by the number of types for $\psi$ is absorbed into $f$). For any number $N \geq  |\avec{n}| + |\varphi|$, we define $\psi_{N}$ as follows:	
\begin{equation*}	
	 \psi_{N} \ \ = \ \ \textsf{only}_{\avec{m}} \ \ \land \ \ \bigwedge_{\contp\in\avec{n}} \ \textsf{atleast}^{\avec{n}(\contp)}_{\contp} \ \ \land \bigwedge_{\begin{subarray}{c}\contp\in\avec{m}\\\avec{m}(\contp)\leq N\end{subarray}} \textsf{atmost}^{\avec{m}(\contp)}_{\contp};
\end{equation*}	 
see~\eqref{eq:gf:only}--\eqref{eq:gf:atmost} in the proof of Lemma~\ref{lem:guardedfragment}.
It should be clear that every first-order structure satisfying $\psi_N$  realises only the types $\contp\in\avec{m}$, and each type $\contp\in\avec{m}$ is realised at least~$\avec{n}(\contp)$ times (by assumption, $\avec{n}(t) \leq N$) and at most $\avec{m}(t)$ times provided that $\avec{m}(t)\leq N$ (in particular, there is no upper bound when $\avec{m}(\contp) = \aleph_0$).
    Now let $N_{0} \geq  |\avec{n}| + |\varphi|$ be minimal such that,  for all types $\contp$ for $\varphi$,
\begin{equation}\label{eq:gf:pr:N0}    
\text{ either } \ \ \avec{m}(\contp) \leq N_{0} \ \ \text{ or } \ \ \avec{m}(\contp) > f(|\psi_{N_{0}}|).
\end{equation}
Such an $N_0$ partitions the types $\contp$ for $\varphi$ into those with the multiplicity we need to take account of in~$\textsf{atmost}_{\avec{m}}^{\contp}$
and those with the required multiplicity exceeding the size of `small' finite first-order structures satisfying $\psi_{N_0}$.  We can obtain $N_0$ by iterating: we start with $N_0 = |\avec{n}| + |\varphi|$ and check whether it satisfies~\eqref{eq:gf:pr:N0}. If it does, then we are done. Otherwise, we take the new $N_0$ to be $f(|\psi_{N_0}|)$ and repeat the check. Since at each step at least one more type will satisfy $\avec{m}(\contp) \leq N_{0}$, the iteration stops after at most $2^{s(\varphi)}$ times. It follows that $N_0$ is bounded by a primitive recursive function in $|\avec{n}| + |\varphi|$ (but  in general not elementary as the height of the exponent depends on the number of types).
Finally, as $\avec{m}$ is a quasistate,  the sentence $\psi_{N_{0}}$ is satisfiable, and the quasistate~$\avec{m}'$ realised by a `small' first-order structure $\Bmf'$ satisfying $\psi_{N_{0}}$ is as required: in particular, $\avec{m}' \leq \avec{m}$ as the size of~$\Bmf'$ is  bounded by $f(|\psi_{N_{0}}|)$.

\smallskip
    
    Now assume that $\varphi$ is a $\CtwomonMLc$-sentence. 
By Lemma~\ref{lem:iancard} and Theorem~\ref{thm:sol}, there is a double exponential function $f$ such that, if a $\textsf{C}^{2}$-sentence $\psi$ built from the~$\overline{\contp}$, for types $\contp$ for $\varphi$, is satisfied in a first-order structure $\Bmf$, then it is satisfied in a first-order structure where every type~$\contp$ for~$\varphi$ is realised at most~$f(|\psi|)$ times if it is realised by finitely many elements in $\Bmf$ (again, each type for $\psi$ determines a type for $\varphi$, but the factor given by the number of types for $\psi$ is absorbed into~$f$).
    For any number $N \geq  |\avec{n}|_{f} + |\varphi|$, define $\psi'_N$ as follows:
\begin{equation*}	
	 \psi'_{N} \ \ = \ \ \textsf{only}_{\avec{m}} \ \ \land \bigwedge_{\begin{subarray}{c}\contp\in\avec{n}\\ \avec{n}(\contp)< \aleph_0\end{subarray}} \textsf{atleast}^{\avec{n}(\contp)}_{\contp} \ \ \land  \bigwedge_{\begin{subarray}{c}\contp\in\avec{n}\\ \avec{n}(\contp)=\aleph_0\end{subarray}}  \textsf{infinite}_{\contp} \ \ \land \bigwedge_{\begin{subarray}{c}\contp\in\avec{m}\\\avec{m}(\contp)\leq N\end{subarray}} \textsf{atmost}^{\avec{m}(\contp)}_{\contp},
\end{equation*}	
where $\textsf{infinite}_{\contp}$ is a $\textsf{C}^{2}$-sentence expressing (using an additional binary relation, see Example~\ref{ex:running-example:1}) that $\contp$ is realised infinitely many times.    
 (Note that $\textsf{atleast}_{\avec{n}}^{\contp}$ and $\textsf{atmost}_{\avec{m}}^{\contp}$ can now also be expressed using counting quantifiers instead of constants.)
   Now let $N_{0} \geq  |\avec{n}|_{f} + |\varphi|$ be minimal such that the~$\psi_{N_{0}}'$  counterpart of~\eqref{eq:gf:pr:N0} holds for all types $\contp$  for $\varphi$
(the iterative construction of $N_0$ is the same as in the  $\QGuardmonMLplusc$ case).
 Then $\psi_{N_{0}}'$ is satisfied by the first-order structure $\Bmf$ realising~$\avec{m}$. Thus, there is
    a first-order structure $\Bmf'$ satisfying~$\psi_{N_{0}}'$ such that each type~$\contp$ for~$\varphi$ with $\avec{m}(\contp)< \aleph_0$ is realised at most~$f(|\psi_{N_{0}}'|)$ times in $\Bmf'$. The quasistate~$\avec{m}'$ realised by $\Bmf'$ is as required, and $|\avec{m}'|_f$ is bounded by a primitive recursive function in $|\avec{n}|_{f} + |\varphi|$.
\end{proof}

\begin{lemma}\label{lem:recboundnew}
For expanding domains,
	there is a primitive recursive function $g$ such that, for each  $\QGuardmonMLplusc$-sentence $\varphi$,  there exists a $\Kfn$ quasimodel satisfying~$\varphi$ based on a tree-shaped frame iff there is 
	a weak quasimodel $\Qmf = (\Fmf,\funcand,\runs,\prset{})$ satisfying~$\varphi$ with a \emph{(}finite\emph{)} tree-shaped $\Kfn$ frame $\Fmf$  such that, for all worlds $\avec{w}\in W$,
	\begin{enumerate}
		\item $|\funcand(\avec{w})|\leq g^{d(\avec{w})+1}(|\varphi|)$\textup{;}
		\item $\funcand(\avec{w})\not\leq \funcand(\avec{v})$, for all $\avec{v}\in W$ with $\avec{w}R_*\avec{v}$\textup{;}
		\item the outdegree of $\avec{w}$  is bounded by $ s(\varphi)2^{s(\varphi)}$.
	\end{enumerate} 
The statement above holds for $\CtwomonMLplusc$-sentences with $|\cdot|$ replaced by $|\cdot|_{f}$.
\end{lemma}  
\begin{proof}
	We consider  $\QGuardmonMLplusc$ first. Let $g\colon \mathbb{N}\to\mathbb{N}$ be defined by taking $g(n) = g_0(2^{5n})$, for $n\in\mathbb{N}$, where $g_0$ is given by Lemma~\ref{lem:squeeze}; we assume that $g_0$ is monotone and  $g_0(n) \geq n$, for all $n\in\mathbb{N}$. Let $\varphi$ be a $\QGuardmonMLplusc$-sentence. The $(\Leftarrow)$ direction is by Lemma~\ref{lem:equivweaknotweak}~(b).
	
\smallskip	
	
$(\Rightarrow)$	
We start from a weak quasimodel $\Qmf_0$ satisfying~$\varphi$ and apply the two rules described below exhaustively to construct a finite sequence $\Qmf_0,\dots,\Qmf_N$ of weak quasimodels  satisfying~$\varphi$ such that  Items~\mbox{1--3} apply to $\Qmf_N$.  
More precisely, by Lemma~\ref{lem:equivweaknotweak}~(a), we take a weak quasimodel $\Qmf_0$ satisfying~$\varphi$ such that the outdegree of every world  is bounded by $ s(\varphi)2^{s(\varphi)}$, thus satisfying Item~3. We then exhaustively apply the \textit{Drop Interval} rule. The result clearly satisfies Items~2 and~3 (as it preserves the outdegree). Next,  to ensure that the bound in Item~1 is satisfied, we exhaustively apply the \textit{Small Quasistate} rule.
Since it preserves the outdegree, the result satisfies Item~3 but may violate Item~2. So we repeat the process: \textit{Drop Interval}, followed by \textit{Small Quasistate}, and so on. Since the frame is finite, and \textit{Drop Interval}  decreases the size of the frame, while \textit{Small Quasistate} preserves it, the process eventually terminates with a weak quasimodel satisfying Items~1--3.

To avoid notational clutter,  we assume that each rule below transforms a given $\Qmf = (\Fmf,\funcand,\runs,\prset{})$ with a tree-shaped $\Fmf = (W,\{R_a\}_{a\in A})$ rooted in $\avec{w}_0$  into  some $\Qmf'$. By Lemma~\ref{lem:rooted:runs}, we assume that the runs in $\runs$ are rooted.
	
We make use of the following additional notation. For any $\avec{w}\in W$, we denote $W^{\uparrow \avec{w}} = W\setminus W_{\downarrow\avec{w}}$, and, for a (weak) run $\rho$, we denote by $\rho^{\uparrow \avec{w}}$ the restriction of~$\rho$ to $W^{\uparrow \avec{w}}$. Note that, while the domain of $\rho^{\uparrow \avec{w}}$ is not necessarily upward-closed and so it is not a (weak) run in general, this map is useful for the splicing operation $\oplus$: given (weak) runs $\rho,\rho'$ and worlds $\avec{u}\in W$ and $\avec{v}\in W_{\downarrow\avec{u}}$ such that $\rho(\avec{u}) = \rho'(\avec{v})$, we denote by $\rho^{\uparrow\avec{u}} \oplus \rho'_{\downarrow\avec{v}}$ the map defined as follows:
\begin{equation*}
\rho^{\uparrow\avec{u}} \oplus \rho'_{\downarrow\avec{v}}\colon \avec{w}\mapsto 
\begin{cases} 
\rho(\avec{w}), & \text{if } \avec{w} \in W^{\uparrow \avec{u}} \cap \dom\rho,\\
\rho'(\avec{w}), & \text{if } \avec{w} \in W_{\downarrow \avec{v}}.
\end{cases}
\end{equation*}%
Observe that $\rho^{\uparrow\avec{u}} \oplus \rho'_{\downarrow\avec{v}}$ is a (weak) run if $\avec{u} = \avec{v}$. In general, it is a weak run on the restriction $(\Fmf', \funcand')$ of $(\Fmf, \funcand)$ to $W^{\uparrow\avec{u}} \cup W_{\downarrow\avec{v}}$ with $\avec{u}'R_a\avec{u}$ replaced by $\avec{u}'R'_a\avec{v}$.

	\medskip
	\noindent
	\emph{Small Quasistate}: Consider $\avec{w}\in W$ violating Item~1. If $\avec{w} = \avec{w}_0$, then we define a quasistate candidate $\avec{n}$ for $\varphi$ by setting
\begin{equation*}	
\avec{n}(t) = 0, \text{ for all types } \contp \text{ for } \varphi.
\end{equation*}
By Lemma~\ref{lem:squeeze}, there exists a quasistate $\avec{y}$
	for $\varphi$ with 
\begin{equation}\label{eq:exp:quasistate}
\avec{n} \leq \avec{y} \leq \funcand(\avec{w}) \quad\text{ and }\quad |\avec{y}|\leq g_0(|\avec{n}| + |\varphi|).
\end{equation}
We define $\Qmf'=(\Fmf,\funcand',\runs',\prset{}')$ as follows. We replace the quasistate at $\avec{w}$  with~$\avec{y}$: we set $\funcand'(\avec{w}) = \avec{y}$ and $\funcand'(\avec{u})=\funcand(\avec{u})$ for all $\avec{u}\ne\avec{w}$. Next, we remove some weak runs through~$\avec{w}$ to match the cardinality.\revised{R2.23}{}  To this end, we construct a multiset~$\Smf_{\downarrow\avec{w}}$. First, for each $\rho\in \runs_{\avec{w}}$, we find $\Smf_{\downarrow\avec{w}}(\rho)\leq \runs(\rho)$ such that 
	 \begin{enumerate}[label=\textbf{\bfseries (sqr$_{\theenumi}$)},leftmargin=*,series=run]
	 	\item\label{small-quasistate:down-pr} $\Smf_{\downarrow\avec{w}}(\rho)>0$ if $\rho\in \prfun{\avec{w}}$ and $\rho(\avec{w})\in\funcand'(\avec{w})$,
	 	\item\label{small-quasistate:down} $\sum\bigl\{\Smf_{\downarrow\avec{w}}(\rho')\mid \rho'\in\runs_{\avec{w},\contp} \bigr\}= \funcand'(\avec{w},\contp)$, for all types $\contp$ for $\varphi$.
	 \end{enumerate}
Next, to match the cardinality in the rest of $W_{\downarrow\avec{w}}$, we compensate the removal of weak runs $\rho\in\runs_{\avec{w}}$ by adding back to~$\Smf_{\downarrow\avec{w}}$ their restrictions to the $W_{\downarrow\avec{u}}$, for $a$-successors $\avec{u}$ of $\avec{w}$: for every weak run $\rho'$ (not necessarily from $\runs$) such that $\dom\rho' \subseteq W_{\downarrow\avec{u}}$, for~$\avec{u}\in W$ with~$\avec{w}R_a\avec{u}$, $a\in A_0$, we set 
\begin{equation*}
\Smf_{\downarrow\avec{w}}(\rho') \ \ = \ \ \runs(\rho') \ \ + \hspace*{-1.5em}\sum_{\rho\in \runs_{\avec{w}}, \ \rho'= \rho_{\downarrow\avec{u}}} \hspace*{-1.5em} \bigl(\runs(\rho) - \Smf_{\downarrow\avec{w}}(\rho)\bigr)
\end{equation*}
($\avec{u}$ is unique as the weak runs are rooted);
note that the bracket is non-negative. 
If $\avec{w}=\avec{w}_0$, then set $\runs' = \Smf_{\downarrow\avec{w}}$. Thus,~\ref{run:exists} holds for $\runs'$ and $\funcand'$ at all worlds.

We now extend  the construction to cover non-root quasistates $\funcand(\avec{w})$. Consider $\avec{w}'\in W$ with $\avec{w}' R_a \avec{w}$ and assume $|\funcand(\avec{w}')|\leq g^{d(\avec{w}')+1}(|\varphi|)$. 
We aim to retain the multiplicity of all weak runs through $\avec{w}'$ 
when we make $\funcand'(\avec{w})$ small. We define a quasistate candidate $\avec{n}$ for $\varphi$ by taking
\begin{equation*}
\avec{n}(\contp) = \sum\bigl\{\runs(\rho) \mid \rho \in \runs_{\avec{w}'}\cap\runs_{\avec{w},\contp}\bigr\},  \text{ for all types } \contp \text{ for } \varphi.
\end{equation*}
Hence, we have the following:
 \begin{equation*}
 |\avec{n}| 
\ \ \leq \ \ 
|\funcand(\avec{w}')| \ \ \leq \ \ g^{d(\avec{w}')+1}(|\varphi|) \ \ = \ \ g^{d(\avec{w})}(|\varphi|).
 \end{equation*}
By Lemma~\ref{lem:squeeze}, there exists a quasistate $\avec{y}$ for $\varphi$ satisfying~\eqref{eq:exp:quasistate}. We again set $\funcand'(\avec{w}) = \avec{y}$,  $\funcand'(\avec{u})=\funcand(\avec{u})$, for all $\avec{u}\ne\avec{w}$, and use the weak runs in $\Smf_{\downarrow\avec{w}}$ for~$W_{\downarrow\avec{w}}$. However, we also need to match them to the weak runs through~$\avec{w}'$.\revised{R2.23}{} So, we choose a multiset $\Smf$: for each $\rho\in \runs_{\avec{w}}$, we find $\Smf(\rho)\leq \runs(\rho)$ such that
	 \begin{enumerate}[label=\textbf{\bfseries (sqr$_{\theenumi}$)},leftmargin=*,resume]
	 	\item\label{small-quasistate:up:pr} $\Smf(\rho)=\runs(\rho)$, if $\rho\in \runs_{\avec{w}'}$,
	 	\item\label{small-quasistate:up} $\sum\bigl\{\Smf(\rho')\mid \rho'\in\runs_{\avec{w},\contp} \bigr\}= \funcand'(\avec{w},\contp)$, for all types $\contp$ for $\varphi$.
	 \end{enumerate}
By~\ref{small-quasistate:down} and~\ref{small-quasistate:up}, there is a bijection $f_\contp$, for each type $\contp$ for $\varphi$, between the indexed sets $\hat{\Smf}_{\downarrow\avec{w}}$ and $\hat{\Smf}$  of weak runs that map $\avec{w}$ to $\contp$. So, we can define $\runs'$ by taking the multiset of 
\begin{itemize}
\item $\rho\in\runs$ (with their multiplicity) such that $\dom\rho\subseteq W^{\uparrow\avec{w}}$, 
\item $\rho\in \Smf_{\downarrow\avec{w}}$ (with their multiplicity)  such that $\dom\rho\subseteq W_{\downarrow\avec{w}} \setminus\{\avec{w}\}$ and 
\item 
$f_{\rho(\avec{w})}(\rho,\ell)^{\uparrow\avec{w}} \oplus \rho_{\downarrow\avec{w}},$  for indexed runs $(\rho,\ell)\in\hat{\Smf}_{\downarrow\avec{w}}$ with $\avec{w}\in\dom\rho$.
\end{itemize}
Thus,~\ref{run:exists} holds for $\runs'$ and $\funcand'$ at all worlds. Observe that the construction of $\runs'$ generalises the case for root $\avec{w}_0$, where we could assume that $\Smf$ consists of the weak runs $\rho$ in $\Smf_{\downarrow\avec{w}}$ with $\avec{w}\in\dom\rho$ so that the~$f_\contp$ are all identities.

Prototype function $\prset{}$ is updated as follows, for $\avec{u}\in W$ and
$t\in\funcand'(\avec{u})$:
\begin{itemize}
\item if $\avec{u}=\avec{w}$, we take
$\prfun[']{\avec{w},\contp}=f_\contp(\rho,0)^{\uparrow\avec{w}}\oplus\rho_{\downarrow\avec{w}}$, for $\rho = \prfun{\avec{w},t}\in\Smf_{\downarrow\avec{w}}$; see ~\ref{small-quasistate:down-pr};
\item if $\avec{u}\in W^{\uparrow\avec{w}}$ and $\prfun{\avec{u},t}\in\Rmf_{\avec{w}'}$, we set 
$\prset{}'_{\avec{u},t}=\prfun{\avec{u},t}^{\uparrow\avec{w}}\oplus\rho_{\downarrow\avec{w}}$, for some $\rho$,
which exists by~\ref{small-quasistate:up:pr}; it agrees with $\prfun{\avec{u},t}$ on $W^{\uparrow\avec{w}}$ and has the
same type at $\avec{w}$;
\item if $\avec{u}\in W_{\downarrow\avec{w}}\setminus\{\avec{w}\}$ and $\Rmf'(\prfun{\avec{u},t})=0$, we take
as $\prfun[']{\avec{u},t}$ a weak run of $\Rmf'$ that coincides with $\prfun{\avec{u},t}$ on
$W_{\downarrow\avec{w}}$ if $\Smf_{\downarrow\avec{w}}(\prfun{\avec{u},t})>0$, and the restriction of
$\prfun{\avec{u},t}$ to the $W_{\downarrow\avec{v}}$ containing $\avec{u}$, for $\avec{w}R_a\avec{v}$,
otherwise;
\item all other prototypes are kept.
\end{itemize}

It can be seen that $\Qmf'$ is a weak quasimodel for $\varphi$ satisfying, by~\eqref{eq:exp:quasistate}, 
\begin{equation*}	
	|\funcand'(\avec{w})| \ \ \leq \ \ g_0(g^{d(\avec{w})}(|\varphi|) + |\varphi|) \ \ \leq \ \ 
        g_0(2^{5g^{d(\avec{w})}(|\varphi|)}) \ \ = \ \
	g^{d(\avec{w})+1}(|\varphi|).
\end{equation*}
Thus, Item~1 holds for~$\avec{w}$. Since the frame is the same, Item~3 is also satisfied.
	
	\medskip
	\noindent
	\emph{Drop Interval}: Consider $\avec{w},\avec{v}\in W$ with $\avec{w}R_{\ast}\avec{v}$ and $\funcand(\avec{w})\leq \funcand(\avec{v})$.
	If $\avec{w} = \avec{w}_{0}$,  then we construct $\Qmf'$
	by restricting $\Qmf$ to $W_{\downarrow\avec{v}}$ in the obvious way.
	
	If there is $\avec{w}'\in W$ with $\avec{w}'R_a\avec{w}$, then 
	we construct $\Qmf'= (\Fmf',\funcand',\runs',\prset{}')$ as follows.
	The frame~$\Fmf'$ is obtained from $\Fmf$ by replacing $W_{\downarrow\avec{w}}$ by $W_{\downarrow\avec{v}}$: formally, $\Fmf'$ is the restriction of $\Fmf$ to $W'= W^{\uparrow\avec{w}}\cup W_{\downarrow\avec{v}}$ and additionally, $\avec{w}'R'_a\avec{v}$; see Fig.~\ref{fig:drop:interval}. 
We set $\funcand'(\avec{u}) = \funcand(\avec{u})$ for all $\avec{u}\in W'$.
%
%
\begin{figure}[t]
\centering%
\begin{tikzpicture}[>=latex,
nd/.style={draw,circle,thick,inner sep=0pt,minimum size=2mm}
]
\fill[pattern=north west lines, pattern color=black!30] (1,0) -- ++(3,1.2) -- ++(0,-0.8) -- (3,0) -- ++(1,-0.4) -- ++(0,-0.8) -- (1,0);
\node[nd,fill=gray!50,label=left:{$\avec{w}_0$}] (w0) at (-1.5,0) {}; 
\node[nd,fill=gray!50,label=above:{$\avec{w}'$}] (w') at (0,0) {}; 
\node[nd,fill=white,label=above:{$\avec{w}$}] (w) at (1,0) {};
\node[nd,fill=gray!50,label=above:{$\avec{v}$}] (v) at (3,0) {};
%
%
\draw (w0) -- ++(5.5, 2.2);
\draw (w0) -- ++(5.5, -2.2);
\draw[dashed] (w) -- ++(3,1.2);
\draw[dashed] (w) -- ++(3,-1.2);
\draw (v) -- ++(1,0.4);
\draw (v) -- ++(1,-0.4);
\draw[->] (w') -- (w);
\draw[dashed, ->] (w) -- (v);
\draw[dashed, ->] (w0) -- (w');
\node at (1,-0.6) {\large $W^{\uparrow\avec{w}}$};
\node at (4,-0.1) {\large $W_{\downarrow\avec{v}}$};
\end{tikzpicture}
\caption{\textit{Drop Interval} operation on quasimodels in the proof of Lemma~\ref{lem:recboundnew}.}\label{fig:drop:interval}
\end{figure}
The new $\runs'$ is obtained  as follows.  Since the weak runs $\rho$ in $\runs$ are rooted, we have the following cases.
\begin{enumerate}[label=\textbf{\bfseries (dir$_{\theenumi}$)},leftmargin=*,series=run]
\item\label{interval:proper:v} All $\rho$ rooted in $W_{\downarrow\avec{v}}\setminus\{\avec{v}\}$ are preserved: $\runs'(\rho) = \runs(\rho)$.
\item\label{interval:discarded} All $\rho$ rooted in $W_{\downarrow\avec{w}}$ without reaching $\avec{v}$ ($\dom\rho \subseteq W_{\downarrow\avec{w}}\setminus W_{\downarrow\avec{v}}$) are discarded: $\runs'(\rho) = 0$.
\item\label{interval:outside:w} All $\rho$ not reaching  $\avec{w}$ ($\dom\rho\subseteq W^{\uparrow\avec{w}}$) are preserved: $\runs'(\rho) = \runs(\rho)$.
\item\label{interval:v} All $\rho$ rooted in $W_{\downarrow\avec{w}}$ and passing through $\avec{v}$ ($\dom\rho\subseteq W_{\downarrow\avec{w}}$ and $\rho\in\runs_{\avec{v}}$), in particular, all $\rho$ rooted in $\avec{v}$ and all $\rho$ rooted in $\avec{w}$ and passing through~$\avec{v}$ give rise to weak runs of the form $\rho_{\downarrow\avec{v}}$.
\item\label{interval:splice} All $\rho$ passing through $\avec{w}'$ and $\avec{w}$ ($\rho\in\runs_{\avec{w}'}\cap \runs_{\avec{w}}$) are spliced with runs $\rho'$ passing through $\avec{v}$ ($\rho'\in\runs_{\avec{v}}$) if $\rho$ and $\rho'$ have a matching type at $\avec{w}$ and~$\avec{v}$, respectively, and give rise to weak runs of the form $\rho^{\uparrow\avec{w}}\oplus\rho'_{\downarrow\avec{v}}$.
\end{enumerate}	
The multiplicities of weak runs in items~\ref{interval:v} and \ref{interval:splice} are chosen  in such a way that,
for each  $\rho' \in\runs_{\avec{v},\contp}$ with $\contp\in\funcand(\avec{v})$,
\begin{multline}\label{eq:interval:splice:2}
\sum\bigl\{\runs(\rho'') \mid \rho''\in \runs_{\avec{v},\contp},\ \rho''_{\downarrow\avec{v}} = \rho'_{\downarrow\avec{v}} \bigr\} = \\
\runs'(\rho'_{\downarrow\avec{v}}) +  \sum \bigl\{\runs'(\rho^{\uparrow\avec{w}}\oplus\rho'_{\downarrow{v}})\mid \rho\in\runs_{\avec{w}'}\cap\runs_{\avec{w},\contp}\bigr\},
\end{multline}
and, for each $\rho\in\runs_{\avec{w}'}\cap\runs_{\avec{w},\contp}$,
\begin{equation}\label{eq:interval:splice:1}
\sum \bigl\{\runs(\rho_0) \mid \rho_0\in \runs_{\avec{w},\contp},\ \rho_0^{\uparrow\avec{w}} = \rho^{\uparrow\avec{w}}\bigr\} = \sum \bigl\{\runs'(\rho^{\uparrow\avec{w}}\oplus\rho'_{\downarrow\avec{v}})\mid \rho' \in \runs_{\avec{v},\contp} \bigr\}.
\end{equation}%
We now define $\prfun[']{\avec{u}, \contp}$, for any $\avec{u}\in W'$ and $\contp\in\funcand'(\avec{u})$.
\begin{itemize}
\item If~$\avec{u}\in W_{\downarrow\avec{v}}$, then there are two cases to consider. 
If $\prfun{\avec{u},\contp}$ is rooted in $W_{\downarrow\avec{v}}\setminus\{\avec{v}\}$, then, by~\ref{interval:proper:v}, we keep $\prfun{\avec{u},\contp}$.
Otherwise, by~\eqref{eq:interval:splice:2}, there is $\rho\in\runs'$ with $\runs'(\rho) > 0$ that coincides with~$\prfun{\avec{u},\contp}$ on $W_{\downarrow\avec{v}}$; we take $\prfun[']{\avec{u},\contp} = \rho$.  It is clearly $a$-saturated at $\avec{u}$.
\item Otherwise, $\avec{u}\in W^{\uparrow\avec{w}}$, and we have two cases.  
If $\dom\prfun{\avec{u},\contp}\subseteq W^{\uparrow\avec{w}}$, then, by~\ref{interval:outside:w}, we keep $\prfun{\avec{u},\contp}$.  
Otherwise, $\prfun{\avec{u},\contp}\in\runs_{\avec{w}'}\cap
\runs_{\avec{w}}$, and so we take as $\prfun[']{\avec{u},\contp}$ a splice
$\prfun{\avec{u},\contp}^{\uparrow\avec{w}}\oplus\rho'_{\downarrow\avec{v}}$ with positive multiplicity,
which exists by~\eqref{eq:interval:splice:1}; it agrees with $\prfun{\avec{u},\contp}$ on $W^{\uparrow\avec{w}}$,
and so is $a$-saturated at $\avec{u}\neq\avec{w}'$, and
$\avec{v}$ is used as witness at~$\avec{w}'$. 
\end{itemize}
Thus, $\Qmf'$ is a weak quasimodel for~$\varphi$.  This completes the proof for $\QGuardmonMLplusc$.

\bigskip

We now consider   $\CtwomonMLplusc$. Let $g\colon \mathbb{N}\to\mathbb{N}$ be defined by taking $g(n) = g_0(2^{5n})$, for $n\in\mathbb{N}$, where $g_0$ is given by Lemma~\ref{lem:squeeze} ($g_0$ is monotone and  $g_0(n) \geq n$). Let~$\varphi$ be a $\CtwomonMLplusc$ -sentence. The $(\Leftarrow)$ direction is again by Lemma~\ref{lem:equivweaknotweak}~(b).
	
\smallskip	
	
$(\Rightarrow)$	 The construction is similar to the $\QGuardmonMLplusc$ case except that we need to deal with  infinite multiplicity, $\aleph_0$. We again exhaustively apply two rules,  \textit{Drop Interval}  and  \textit{Small Quasistate}, until the conditions of the lemma are satisfied.  The \textit{Drop Interval} rule is exactly the same as for $\QGuardmonMLplusc$. The  \textit{Small Quasistate} rule, however, needs to take account of infinite multiplicity. We introduce the set of infinite types at $\avec{w}\in W$:
\begin{align*}
\inftypes(\avec{w}) = \bigl\{ \contp \text{ a type for } \varphi \bigm| \funcand(\avec{w},\contp) = \aleph_0 \bigr\}.
\end{align*}
The rule also requires an additional normalisation step, which ensures that, at each~$\avec{w}$, for all weak runs except a small number of prototypes, infinite multiplicity of $\rho(\avec{w})$ implies infinite multiplicity of $\rho(\avec{v})$, for all $\avec{v}\in W_{\downarrow\avec{w}}$. In order to formalise this condition, we define an \emph{indexed prototype function} $\iprset$ that maps each $\avec{w}\in W$ and $\contp\in\funcand(\avec{w})$ to an indexed weak run $\iprfun{\avec{w},\contp} = (\prfun{\avec{w},\contp}, \ell)$, for some $0 \leq \ell <\runs(\prfun{\avec{w},\contp})$. Then we denote the set of indexed prototype weak runs for $\avec{w}$ and all its $R_*$-predecessors by
\begin{equation*}
\pranc{\avec{w}} = \bigl\{ \iprfun{\avec{v},\contp} \bigm| \contp\in\funcand(\avec{v}),  \text{ for some } \avec{v}\in W \text{ with } \avec{w}\in W_{\downarrow\avec{v}} \bigr\}.
\end{equation*}
With this notation at hand, we show that we can assume that
 \begin{enumerate}[label=$(\dagger_{\infty})$,leftmargin=*,series=run]
 \item\label{dagger:norm} 
 for all $(\rho,\ell)\in \hat{\runs}_{\avec{w}}\setminus \pranc{\avec{w}}$, if $\rho(\avec{w})\in \inftypes(\avec{w})$, then $\rho(\avec{v})\in\inftypes(\avec{v})$,  for all $\avec{v}\in W_{\downarrow \avec{w}}$.
\end{enumerate} 
	We prove~\ref{dagger:norm} by induction on the depth of $\avec{w}\in W$, updating $\runs$ and $\prset$ at each step. Let $\avec{w}\in W$ and assume~\ref{dagger:norm} holds for $\smash{\hat{\runs}}$ and $\iprset{}$ at all 
	$\avec{w}'\in W$ with $d(\avec{w}') < d(\avec{w})$. We modify~$\runs$ and $\prset$ so that~\ref{dagger:norm} also holds for~$\avec{w}$ (under a possibly different indexing of the weak runs). Suppose $(\rho,\ell)\in\iruns_{\avec{w}}\setminus \pranc{\avec{w}}$ with $\rho(\avec{w})\in\inftypes(\avec{w})$ violates the claim.
It follows, in particular, that $|\runs_{\avec{w},\rho(\avec{w})}| = \aleph_0$. 
	Since $W$ is finite, 
	$\sum_{\avec{v}\in W} \{ \runs(\rho')\mid \rho'\in\runs_{\avec{v}},\ \rho'(\avec{v})\notin\inftypes(\avec{v}) \}<\aleph_{0}$.
	Thus, there is a weak run $\rho'\in \runs_{\avec{w},\rho(\avec{w})}$ such that 
	$\rho'(\avec{v})\in\inftypes(\avec{v})$
	for all $\avec{v}\in \dom \rho'$ (in fact, the total multiplicity of such weak runs is $\aleph_0$). Now we obtain $\runs'$ from~$\runs$:  we define a weak run $\rho_* =  \rho^{\uparrow\avec{w}} \oplus \rho'_{\downarrow\avec{w}}$
\begin{figure}[t]%
\centering%
\begin{tikzpicture}[>=latex, xscale=0.9,
nd/.style={draw,circle,thick,inner sep=0pt,minimum size=2.2mm},
wn/.style={rectangle,rounded corners=2mm,draw,minimum width=11mm,minimum height=15mm,thin}
]
\begin{scope}
\node[wn, minimum width=6mm, fill=gray!10] (w0) at (-1,-0) {};	
\node[wn, minimum width=6mm, fill=gray!10,label=above:{$\avec{w}$}] (w0) at (-0,-0) {};	
\node[wn, minimum width=6mm, fill=gray!10] (w0) at (1,-0) {};	
\node[wn, minimum width=6mm, fill=gray!10] (w0) at (2,-0) {};	
\node[wn, minimum width=6mm, fill=gray!10,label=above:{$\avec{v}_0$}] (w0) at (3,-0) {};	
\node[nd,fill=white] (w) at (0,0) {};
\node[nd,fill=gray] (v) at (3,0) {};
\node at (-2,0) {$\rho$};
\node at (-2,-0.5) {$\rho'$};
\draw[ultra thick] (-1.5,0) -- (w) -- (v) -- +(0.5,0);
\draw[thick] (-1.5,-0.5) -- ++(1,0) -- (w) -- ++(0.5, 0.5) -- +(3,0);
\node[nd,fill=white] at (-1,0) {};
\node[nd,fill=white] at (-1,-0.5) {};
\node[nd,fill=white] at (1,0.5) {};
\node[nd,fill=white] at (2,0.5) {};
\node[nd,fill=white] at (3,0.5) {};
\node[nd,fill=white] at (1,0) {};
\node[nd,fill=white] at (2,0) {};
\end{scope}
\begin{scope}[xshift=66mm]
\node[wn, minimum width=6mm, fill=gray!10] (w0) at (-1,-0) {};	
\node[wn, minimum width=6mm, fill=gray!10,label=above:{$\avec{w}$}] (w0) at (-0,-0) {};	
\node[wn, minimum width=6mm, fill=gray!10,label=above:{$\avec{u}$}] (w0) at (1,-0) {};	
\node[wn, minimum width=6mm, fill=gray!10] (w0) at (2,-0) {};	
\node[wn, minimum width=6mm, fill=gray!10,label=above:{$\avec{v}_0$}] (w0) at (3,-0) {};	
\draw[ultra thick] (-1.5,0) -- (-0.08,0) -- ++(0.55,0.55) -- ++(3.03,0);
\draw[ultra thick] (1,0) -- ++(2.5,0);
\draw[thick] (-1.5,-0.5) -- ++(1.05,0) -- ++(0.95, 0.95) -- +(3,0);
\node[nd,fill=white] (w) at (0,0) {};
\node[nd,fill=gray] (v) at (3,0) {};
\node at (4,0) {$\rho_{\downarrow\avec{u}}$};
\node at (-2,0) {$\rho_*$};
\node at (-2,-0.5) {$\rho'$};
\draw[dashed] (-1.5,0) -- (w) -- (v) -- +(0.5,0);
\node[nd,fill=white] at (-1,0) {};
\node[nd,fill=white] at (-1,-0.5) {};
\node[nd,fill=white] at (1,0.5) {};
\node[nd,fill=white] at (2,0.5) {};
\node[nd,fill=white] at (3,0.5) {};
\node[nd,fill=white] at (1,0) {};
\node[nd,fill=white] at (2,0) {};
\end{scope}
\end{tikzpicture}
\caption{Replacing $\rho$ with $\rho_*$ and the $\rho_{\downarrow\avec{u}}$ to satisfy~\ref{dagger:norm}: types with infinite multiplicity are white circles, types with finite multiplicity are grey circles.}\label{fig:dagger:infty}
\end{figure}%
and replace $(\rho,\ell)\in\iruns$ with one fresh copy of $\rho_*$ and one fresh copy of $\rho_{\downarrow\avec{u}}$, for each $R_a$-successors $\avec{u}$ of~$\avec{w}$. Formally, $\runs'$ coincides with $\runs$ except that
\begin{align*}
\runs'(\rho)& =\runs(\rho) - 1,\\ 
\runs'(\rho_*) & =\runs(\rho_*)+1, \\
\runs'(\rho_{\downarrow\avec{u}}) &=\runs(\rho_{\downarrow\avec{u}})+1, \text{ for all $\avec{u}$ with $\avec{w} R_{a} \avec{u}$, $a\in A_{0}$}.
\end{align*}
Note again that we have to include $\runs(\rho_*)$ and $\runs(\rho_{\downarrow\avec{u}})$ as summands because $\runs$ can already contain these weak runs. Also, the types in $\rho'_{\downarrow\avec{w}}$ are effectively counted twice, in $\rho'$ as well as in $\rho_*$ (see Fig.~\ref{fig:dagger:infty}), but that has no effect on~\ref{run:exists} as they have infinite multiplicity. 
We also update $\prset$: if $\iprfun{\avec{v},\contp}=(\rho,\ell)$ for some~$\avec{v}$ with
$\avec{w}\notin W_{\downarrow\avec{v}}$, then we set $\iprfun[']{\avec{v},\contp}$ to the fresh copy of $\rho_{\downarrow\avec{u}}$, if
$\avec{v}\in W_{\downarrow\avec{u}}$ for some $\avec{u}$ with $\avec{w}R_a\avec{u}$, $a\in A_0$, and to
the fresh copy of $\rho^\ast$, otherwise; all other prototypes, in particular
$\pranc{\avec{w}}$, are preserved.
We repeat this construction until no indexed weak run violating~\ref{dagger:norm} at $\avec{w}$ is left. This concludes the induction step and the proof of the claim. Note that normalisation preserves Items~1--3 of the lemma.

\smallskip

	\medskip
	\noindent
	\emph{Small Quasistate}: We assume that $\runs$ and $\prset{}$ satisfy~\ref{dagger:norm} for all $\avec{w}\in W$ (for some indexing of the weak runs). Consider $\avec{w}\in W$ violating Item~1. If $\avec{w} = \avec{w}_0$, then we define a quasistate candidate $\avec{n}$ for $\varphi$ by setting
\begin{equation*}	
\avec{n}(t) = \begin{cases}
\aleph_{0}, &\text{if } \contp \in \inftypes(\avec{w}),\\
0, &  \text{otherwise.}
\end{cases}
\end{equation*}
By Lemma~\ref{lem:squeeze}, there exists a quasistate $\avec{y}$
	for $\varphi$ with 
\begin{equation}\label{eq:exp:quasistate:f}
\avec{n} \leq \avec{y} \leq \funcand(\avec{w}) \quad\text{ and }\quad |\avec{y}|_f\leq g_0(|\avec{n}|_f + |\varphi|).
\end{equation}
We define $\Qmf'=(\Fmf,\funcand',\runs',\prset{}')$ as in the $\QGuardmonMLplusc$ case, except that~\ref{small-quasistate:down} in the construction of $\Smf_{\downarrow\avec{w}}$ now reflects the $|\cdot|_f$ norm:
	 \begin{enumerate}[label=\textbf{\bfseries (sqr$^f_2$)},leftmargin=*,series=run]
	 	\item\label{small-quasistate:down:f} $\sum\bigl\{\Smf_{\downarrow\avec{w}}(\rho')\bigm| \rho'\in\runs_{\avec{w},\contp} \bigr\}= \funcand'(\avec{w},\contp)$, for all $\contp\notin\inftypes(\avec{w})$, and\\[4pt] $\Smf_{\downarrow\avec{w}}(\rho) = \runs(\rho)$, for all $\rho\in\runs_{\avec{w},\contp}$ and 
		$\contp\in\inftypes(\avec{w})$.
	 \end{enumerate}

The construction to cover non-root quasistates $\funcand(\avec{w})$ is modified accordingly. Take $\avec{w}'\in W$ with $\avec{w}' R_a \avec{w}$ and assume $|\funcand(\avec{w}')|_f\leq g^{d(\avec{w}')+1}(|\varphi|)$. 
We define a quasistate candidate $\avec{n}$ for $\varphi$ by taking
\begin{equation*}
\avec{n}(\contp) = \begin{cases}
\aleph_0, & \contp \in \inftypes(\avec{w}),\\
\sum\bigl\{\runs(\rho) \bigm| \rho \in \runs_{\avec{w}'}\cap\runs_{\avec{w},\contp}\bigr\}, &  \text{otherwise};
\end{cases}	
\end{equation*}
informally, $\avec{n}(\contp)$ keeps the number of runs from the predecessor $\avec{w}'$ of $\avec{w}$ if the type $\contp$ has finite multiplicity in $\funcand(\avec{w})$, but requires infinite multiplicity otherwise.
     By~\ref{dagger:norm}, for every $(\rho,\ell)\in \iruns_{\avec{w}'}$ with \revised{R1.8}{$\rho(\avec{w})\notin\inftypes(\avec{w})$, we have  $\rho(\avec{w}')\notin\inftypes(\avec{w}')$} or
     $(\rho,\ell) \in \pranc{\avec{w}'}$ (or both).
    Observe that $|\pranc{\avec{w}'}|\leq (d(\avec{w}') + 1) \cdot 2^{s(\varphi)}$.
 Hence,
 \begin{multline*}
 |\avec{n}|_{f} =\sum\bigl\{\runs(\rho)\bigm| \rho\in\runs_{\avec{w}'},\ \rho(\avec{w})\notin\inftypes(\avec{w})\bigr\}
\ \ \leq \ \ 
|\funcand(\avec{w}')|_{f}+|\pranc{\avec{w}'}| \\ 
\ \ \leq \ \ g^{d(\avec{w}')+1}(|\varphi|)+(d(\avec{w}') + 1) \cdot 2^{s(\varphi)} \ \ = \ \ g^{d(\avec{w})}(|\varphi|)+d(\avec{w}) \cdot 2^{s(\varphi)}.
 \end{multline*}
By Lemma~\ref{lem:squeeze}, there exists a quasistate $\avec{y}$ for $\varphi$ satisfying~\eqref{eq:exp:quasistate:f}. The rest of  the definition of $\Qmf'=(\Fmf,\funcand',\runs',\prset{}')$ is as in the $\QGuardmonMLplusc$ case, except that~\ref{small-quasistate:up} in the construction of $\Smf$ now reflects the $|\cdot|_f$ norm:
	 \begin{enumerate}[label=\textbf{\bfseries (sqr$^f_4$)},leftmargin=*,resume]
	 	\item\label{small-quasistate:up:f} $\sum\bigl\{\Smf(\rho')\bigm| \rho'\in\runs_{\avec{w},\contp} \bigr\}= \funcand'(\avec{w},\contp)$, for all $\contp\notin\inftypes(\avec{w})$,  and\\[4pt] $\Smf(\rho) = \runs(\rho)$, for all $\rho\in\runs_{\avec{w},\contp}$ and $\contp\in\inftypes(\avec{w})$.
\end{enumerate}
	It remains to observe that~\eqref{eq:exp:quasistate:f} implies 
\begin{multline*}	
	 |\funcand'(\avec{w})|_{f} \ \ \leq \ \ g_0(g^{d(\avec{w})}(|\varphi|) + d(\avec{w}) \cdot 2^{s(\varphi)} + |\varphi|)\\ \ \ \leq\ \  
        g_0(2^{5g^{d(\avec{w})}(|\varphi|)}) \ \ = \ \ 
	g^{d(\avec{w})+1}(|\varphi|). 
\end{multline*}
This completes the proof of Lemma~\ref{lem:recboundnew} for $\CtwomonMLplusc$.
\end{proof}

\begin{theorem}\label{thm:exp}
    For expanding-domain models, $\Kfn$-validity in both $\CtwomonMLplusc$ and $\QGuardmonMLplusc$ are decidable.
\end{theorem}
\begin{proof}
	We give the proof for $\CtwomonMLplusc$, the proof for $\QGuardmonMLplusc$ is similar. 
	Let $g$ be the recursive function from Lemma~\ref{lem:recboundnew} and assume a $\CtwomonMLplusc$-sentence $\varphi$ is given.
	Then $\varphi$ is satisfiable in a $\Kfn$ frame iff there is weak quasimodel satisfying~$\varphi$ and Items~1--3 of Lemma~\ref{lem:recboundnew}. The existence of such a weak quasimodel $\Qmf = (\Fmf,\funcand,\runs,\prset{})$ is decidable: by Theorem~\ref{thm:bound} and Items~1 and 2 of Lemma~\ref{lem:recboundnew}, the length $n$ of paths $\avec{w}_{0}R_{a_{0}}\cdots R_{a_{n - 1}}\avec{w}_{n}$ in $\Fmf$ is recursively bounded in the size~$|\varphi|$ of $\varphi$. Hence, by Item~3 of Lemma~\ref{lem:recboundnew}, the number of worlds of $\Fmf$ is recursively bounded in $|\varphi|$. Then also $|\funcand(\avec{w})|_{f}$ and the number of distinct $\rho\in \runs$ are recursively bounded.
\end{proof}

It is worth pointing out the size of weak quasimodels constructed in Lemma~\ref{lem:recboundnew} cannot be significantly improved because, as we shall see below (Theorem~\ref{thm:temp} and Lemma~\ref{lem:restemp}), the validity problem is Ackermann-hard even for the one-variable fragment.


%% file: 4.4_pspace.tex

\subsection{One-Variable Fragment}\label{sec:pspace}
We consider the one-variable fragment $\QoneMLc$ and show that, for expanding domains, $\Kn$-validity in this fragment is \PSpace-complete; recall that 
the problem for constant domains is co\NExpTime-complete (see Theorem~\ref{thm:complexity:onevar}).
We begin by proving a variant of Lemma~\ref{lem:squeeze} 
for $\QoneMLc$.
 \begin{lemma}\label{lem:squeeze:poly:fmp}
	For each $\QoneMLc$-sentence $\varphi$, each quasistate $\avec{m}$ and each finite quasistate candidate $\avec{n}\leq \avec{m}$ for $\varphi$, there exists a finite quasistate $\avec{m}'$ for $\varphi$ such that $\avec{n}\leq \avec{m}'\leq \avec{m}$ and
	$|\avec{m}'| \leq |\avec{n}|+|\varphi|$.
\end{lemma}
\begin{proof}
Let $B$ be the set of types in $\avec{m}$ that contains $\contp_c$ with $x = c \in \contp_c$, for each  $c\in\Ind(\varphi)$, and $\contp_\psi$ with $\psi\in\contp_{\psi}$, for each $\exists y\,\psi\in\sub[x]{\varphi}$  such that every (some) type $\contp\in\avec{m}$ contains $\exists y\,\psi$.
It is clear that $|B|\leq |\varphi|$.
	We define a quasistate candidate $\avec{m}'$ for~$\varphi$ by taking
	\begin{equation*}
		\avec{m}'(\contp) = \begin{cases} \min(\avec{m}(\contp),\avec{n}(\contp)+1), &\mbox{if } \contp\in B,\\
			\avec{n}(\contp), & \mbox{otherwise}.
			\end{cases}
	\end{equation*}
	It follows that $\avec{n}\leq \avec{m}'\leq \avec{m}$ and $|\avec{m}'|\leq |\avec{n}| + |B| \leq |\avec{n}|+|\varphi|$, and so it remains to show that $\avec{m}'$ is realisable.

	To this end, let $\Bmf$ be the first-order interpretation with domain $\avec{m'}$ and such that $c^\Bmf=\contp_c$, where $\contp_c\in B$ is the (unique) type with $(x=c)\in \contp_c$
	and $P^\Bmf = \{\contp\mid P(x) \in \overline{\contp}\}$, for predicate names~$P$. For each  $\contp\in \avec{m'}$, it then follows from a routine induction on $\psi$ that $\Bmf\models \overline{\psi}[\contp]$ iff $\psi\in \contp$, for all $\psi\in \sub[x]{\varphi}$. 
	In particular, for each $\contp\in \avec{m}'$, we have $\Bmf\models \overline{\contp}[\contp]$, whence $\avec{m}'(\contp) = |\{\contp\in \Bmf \mid \Bmf \models \overline{\contp}[\contp]\}|$, which is to say that $\avec{m}'$ is realisable.
\end{proof}

We next show a variant of Lemma~\ref{lem:recboundnew} for $\QoneMLc$.
\begin{lemma}\label{lem:expanding:poly:fmp}
	For expanding domains, for each  $\QoneMLc$-sentence $\varphi$, there exists a quasimodel satisfying~$\varphi$ based on a $\Tree{d}$  frame iff there is a weak quasimodel $\Qmf=(\Fmf,\funcand,\Rmf,\prset)$ satisfying $\varphi$ based on a \textup{(}finite\textup{)} $\Tree{d}$ frame $\Fmf$ such that, for all worlds $\avec{w}\in W$,
	\begin{equation}
		\label{eq:squeeze:poly:fmp}
		|\funcand(\avec{w})|\leq 1+(\md(\avec{w}) + 1)\cdot|\varphi|.
	\end{equation}
\end{lemma}
\begin{proof} %
	The ($\Leftarrow$) direction is by Lemma~\ref{lemma:k:finite-approx}. 
	For the ($\Rightarrow$) direction, by Lemma~\ref{lem:weakprequasi:k}, we obtain a weak quasimodel $\Qmf_0=(\Fmf,\funcand_0,\Rmf_0,\prset_0)$ satisfying $\varphi$ based on a finite $\Tree{d}$ frame $\Fmf$ with root $\avec{w}_0$. 
    We may assume that $\varphi\in t_0$ for some $t_0\in \funcand_0(\avec{w}_0)$. 
	We now follow a similar approach to that taken in the proof of Lemma~\ref{lem:recboundnew}, by exhaustively applying transformation rules to $\Qmf_0$ to reduce the size of each quasistate: we construct a sequence $\Qmf_0,\dots,\Qmf_N$ of weak quasimodels satisfying $\varphi$ such that $\Qmf_N$ fully meets the size restrictions in the claim. 
	The difference here is that we only require the  \textit{Small Quasistate} rule because the frame remains the same and is in fact of depth bounded by $d$.  However, here we utilise Lemma~\ref{lem:squeeze:poly:fmp} in place of Lemma~\ref{lem:squeeze}, which allows us to bound the quasistate size by a polynomial in $|\varphi|$. 
	
\smallskip
\noindent\textit{Small Quasistate:} Consider $\avec{w}\in W$ violating the bound \eqref{eq:squeeze:poly:fmp}. 
If $\avec{w}=\avec{w}_0$, then we define a quasistate candidate $\avec{n}$ for $\varphi$ by setting $\avec{n}(\contp_0) = 1$ and $\avec{n}(\contp) =0$ for all $\contp\ne\contp_0$.
By Lemma~\ref{lem:squeeze:poly:fmp}, there exists a quasistate $\avec{y}$ for $\varphi$ with $\avec{n}\le \avec{y}\le\funcand(\avec{w})$ and $|\avec{y}|\leq 1+|\varphi|$. 
Replace the quasistate at $\avec{w}$ with $\avec{y}$ by setting $\funcand'(\avec{w})=\avec{y}$ and $\funcand'(\avec{u})=\funcand(\avec{u})$ for all $\avec{u}\ne\avec{w}$. 
We construct $\Rmf'$ and $\prset{}'$ exactly as in the {\it Small Quasistate} construction in the proof of Lemma~\ref{lem:recboundnew}, by first identifying some multiset of runs $\Smf_{\downarrow \avec{w}}$ satisfying \ref{small-quasistate:down-pr}--\ref{small-quasistate:down}, with $\Smf_{\downarrow \avec{w}}(\rho)\leq \Rmf(\rho)$, for each $\rho\in\Rmf_{\avec{w}}$, and then compensating for the removal of weak runs $\rho\in \Rmf_{\avec{w}}$ by adding back to $\Smf_{\downarrow \avec{w}}$ their restrictions to $W_{\downarrow \avec{u}}$, for each $a$-successor $\avec{u}$ of $\avec{w}$.

If $\avec{w}\not=\avec{w}_0$ then let $\avec{w}'\in W$ be such that $\avec{w}'R_a \avec{w}$ and assume that $|\funcand(\avec{w'})|\leq 1+(\md(\avec{w}') + 1)\cdot|\varphi|$. 
We again, first define a quasistate candidate $\avec{n}$ for $\varphi$, as in the proof of Lemma~\ref{lem:recboundnew}, by taking
\begin{equation*}
    \avec{n}(\contp) = \sum\bigl\{\runs(\rho) \mid \rho \in \runs_{\avec{w}'}\cap\runs_{\avec{w},\contp}\bigr\}, \text{ for all  types $\contp$ for $\varphi$},
\end{equation*}
so that $|\avec{n}|\leq |\funcand(\avec{w}')|\leq 1+(\md(\avec{w}') + 1)\cdot|\varphi| = 1+\md(\avec{w})\cdot|\varphi|$. 
Hence, by Lemma~\ref{lem:squeeze:poly:fmp}, there is a quasistate $\avec{y}$ for $\varphi$ with $\avec{n}\le \avec{y}\le\funcand(\avec{w})$ satisfying~\eqref{eq:squeeze:poly:fmp}. 
Again, we construct $\Rmf'$ and $\prset{}'$ as in the {\it Small Quasistate} construction in the proof of Lemma~\ref{lem:recboundnew}, by identifying, in addition to the multiset $\Smf_{\downarrow\avec{w}}$, an additional multiset $\Smf$ satisfying~\ref{small-quasistate:up:pr}--\ref{small-quasistate:up}, with $\Smf(\rho)\leq \Rmf(\rho)$, for each $\rho\in\Rmf_{\avec{w}}$. 
For each type $\contp$, we fix a bijection~$f_\contp$ between the indexed runs of $\hat{\Smf}_{\downarrow\avec{w}}$ and of $\hat\Smf$ that map $\avec{w}$ to $\contp$, and define $\runs'$ and~$\prset'$ as in the proof of Lemma~\ref{lem:recboundnew}.

By repeated application of this modified {\it Small Quasistate} rule from the root to the leaves, we construct the required weak quasimodel.
\end{proof}

In fact,
the construction in Lemma~\ref{lemma:k:finite-approx} (which preserves the bound on the size of quasistates at a given depth) implies the following stronger result:
\begin{lemma}\label{lem:expanding:poly:fmp:2}
	For expanding domains, for each  $\QoneMLc$-sentence $\varphi$, if there is a quasimodel satisfying~$\varphi$ based on a $\Tree{d}$  frame, then there is a  quasimodel $\Qmf=(\Fmf,\funcand,\Rmf)$ satisfying $\varphi$ based on a \textup{(}finite\textup{)} $\Tree{d}$ frame $\Fmf$ such that, for all $\avec{w}\in W$,
	\begin{equation*}
		|\funcand(\avec{w})|\leq 1+ (\md(\avec{w}) + 1)\cdot|\varphi|.
	\end{equation*}
\end{lemma}

We are now in a position to show the \PSpace complexity bound.
We generalise the standard $\PSpace$ upper bound proof 
for propositional $\Kn$~(see, e.g.,~\cite{Spaan93,DBLP:books/cu/BlackburnRV01})
to $\QoneMLc$ using the lemma above.

\medskip

In what follows it will be convenient to associate each tuple $\avec{\contp}=(\contp_0,\dots, \contp_\ell)$ of types  with a multiset $\avec{n}_{\avec{t}}(\contp)=|\{0 \leq i\leq \ell : \contp_i=\contp\}|$, by discarding the inherent ordering.

\begin{theorem}\label{thm:kn:exp}
	For expanding-domain models, $\Kn$-validity in $\QoneMLc$ is \PSpace-complete. 
\end{theorem}
    \begin{proof} 
    	The \PSpace lower bound follows from \PSpace-hardness of the underlying (propositional) modal logic $\Kn$. 
    	For the upper bound, we define a recursive function $\QKsat(k,\avec{\contp})$ that takes as input an integer $k\in \mathbb{N}$ and a tuple $\avec{\contp}=(\contp_0,\dots, \contp_\ell)$ of types for $\varphi$ of size $\ell<N$ and returns \true iff the following conditions are met:
    	\begin{enumerate}[label={\bf (tab\arabic*)}]\itemindent=1em
    		\item \label{tab:quasistate} the multiset $\avec{n}_{\avec{\contp}}$ associated with $\avec{\contp}$ is a quasistate for $\varphi$;
  
    		\item \label{tab:successor} 
		for each type $\contp_i$ in $\avec{\contp}$ and $\D_a\psi\in \contp_i$, there exists a tuple $\avec{\contp}'=(\contp_0',\dots, \contp_m')$ of types for $\varphi$, with $\ell\leq m<N$, 
    		such that 
    		\begin{itemize}
    			\item[(i)] $\psi\in \contp_i'$,
    			\item[(ii)] $\contp_j \to_a \contp_j'$ for all $0 \leq j\leq \ell$,
    			\item[(iii)] $k > 0$ and $\QKsat(k-1,\avec{\contp}')$ returns \true.
    		\end{itemize}
    	\end{enumerate}	
    	
    	Note that $\QKsat(k,\avec{\contp})$ is inherently non-deterministic as it must explore all possible choices for $\avec{\contp}'$ in condition \ref{tab:successor}. However, since the size of~$\avec{\contp}'$ is bounded by $N$, by appealing to Savitch's theorem, we require only that $\QKsat$ uses at most a polynomial amount of space.  
    	We note that \ref{tab:quasistate} can be checked in polynomial space since the one-variable fragment $\FOO$ is \NPclass-complete. 
	Checking \ref{tab:successor}(i)--(ii) can be done `in-place,' while the depth of the recursion required for \ref{tab:successor}(iii) is bounded by $k$, and so it too can be performed in polynomial space. 
    	The fixed parameter $N$ places a bound on the maximum length of the tuple of types that can ever appear as arguments to $\QKsat$, which, by Lemma~\ref{lem:expanding:poly:fmp:2}, will be chosen to be polynomial in~$|\varphi|$.
    	
    	\medskip

    	For soundness and correctness, we claim that, for all $k\in \mathbb{N}$ and tuples $\avec{\contp}=(\contp_0,\dots, \contp_\ell)$, the call $\QKsat(k,\avec{\contp})$ returns \true if and only if there is quasimodel $\Qmf=(\Fmf,\funcand,\Rmf)$ for $\varphi$ (not necessarily satisfying~$\varphi$) such that
    	\begin{enumerate}[label={(QK\arabic*)}]\itemindent=1em
    		\item \label{IH:1} $\Fmf=(W,\{R_a\}_{a\in A})$ is a $\Tree{k}$ frame rooted in $\avec{w}_0$,
    		\item \label{IH:2} $\funcand(\avec{w}_0)=\avec{n}_{\avec{t}}$,
    		\item \label{IH:3} $|\funcand(\avec{w})|\leq N$ for all $\avec{w}\in W$.
    	\end{enumerate}
    	We prove this by induction on $k$, so let $k\geq 0$ be fixed and suppose that the claim holds for all $m<k$.
    		
    		$(\Rightarrow)$ 
    		Suppose that $\QKsat(k,\avec{\contp})$ returns \true, for $\avec{\contp}=(\contp_0,\dots, \contp_\ell)$. 
    		%
		Let $S$ be the set of all pairs $(i,\D_a\psi)$ with $\D_a\psi\in \contp_i$. 
	For each \mbox{$\sigma=(i,\D_a\psi)\in S$},  by \ref{tab:successor}, there is some $\avec{\contp}^\sigma=(\contp_0^\sigma,\dots, \contp_{m}^\sigma)$, for $\ell\leq m<N$, such that (i)~\mbox{$\psi\in \contp_i^\sigma$}, (ii)~$\contp_j\to_a \contp_j^\sigma$, for all $j\leq \ell$, and (iii)~$k > 0$ and $\QKsat(k-1,\avec{\contp}^\sigma)$ returns \true.
    		By the induction hypothesis, there is some quasimodel $\Qmf^\sigma=(\Fmf^\sigma,\funcand^\sigma,\Rmf^\sigma)$ for~$\varphi$ based on a $\Tree{k-1}$ frame $\Fmf^\sigma$ rooted in $w^\sigma_0$ with $\funcand^\sigma(w^\sigma_0)=\avec{n}_{\avec{t}^\sigma}$ and $|\funcand^\sigma(\avec{w})|\leq N$ for all $\avec{w}\in W^\sigma$. 
    		From this (possibly empty) collection of quasimodels, we define a new quasimodel $\Qmf=(\Fmf,\funcand,\Rmf)$ by taking $\Fmf$ to be the disjoint union of the frames $\Fmf^\sigma$, for $\sigma\in S$, conjoined by introducing fresh root node $w_0$:
    		\begin{equation*}
    			W = \{w_0\} \cup \{w_0a\avec{w} \mid \avec{w} \in W^\sigma,\ \sigma=(i,\D_a\psi)\in S\}.
    		\end{equation*}
    		(Note that when $k=0$, it follows from \ref{tab:successor}(iii) that $S$ must be empty, and so $\Fmf$ is a $\Tree{0}$ frame comprising just the root node $w_0$.)
    		Over this new frame, we define $\funcand$ such that $\funcand(w_0) = \avec{n}_{\avec{\contp}}$, which is guaranteed to be a quasistate by \ref{tab:quasistate}, and $\funcand(w_0a\avec{w}) = \funcand^\sigma(\avec{w})$ for all~$\avec{w}\in W^\sigma$.
    		
    		By \textbf{(card)}, for each $0 \leq j\leq \ell$ and $\sigma\in S$, there is some distinct $\rho_j^\sigma\in \Rmf^\sigma$ such that $\rho_j^\sigma(w_0^\sigma)=\contp_j^\sigma$. 
    		Hence, we may define a new run $\rho_j$ by taking $\rho_j(w_0)=\contp_j$ and $\rho_j(w_0a\avec{w})=\rho_j^\sigma(\avec{w})$ for all $\avec{w}\in W^\sigma$.
    		We then take $\Rmf$ to be the multiset of all such $\rho_j$ for $0 \leq j\leq \ell$, together with all $\rho\in \Rmf^\sigma$ that are not the restriction of any $\rho_j$ to $W^\sigma$.
         		
    		It is clear from the construction that $\Qmf$ satisfies~\ref{IH:1}--\ref{IH:3}, and so it remains to show that $\Qmf$ satisfies conditions~\ref{rn:modal} and~\ref{rn:modal2}. 
    		For~\ref{rn:modal}, we note that each run is coherent across each of the subframes $\Fmf^\sigma$, since the restriction of each run to $W^\sigma$ is an original run of $\Qmf^\sigma$, which are coherent. 
    		For the root, suppose that $\psi\in \rho_j(w_0aw_0^\sigma)$ for some $0 \leq j\leq \ell$
		and $\sigma\in S$. By construction, $\rho_j(w_0aw_0^\sigma)=\rho_j^\sigma(w_0^\sigma)=\contp_j^\sigma$. Hence, by (ii),  $\D_a \psi\in \contp_j=\rho_j(w_0)$, as required. 
    		Similarly, for~\ref{rn:modal2}, we note that each run is saturated across each~$\Fmf^\sigma$. For the root, if $\D_a\psi\in \rho_i(w_0)=\contp_i$ then $(i,\D_a\psi)\in S$, and so, by~(i), \mbox{$\psi\in \contp_i^\sigma$}. Hence, by construction, $\psi\in \rho_i^\sigma(w_0^\sigma)$, which is to say that $\psi\in \rho_i(w_0aw_0^\sigma)$, as required.
    		
    		$(\Leftarrow)$ 
    		Conversely, suppose that $\Qmf=(\Fmf,\funcand,\Rmf)$ is a quasimodel for $\varphi$ satisfying \ref{IH:1}--\ref{IH:3}.
    		By \ref{IH:2}, $\avec{q}(\avec{w}_0)$ is a quasistate for $\varphi$, as required by \ref{tab:quasistate}. 
            By \textbf{(card)}, $\funcand(\avec{w},\contp)=|\Rmf_{\avec{w},\contp}|$, for every $\avec{w}\in W$ and type $\contp$ for $\varphi$. Hence, without loss of generality we may enumerate the runs of $\Rmf$ as $\rho_0,\rho_1,\dots$, such that $\rho_j(\avec{w}_0)=\contp_j$, for all $0 \leq j\leq \ell$.
    		For \ref{tab:successor},
    		let $i\leq \ell$ and $\D_a\psi\in \contp_i$, which is to say that $\D_a\psi\in\rho_i(\avec{w}_0)$.  
            Hence, by~\ref{rn:modal2}, there is some $\avec{v}\in W$ such that $\avec{w}_0R_a \avec{v}$ and $\psi\in \rho_i(\avec{v})$, 
    		which also implies that $k>0$. 
    		So let $\avec{\contp}'=(\rho_j(\avec{v}) : \rho_j\in \Rmf_{\avec{v}})$ be a tuple of types for~$\varphi$ of size $|\Rmf_{\avec{v}}|$, so that \ref{tab:successor}(i) holds trivially.
            For \ref{tab:successor}(ii), suppose that $\psi\in \rho_j(\avec{v})$ for $\contp_j\in \avec{\contp}$, which is to say $\rho_j\in \Rmf_{\avec{w}_0}$ and $\rho_j(\avec{w}_0)=\contp_j$. Hence, by~\ref{rn:modal}, $\D_a\psi\in \contp_j$, as required. 
    		For \ref{tab:successor}(iii), consider the restriction $\Qmf^v$ of $\Qmf$ to $W_{\downarrow v}$,
which clearly satisfies \ref{IH:1}--\ref{IH:3}: 
    		by the induction hypothesis, $\QKsat(k-1,\avec{\contp}')$ returns \true, as required.

    	\medskip

    	It then follows from Lemmas~\ref{lemma:quasimodel}, \ref{lem:k:tree-shaped} and~\ref{lem:expanding:poly:fmp:2} that $\varphi$ is $\Kn$-satisfiable iff there is a tuple $\avec{\contp}$ of types with $\varphi\in \contp_0$, such that $\QKsat(\md(\varphi),\avec{\contp})$ returns \true, where $N=1 + (\md(\varphi)+1)\cdot |\varphi|=O(|\varphi|^2)$, which can be decided in \PSpace. 
    	\end{proof}

%% file: 5_temporal.tex

\section{First-Order Temporal Logic}
We consider standard temporal logics with the temporal operators for `eventually' and `next' defined over the natural numbers and finite strict linear orders. The decidability and complexity of monodic fragments of temporal logics without \NRDC{} features has been investigated extensively~\cite{HodEtAl03,DBLP:journals/tocl/DegtyarevFK06,DBLP:journals/apal/Hodkinson06,KOURTIS2025115319,AMO:TOCL24}.
Here we discuss what  happens with \NRDC{} features. We exploit known negative results for the one-variable fragment by giving a polytime reduction to modal logics over the frame classes $\K_{\ast n}$ and $\Kfn$. This allows us to prove matching lower bounds for $\K_{\ast n}$ and $\Kfn$. Conversely, our decidability results for $\Kfn$ with expanding domains translate to decidability results for temporal logics over finite strict linear orders with expanding domains.

We introduce the relevant frame classes and notation for languages. By $\LTLd$ we denote 
the frame class containing single frame $(\mathbb{N},<)$ with $\mathbb{N}$ the natural numbers and $<$ its standard strict ordering. The class of finite strict orders can then be defined as $\LTLfd = \{ (\{0,\ldots,n\},<_{|\{0,\ldots,n\}}) \mid n\geq0\}$. In these frames we interpret the modal language with a single temporal operator, $\Diamond$ (`eventually'). For convenience, we denote the languages $\QoneML$, $\CtwomonML$, and $\QGuardmonMLb$ with a single modality interpreted over these frames by $\QoneLTLd$, $\CtwomonLTLd$, and $\QGuardmonLTLdb$, respectively. 

We also consider the classes of frames extended with the successor relation $S=\{\, (i,i+1) \mid i\in\mathbb{N}\,\}$ interpreting operator $\Next$ (`next'): we denote $\LTL = \{ (\mathbb{N},<,S)\}$  
and $\LTLf =\{ (\{0,\ldots,n\},<_{|\{0,\ldots,n\}},S_{|\{0,\ldots,n\}}) \mid n\geq 0\}$. The bimodal languages with $\Diamond$ and $\Next$ interpreted over these frame classes are denoted $\QoneLTL$, $\CtwomonLTL$, and $\QGuardmonLTLb$, respectively. 

The proof of Lemma~\ref{lem:redglobloc} shows that, for all temporal frame classes $\Cmc$ introduced above and temporal fragments $\Lmc$, global $\Cmc$-consequence in $\Lmc$ is polytime-reducible to $\Cmc$-validity in $\Lmc$. In what follows we therefore consider validity only.

The computational behaviour of $\QoneLTLd$ and $\QoneLTL$ is well understood. We use the facts that, by Theorem~\ref{thm:reductions}, we can always eliminate definite descriptions and partial designators and, by Theorem~\ref{th:diff}, we can replace constants by the ``elsewhere'' quantifier.
Then the following result follows from~\cite[Table~1]{HamKur15}, which is partly based on ideas first developed in~\cite{DBLP:conf/cade/KonevWZ05,DBLP:journals/apal/GabelaiaKWZ06}.
\begin{theorem}\label{thm:temp}
In $\QoneLTL$ and its fragment $\QoneLTLd$ with the $\Diamond$-operator only\textup{:}

\textup{(1)} for constant domains, $\LTL$-validity  is $\Pi^1_1$-complete, while $\LTLf$-validity  is undecidable and co-r.e.\textup{;}

\textup{(2)} for expanding domains, $\LTL$-validity is undecidable and r.e., while $\LTLf$-validity is decidable but Ackermann-hard.
\end{theorem}

We note that monodic fragments with NRDC features
beyond the one-variable fragment have not yet been considered in the temporal context.
The following result allows us to transfer decidability and complexity results between the temporal and modal domain. 
\begin{lemma}\label{lem:restemp}
For both constant and expanding domains, $\LTL$-validity in $\QoneLTL$, $\CtwomonLTL$ and $\QGuardmonLTLb$ are polytime-reducible to $\K_{\ast n}$-validity in $\QoneML$, $\CtwomonML$ and $\QGuardmonMLb$, respectively. 
This also holds
if
$\LTL$ and $\K_{\ast n}$ are replaced by $\LTLf$ and $\Kfn$, respectively.
\end{lemma} 
\begin{proof}
The proof of this reduction from logics of linear frames with transitive closure to logics of branching frames with transitive closure is not trivial but can be done by adapting in a straightforward way the reduction given in the proof of \cite[Theorem 6.24]{GabEtAl03} for product modal logics. 
\end{proof}

Using Lemma~\ref{lem:restemp} and Theorem~\ref{thm:temp}, we obtain all lower bounds for the frame classes $\K_{\ast n}$ and $\Kfn$ stated in Table~\ref{table:complexity}. Conversely, from Theorem~\ref{thm:exp} and Lemma~\ref{lem:restemp}, we obtain the following decidability result for temporal monodic fragments over strict finite orders and expanding domains.
\begin{theorem}\label{thm:monodictemp}
 For expanding-domain models, $\LTLf$-validity in $\CtwomonLTL$ and $\QGuardmonLTLb$ are decidable.
\end{theorem}


%% file: 6_conclusion.tex

\section{Discussion}
We have shown that monodic fragments of important first-order modal logics are decidable even if non-rigid constants, closed definite descriptions, or counting (i.e., \NRDC{} features) are present in the language. The proof method is illustrated for the guarded fragment with constants and equality, and for the two-variable fragment with counting. We conjecture that our technique can be extended to other decidable fragments of first-order logic,
such as variants of the guarded negation
fragment~\cite{DBLP:journals/jsyml/BaranyBC18},  fluted fragments~\cite{DBLP:journals/apal/Pratt-HartmannT22,DBLP:conf/csl/Pratt-HartmannT23}, or extensions of the two-variable fragment with semantically-constrained relations, e.g., transitive or equivalence relations~\cite{pratt2023fragments}. 

Our results can also be applied to many modal and temporal description logics (DLs).
These are monodic fragments of modal and temporal first-order logics, where the explicit quantification of variables is replaced by the implicit quantification in the DL constructs, and which have been investigated and applied extensively~\cite{LutEtAl08,BaaEtAl12,ArtEtAl14,baader2020metric}.
Very powerful positive results directly follow from what is shown in this article: for instance, decidability and upper complexity bounds for
monodic two-variable fragments with counting transfer  
to modal and temporal DLs based on $\mathcal{ALCQHIO}^{u}$ simply because the latter can be regarded as a fragment of $\mathsf{C}^2$. It would be of interest to explore whether our techniques can be applied to analyse temporal ontology-mediated query answering~\cite{DBLP:conf/time/ArtaleKKRWZ17} with \NRDC{} features or very expressive DLs with \NRDC{} features not yet considered in a modal or temporal context~\cite{DBLP:conf/dlog/GrauHKS06}.

An alternative approach to defining decidable fragments of first-order modal logic is to restrict the relative position of modal operators and quantifiers to certain patterns, called \emph{bundles} (such as 
$\exists x \Diamond$ and $\Diamond \forall x$). \emph{Bundled fragments} of first-order modal logic have been investigated extensively~\cite{DBLP:journals/iandc/LiuPRW23,DBLP:conf/mfcs/LiuPRW22}. However, to the best of our knowledge, so far this approach has only been applied to fragments without \NRDC{} features. Thus, exploring how bundled fragments behave with \NRDC{} features remains a topic for future research.

Finally,
it would be of interest to systematically consider the role of monodicity and \NRDC{} features within \emph{term modal logics}, a large family of languages in which modal operators are indexed by typically non-rigid names for agents~\cite{FitEtAl01, Koo08, CorOrl13}.
Even though decidability has already been extensively studied~\cite{Lib20, WanEtAl22, PadRam23}, we conjecture that our techniques can be applied
to variations of term modal logics as well.

\paragraph{Acknowledgements} We thank the anonymous reviewers for their thoughtful remarks and useful suggestions that helped us improve the paper.
